\documentclass[11pt,a4paper,oneside,english]{article}
\usepackage[T1]{fontenc}
\usepackage[utf8]{inputenc}
\usepackage[english]{babel}
\usepackage{latexsym}
\usepackage{midpage}
\usepackage{newlfont}
\usepackage{color}
\usepackage{array}
\usepackage{amsmath} 
\usepackage{amsfonts} 
\usepackage{amsthm} 
\usepackage{amssymb}
\usepackage{setspace}
\usepackage[hidelinks]{hyperref} 
\usepackage{url}
\usepackage[normalem]{ulem}
\usepackage{cleveref}
\usepackage{graphicx} 
\usepackage{caption}
\usepackage{subcaption} 
\usepackage{mathdots}
\usepackage{longtable}
\usepackage{booktabs}
\usepackage[margin=1in,left=1.04in,right=1.04in,includefoot]{geometry}
\usepackage{rotating}
\usepackage{listings}
\usepackage{arydshln}
\usepackage{tikz}
\usepackage{color}
\usepackage{bbm}
\usepackage{mathrsfs}

\definecolor{darkgreen}{rgb}{0, .5, 0}

\theoremstyle{plain} 
\newtheorem{thrm}{Theorem}[section] 
 
\newtheorem{lem}[thrm]{Lemma} 
\newtheorem{assumption}[thrm]{Assumption} 
\newtheorem{prop}[thrm]{Proposition} 
\theoremstyle{definition} 
\newtheorem{defn}{Definition}[section] 
\theoremstyle{remark} 
\newtheorem{oss}{Remark}[section]
 
\numberwithin{equation}{section}

\newcommand{\E}{\mathbb{E}} 
\newcommand{\p}{\mathbb{P}} 
\newcommand{\Q}{\mathbb{Q}} 
\newcommand{\FF}{\mathbb{F}} 
\newcommand{\R}{\mathbb{R}} 
\newcommand{\F}{\mathcal{F}} 

\newcommand{\cN}{\mathcal{N}}

\newcommand{\dd}{\mathrm d}

\usepackage{authblk}

\author[1]{Alessandro Gnoatto}
\author[2]{Silvia Lavagnini}
\author[1]{Athena Picarelli}
\affil[1]{\textsc{Department of Economics, University of Verona, 37129 Verona, Italy}}
\affil[2]{\textsc{Department of Data Science and Analytics, BI Norwegian Business School, 0848 Oslo, Norway}}

\title{Deep Quadratic Hedging\footnote{Accepted Manuscript. Final version available at: \url{https://doi.org/10.1287/moor.2023.0213}}}

\begin{document}
\maketitle
	
\begin{abstract}

	
We propose a novel computational procedure for quadratic hedging in high-dimensional incomplete markets, covering mean-variance hedging and local risk minimization. Starting from the observation that both quadratic approaches can be treated from the point of view of backward stochastic differential equations (BSDEs), we (recursively) apply a deep learning-based BSDE solver to compute the entire optimal hedging strategies paths. This allows to overcome the curse of dimensionality, extending the scope of applicability of quadratic hedging in high dimension. We test our approach with a classical Heston model and with a multi-asset and multi-factor generalization of it, showing that this leads to high levels of accuracy.
\end{abstract}

\paragraph*{Keywords} Deep hedging; Quadratic hedging; Deep BSDE solver; Mean-variance hedging; Local risk minimization; Multidimensional Heston model.

\section{Introduction}
The problem of hedging a contingent claim is typically faced by a financial institution who sold a derivative contract to a customer and is confronted with a future liability. The liability is the payoff of the derivative at the terminal time and is represented by a positive random variable $H$. To find the optimal hedging strategy that replicates the contingent claim $H$ is a central problem in financial mathematics and operations research. It corresponds to finding a representation for $H$ in the form of a stochastic integral of the trading strategy with respect to a semimartingale \cite{harrisonPliska81}.  When such a representation holds for any possible choice of the contingent claim $H$, the market is said to be complete. The standard example in this case is the Black-Scholes market model \cite{blackScholes73}.

Completeness, however, is easily lost as soon as we consider even slight generalizations of the Black-Scholes market. For example, if we consider random coefficients due to stochastic volatility or stochastic interest rates, or if we allow for discontinuous price paths in the underlying security. In the last decades, this has motivated a multitude of alternative approaches for hedging in incomplete markets. Here we consider quadratic hedging in a pure diffusive setting, both in the form of mean-variance hedging, as introduced in \cite{bouleauLamberton89}, \cite{duffieRichardson91}, \cite{schweizer94}, and in the form of local risk minimization, as proposed in \cite{fs91}, \cite{foellmerSondermann86},  \cite{Schweizer91} and \cite{schweizer94}. For a survey on the two approaches, featuring also a complete list of early references, we refer to \cite{Schweizer01}. 

Both mean-variance hedging and local risk minimization have been studied and extended in different directions in the last decades. For stochastic volatility models, which are the focus of the present paper, we mention \cite{bgp2000}  and \cite{dmkr1995} for mean-variance hedging, and \cite{hps1992} for local risk minimization.  For the class of affine stochastic volatility models, we also refer to \cite{cerny2008} and \cite{kp2010}, and for recent results on mean-variance hedging we mention \cite{CziCer22}. Finally, approaches based on the theory of stochastic optimal control are treated in \cite{jbmss2012} and \cite{lim2004} for the mean-variance hedging case. We anticipate that the link with stochastic control theory will be very important in our approach as we shall see in the sequel.

Concerning the computation of quadratic hedging strategies, we mention the works of \cite{hps2001} and \cite{heath2001numerical}: here mean-variance hedging and  local risk minimization are compared in the context of certain stochastic volatility models. In particular, both approaches are studied numerically in a Markovian framework by partial differential equation (PDE) techniques in the form of finite difference methods. However, it is known that such numerical methods for PDEs suffer from the curse of dimensionality. Hence it seems problematic to apply quadratic hedging to contingent claims that depend on a high number of risk factors. This is the problem that we address with the present paper: namely, we propose a computational approach which makes the implementation of quadratic hedging feasible even when the number of underlying risk factors is large. 

Our strategy is based on two observations. First of all, both quadratic hedging approaches can be treated from the point of view of backward stochastic differential equations (BSDEs). This link is particularly evident for the mean-variance hedging, due to the very formulation of the problem as a stochastic control problem of linear-quadratic type. Here the solution involves a system of two BSDEs that need to be solved sequentially. On the other hand, the link with BSDEs is not immediately clear in the original definition of local risk minimizing strategy. However, it is known that finding local risk minimizing strategies is equivalent to finding the so-called  F{\"o}llmer-Schweizer decomposition of the discounted payoff, \textcolor{black}{ see \cite{fs91}}. This in turn involves {\color{blue}solving} a linear BSDE. In summary, both mean-variance hedging and local risk minimization can be treated from the point of view of BSDEs.

The second observation is that high-dimensional BSDEs (or, equivalently, the corresponding PDEs associated via the Feynman-Kac formula) can be efficiently solved by means of deep learning methods. In recent years, the introduction of advanced and highly parallelized hardware, in particular graphic processing units, have motivated, among other factors, an increasing relevance of statistical learning methods for high-dimensional problems. In the context of numerical methods for high-dimensional PDEs, this has lead in a very short time to a large stream of new deep learning-based methods. {\color{blue}In general, these approaches consist in approximating a specific quantity related to the solution function by deep artificial neural networks (ANNs)}.

The works in \cite{weinar2017} and \cite{Han2018} are of particular relevance for our task. Here the starting point is to solve a PDE, which is transformed via the Feynman-Kac formula to the associated BSDE. The solution of the BSDE is then viewed as the minimizer of a stochastic control problem involving a quadratic loss. First, the BSDE is discretized forward in time via a standard Euler scheme. Then the controls of the BSDE are parametrized at each point in time  by an ANN, hence introducing a family of ANNs. The initial condition and the parameters of the ANNs are optimized by minimizing the expected squared distance between the known terminal condition and the terminal value of the discretized BSDE.

In view of the two observations above, our procedure consists in expressing both quadratic hedging approaches by means of the associated BSDEs. Successively, we apply the deep BSDE solver of \cite{weinar2017} in order to compute all the quantities of interest in a diffusive setting of arbitrary dimension. As a running example, we consider a multi-asset and multi-factor version of the Heston model \cite{heston1993}. On the one hand, this allows to validate our procedure in the one-dimensional setting against well established semi-explicit solutions, notably the results of \cite{cerny2008} and \cite{heath2001numerical}. On the other hand, this hints at the possibility of applying our procedure to the more broad class of diffusive stochastic volatility models. 

In fact, the idea of parametrizing hedging strategies via neural networks is not new. The survey in \cite{Ruf2020} lists indeed around 150 published papers on the topic, and the amount of related literature has been growing even further as a consequence of the recent increased popularity of machine learning in quantitative finance and financial engineering. Similar approaches to the one of \cite{weinar2017} are, for example, \cite{beck2019machine} and \cite{hure2020deep}. Important examples are also \cite{bgtw2019}, who perform hedging by using the expected shortfall as a performance measure in a market with frictions, and \cite{htz2021}, which extends the previous approach to rough volatility models. In the context of equity-linked life-insurance contracts pricing in a model with jumps, we also mention \cite{barigou2022}. We contribute to the literature by showing that deep learning-based methods extend the scope of applicability of quadratic hedging to higher dimensions. This can be interesting from the point of view of institutions that employ financial markets to manage their risks, but also from the financial engineering researchers point of view. Our numerical results show a relative error in pricing the contingent claim roughly below $1\%$ for most of the experiments.

We point out that the specific choice of the deep BSDE solver of \cite{weinar2017} is purely conventional, and we do not exclude that other solvers could be employed in the same context. This choice does however restrict the numerical study of the paper to Markovian models only, while the theoretical treatment of deep quadratic hedging that we provide is more general. Since several deep learning methods have been recently proposed for dealing also with non-Markovian models \textcolor{blue}{and path-dependent payoffs, such as \cite{bfz2024}, \cite{flz2023}, \cite{glp2023} and  \cite{hkt2024}}, we believe that our choice is not a restriction to the potential of the methodology. The solver proposed by \cite{weinar2017} fits well with the multidimensional Heston model considered. Hence, more generally, we can say that the concrete mathematical structure of the model of choice will determine a certain variation of our approach in relation to a certain variation of the solver to be employed. See Section \ref{sec:DeepQuadraticHedging} for a deeper discussion.

The outline of the paper is the following: in Section \ref{setting} we introduce the probabilistic setting and a general market model. To make the discussion self-contained, we devote Section \ref{sec:MeanVarTheory} and \ref{sec:LocalRiskTheory} to presenting respectively mean-variance hedging and local risk minimization, with an emphasis on the link with the BSDEs that we intend to solve numerically. Section \ref{sec:multiHeston} introduces the multi-factor and multi-asset version of the Heston model that we use in the numerical experiments, and contains some technical proofs that are needed in order to ensure the well-posedness of the mean-variance hedging strategy. In Section \ref{sec:DeepQuadraticHedging} we present the methodology: we first provide a self-contained introduction to the deep BSDE solver in Section \ref{sec:solver}, and then we apply it in Sections \ref{sec:deepMV} and \ref{sec:deepLRM} to mean-variance hedging and local risk minimization, respectively.  Finally, Section \ref{sec:numerics} is devoted to the numerical experiments, and Section \ref{sec:concl} concludes the paper with some remarks.

\section{Setting and main assumptions}\label{setting}
Let $(\Omega, \mathcal F,\FF,\p)$ be a complete filtered probability space, with $\FF =(\F_t)_{t\geq 0}$ satisfying the usual hypotheses.  For $m\geq 1$ and $d\geq 0$, we consider an $m$-dimensional and a $d$-dimensional Brownian motion 
$W_t = (W^1_t,\ldots, W^m_t)^\top$ and $B_t = (B^1_t,\ldots, B^d_t)^\top$, where all components are assumed to be mutually independent.
Throughout the paper we will also assume that the filtration $\FF$ is the standard $\p$-augmentation of the filtration generated by both $W$ and $B$. In particular, the components of $W$ drive the diffusion of the tradable assets, while $B$ models the incompleteness of the market. When $d=0$, we use the convention that there is no Brownian motion $B$ and the market is therefore complete.

Let now $q\in \mathbb N\setminus \{0\}$.
We introduce the following notation:
\begin{enumerate}
\item $L^2_{\F_T}(\Omega;\R^q)$ denotes the space of all $\F_T$-measurable $\R^q$-valued random variables $X:\Omega\to\R^q$ such that $$\left\| X\right\|_2:=\left(\E\,\left[\left|X\right|^2\right]\right)^{1/2}<\infty;$$

\item $L^p_{\FF}([0,T];\R^q)$, for $1\leq p<\infty$, denotes the space of all $\FF$-adapted $\R^q$-valued processes $(X_t)_{t\in [0,T]}$ such that $$\left\| X\right\|_p:=\left(\E\,\left[\int^T_0 \left|X_t\right|^p \dd t\right]\right)^{1/p}<\infty;$$

\item ${C}([0,T];\R^q)$ denotes the space of all $\R^q$-valued continuous functions on $[0, T]$;

\item $L^{2}_{\FF}(\Omega;{C}([0,T];\R^q))$ denotes the space of all $\FF$-adapted $\R^q$-valued processes $(X_t)_{t\in [0,T]}$ with $\p$-a.s. continuous sample paths such that
\begin{equation*}
	\E\left[\sup_{t \in [0, T]} \left|X_t\right|^2\right] < \infty;
\end{equation*}

\item $L^{\infty}_{\FF}(\Omega;{C}([0,T];\R^q))$ denotes the space of all $\FF$-adapted $\R^q$-valued essentially bounded processes with $\p$-a.s. continuous sample paths. 
\end{enumerate}


\subsection{The market model}\label{sec:MarketModel}
We consider a financial market with one cash account and $m$ stocks. We denote respectively by $S^0_t$ and $S^i_t$ the price at time $t$ of the cash account and of the $i$-th stock, with $i=1,\ldots, m$. We assume that their dynamics are given by the following stochastic differential equations (SDEs)
 \begin{equation}\label{eq:SDE S}
	\begin{cases}
		 \dd S^0_t = S^0_t r_t  \dd t,  & \qquad S^0_0 = 1,\\
		 \dd S^i_t =S^i_t \left[\mu^i_t \dd t + \sum^m_{j=1} \sigma^{ij}_t \dd W^j_t\right],  & \qquad S^i_0 = s_i, \quad  i=1,\ldots,m,
			\end{cases}
\end{equation}
where the processes $r$, $\mu^i$, $\sigma^{ij}$ ($i,j =1,\ldots,m$) in $\R$ are $\FF$-adapted and are such that existence and uniqueness for solutions to \eqref{eq:SDE S} are guaranteed. We also denote by $\mu:=(\mu^1, \ldots, \mu^m)^\top$ the drift vector  and by  $\sigma:=(\sigma^{ij})_{i,j=1}^m$ the volatility matrix, which is assumed to be invertible $\dd\mathbb{P}\otimes \dd t$-a.s.. We set the market price of risk 
\begin{align}
\label{eq:mktPriceOfRisk}
\phi_t:= \sigma^{-1}_t (\mu_t - r_t \mathbbm 1),
\end{align}
where $\mathbbm 1\in \R^m$ denotes the vector of $1$'s. We assume that the short rate of interest $r$ is bounded from below $\dd \p \otimes \dd t$-a.s..

We remark that the assumption of $\FF$-adapted processes in \eqref{eq:SDE S} covers the case where the coefficients are measurable functions of a set of factors evolving according to additional SDEs: this is the situation which is typically encountered in the context of stochastic volatility models.

Throughout the paper we mainly work with discounted values, and we use the symbol $\tilde{}$ to characterize discounted quantities. In particular, we denote by $\tilde S^i_t:= {S^i_t}/{S^0_t}$  the discounted stock prices. Clearly, one has 
\begin{align}
\label{eq:tildeS}
\quad \dd \tilde S^i_t =\tilde S^i_t \left[(\mu^i_t -r_t) \dd t + \sum^m_{j=1} \sigma^{ij}_t \dd W^j_t\right],  \qquad \tilde S^i_0 = s_i, \quad  i=1,\ldots,m,
\end{align}
and $\tilde S^0_t \equiv 1$. We write $\tilde S_t = (\tilde S^1_t,\ldots, \tilde S^m_t)^{\top}$.

Given the assets on the market, an agent may trade on them by constructing a trading strategy. We introduce the following quantities: 
\begin{enumerate}
\item Let $\xi^i_t\in \mathbb{R}$ be the number of shares of the $i$-th stock  owned by the agent at time $t$, for $i=1,\ldots, m$, and set $\xi_t:=(\xi^1_t,\ldots, \xi^m_t)^\top$.
\item Let $\psi_t\in \R$ denote the units of the cash account in the agent portfolio at time $t$.
\end{enumerate}

\begin{assumption}
\label{ass:strategy}	 
We assume that
\begin{enumerate}
\item $(\xi_t)_{t\geq 0}$ is an $\FF$-predictable process in $\R^{m}$;
\item 
$ \mathbb{E}\left[\int_0^T \left|\xi_s^\top \operatorname{diag}(\tilde{S}_s)\sigma_s\right|^2 \dd s+\left(\int_0^T\left|\xi^\top_s \operatorname{diag}(\tilde{S}_s) (\mu_s - r_s \mathbbm 1)\right|\dd s\right)^2\right]<\infty,$\\
meaning that the stochastic integral $\int_0^T \xi^\top_s \dd \tilde{S}_s$ is well defined in the space of semimartingales. Notice that this comprises the following condition: the process $(\sigma^\top_t \operatorname{diag}(\tilde{S}_t) \xi_t)_{t\geq 0}$ belongs to  $L^2_{\FF}([0,T];\R^m)$;
\item $(\psi_t)_{t\geq 0}$ is an $\FF$-adapted process in $\R$.
\end{enumerate}
\end{assumption}

\begin{defn}
We call any couple of processes $(\xi, \psi)=(\xi_t, \psi_t)_{t\in [0,T]}$  satisfying the above assumptions a  \emph{trading strategy} and we denote by $\mathcal{A}$ the set of all such strategies.
\end{defn}

The (discounted) value process associated to a trading strategy $(\xi, \psi)$ is given by
\begin{equation}\label{eq:defVtilde}
\tilde{V}_t  =  \sum^m_{i=1} \xi^i_t  \tilde{S}^i_t + \psi_t \tilde{S}^0_t=  \sum^m_{i=1} \xi^i_t  \tilde{S}^i_t + \psi_t.
\end{equation}
Notice that, for trading strategies satisfying Assumption \ref{ass:strategy}, the SDE \eqref{eq:defVtilde} admits a unique strong solution which is a square integrable martingale.

\subsection{Self-financing property and financing costs}
A fundamental concept to consider when trading is the self-financing property of a strategy.
\begin{defn}
A trading strategy $(\xi,\psi)\in\mathcal{A}$ is \textit{self-financing} if one has 
\begin{equation*}
	\dd \tilde V_t =  \sum^m_{i=1} \xi^i_t \dd \tilde S^i_t,  \quad \mbox{or} \quad \tilde V_t = y + \int^t_0 \sum^m_{i=1} \xi^i_s \dd \tilde S^i_s,
\end{equation*}
for a given initial wealth $\tilde V_0=y\ge0$.

\end{defn}
It follows that, if the trading strategy is self-financing, then $\tilde V$ solves the following SDE:
\begin{equation}\label{eq:xdyn}
\begin{cases}
\dd  \tilde V_t =  \left[\sum^m_{i=1} \left(\mu^i_t -r_t\right)  \xi^i_t \tilde{S}^i_t \right] \dd t  + \sum^m_{i=1} \xi^i_t \tilde{S}^i_t \sum^m_{j=1} \sigma^{ij}_t  \dd W^j_t,\\
\tilde V_0 = y.
\end{cases}
\end{equation}
In absence of the self-financing property, a strategy may generate inflows or outflows of cash. We introduce the following process to monitor such cashflows over time.
\begin{defn}
The \emph{(discounted) cumulative cost process} $\tilde{C}$ at time $t$ of a strategy $(\xi,\psi)\in \mathcal A$ is defined by
\begin{align}
\label{eq:CostProcess}
\tilde{C}_t := \tilde V_t - \int^t_0 \sum^m_{i=1} \xi^i_s \dd \tilde S^i_s.
\end{align}
We say that the strategy is mean-self-financing if the process $\tilde{C}$ is a square integrable martingale.
\end{defn}
Observe that according with this definition, the cumulative cost for a self-financing strategy coincides with the initial wealth $y$, as it should be.

In the present paper, we are interested in pricing and hedging a (discounted) contingent claim of European type, represented by a positive random variable $\tilde{H}\in L^2_{\mathcal{F}_T}(\Omega;\R)$. The random variable $\tilde H$ is the unknown payoff at time $T$ that is received by the holder of the contract, subject to a certain set of market conditions. The price paid by the buyer of the contract allows the seller to set up a hedging portfolio to possibly cover the contractual liability at time $T$. This rather natural approach is at the basis of the study of dynamic trading strategies of the previously announced form $(\xi,\psi)$.

The ideal situation is reached when the seller/hedger of the contract is able to guarantee the condition $\tilde V_T=\tilde H$ $\p$-a.s. by means of an admissible and self-financing trading strategy. We say in this case that the claim is attainable and, if all contingent claims are attainable, then the market is said to be complete. In our probabilistic setting, anytime that $d>0$ the number of Brownian motions is larger than the number of risky assets that are available for trading. This is a well-known situation where the market is incomplete. Incompleteness means that it will not be possible for some claims to construct an admissible self-financing strategy such that $\tilde V_T=\tilde H$ $\p$-a.s..

Several approaches for pricing and hedging have been proposed for incomplete markets. We can name, e.g., utility indifference pricing \cite{bookIndiffPric} or superhedging \cite{Kramkov1996}. In the present paper, we focus on the class of quadratic hedging approaches. Within this category, we study both mean-variance hedging and local risk minimization. As perfectly illustrated in \cite{Schweizer01}, such approaches are defined by relaxing the structure of the set of strategies over which one optimizes investment decisions. In particular:
\begin{itemize}
\item If we insist on the fact that strategies should be self-financing, while accepting a tracking error at time $T$, then we are employing the mean-variance hedging criterion;
\item If instead we insist on the idea that $\tilde V_T=\tilde H$ $\p$-a.s., while accepting that strategies will fail to be self-financing, then we are considering the local risk minimization approach.
\end{itemize}

We shall provide a brief overview of these two hedging approaches. We refer to \cite{Schweizer01} and references therein for a more formal treatment. In particular we will focus on presenting the link between these hedging approaches and certain BSDEs we would like to solve by means of deep learning methods. Here we consider a single cashflow paid/received at the terminal time $T$. The concepts that we present can however be extended to the case of streams of cashflows over the time interval $[0,T]$, see \cite{Delongbook} and in particular \cite{Schweizer08}.

\section{Mean-variance hedging}\label{sec:MeanVarTheory}

The mean-variance hedging approach corresponds to the following stochastic optimal control problem
\begin{equation}
	\label{hedgingproblemH}
 		 \underset{(\xi,\psi)\in \mathcal{A}_{\text{mv}}}\inf \E\left[\left.\left(\tilde{V}_T-\tilde{H}\right)^2\right| \F_t\right], \ 0\leq t\leq T,\qquad 
 		 \mbox{subject to \eqref{eq:xdyn},}
\end{equation}
where $\mathcal{A}_{\text{mv}}$ denotes the set of admissible trading strategies for such a problem (see Definition \ref{def:Amv} below). 
Following the approach of \cite{lim2004}, the solution to this problem can be linked to the following  system of two BSDEs\\
\begin{align}
	&\label{eq:StochasticRiccati}
	\begin{cases}
		\dd L_t = \left( |\phi_t|^2 L_t + 2\phi^\top_t \Lambda_{1,t} +\frac{\Lambda^\top_{1,t}\Lambda_{1,t}}{L_t}\right)\dd t +\Lambda^\top_{1,t} \dd W_t+\Lambda^\top_{2,t} \dd B_t,\\
		L_T = 1,\\
		L_t >0,
	\end{cases}\\
&\mbox{and}\notag \\
&\label{hedgingproblemX}
	\begin{cases}
		\dd \tilde{X}^{\mathrm{mv}}_t = \frac{1}{S^0_t}\left(\phi^\top_t \eta^{\mathrm{mv}}_{1,t}  -\frac{\Lambda^\top_{2,t}\eta_{2,t}^{\mathrm{mv}}}{L_t}\right)\dd t +\frac{1}{S^0_t}\eta^{\mathrm{mv}, \top}_{1,t} \dd W_t+\frac{1}{S^0_t}\eta^{\mathrm{mv}, \top}_{2,t} \dd B_t,\\
		\tilde{X}^{\mathrm{mv}}_T = \tilde{H},
	\end{cases}
\end{align}
where $\Lambda =(\Lambda_1, \Lambda_2)\in L^2_{\FF}\left([0, T];\R^{m+d}\right)$, $\eta^{\mathrm{mv}} = (\eta_1^{\mathrm{mv}}, \eta_2^{\mathrm{mv}})\in L^2_{\FF}\left([0, T];\R^{m+d}\right)$ and $\phi$ is the market price of risk defined in equation \eqref{eq:mktPriceOfRisk}. 
\\
The first BSDE \eqref{eq:StochasticRiccati} is called \textit{backward stochastic Riccati equation} (BSRE): this is a quadratic BSDE for which proving existence and uniqueness of a solution is not trivial. For example, \cite{lim2004} provides an existence and uniqueness result based on the work of \cite{kob2000} covering, e.g., the Hull-White stochastic volatility case. However, his results do not cover, e.g., the Heston model that we consider for the numerical experiments in Section \ref{sec:numerics}, because of the unboundedness of the coefficients. In Section \ref{sec:existenceUniqueness} we shall discuss existence and uniqueness for the particular BSRE that we consider, while for now we assume existence and uniqueness of a solution to \eqref{eq:StochasticRiccati}. 

The set $\mathcal{A}_{\mathrm{mv}}$ in \eqref{hedgingproblemH} is the set of admissible strategies defined as follows:

\begin{defn}\label{def:Amv} A trading strategy $(\xi,\psi)\in \mathcal{A}$ is admissible for the mean-variance hedging problem \eqref{hedgingproblemH}, if it is self-financing and  if the quantity $L_{\tau_k \wedge T}(\tilde{X}^{\mathrm{mv}}_{\tau_k \wedge T}-\tilde{V}_{\tau_k \wedge T})$ is uniformly integrable for any sequence of $\mathbb{F}$-stopping times $\tau_k \nearrow \infty$ as $k\to\infty$.
\end{defn}

We state now in a more formal way the link between the stochastic control problem \eqref{hedgingproblemH} and the system of BSDEs \eqref{eq:StochasticRiccati}--\eqref{hedgingproblemX}  via the following result, which closely follows \cite[Proposition 3.3]{lim2004}.

\begin{prop}\label{prop:Lim200433}
If the BSDEs \eqref{eq:StochasticRiccati} and \eqref{hedgingproblemX} admit unique solutions $(L,\Lambda)\in L^{\infty}_{\FF}(\Omega;{C}([0,T];\R)) \times L^2_{\mathbb{F}}\left([0, T];\R^{m+d}\right)$ and $(\tilde{X}^{\mathrm{mv}},\eta^{\mathrm{mv}})\in L^{2}_{\FF}(\Omega;{C}([0,T];\R^m)) \times L^2_{\mathbb{F}}\left([0, T];\R^{m+d}\right)$, then 
\begin{equation}\label{eq:ximv}
\xi^{\mathrm{mv}}_t=\operatorname{diag}(\tilde S_t)^{-1}\left(\left(\sigma_t^{-1}\right)^{\top}\left[\phi_t+\frac{\Lambda_{1,t}}{L_t}\right](\tilde X^{\mathrm{mv}}_t- \tilde V^{\mathrm{mv}}_t)+\left(\sigma_t^{-1}\right)^{\top} \eta_{1,t}^{\mathrm{mv}}\right)
\end{equation}
is the unique optimal control for the stochastic control problem \eqref{hedgingproblemH}, where $\tilde V^{\mathrm{mv}}$ is the solution of the SDE \eqref{eq:xdyn} with $\xi = \xi^{\mathrm{mv}}$. 
\end{prop}

\begin{proof}
See Appendix \ref{Proofprop}.
\end{proof}

Notice that the two BSDEs need to be solved sequentially: first, one solves \eqref{eq:StochasticRiccati} which provides the form of the processes $\Lambda_2$ and $L$ appearing in the right-hand side of \eqref{hedgingproblemX}.

 As shown in \cite{lim2004}, the solution to the stochastic Riccati equation characterizes a particular equivalent martingale measure, which is the \textit{variance optimal martingale measure}: this is the pricing measure implicit in the mean-variance hedging approach. More specifically, for $\nu \in L_{\FF}^{2}\left([0, T] ; \mathbb{R}^{d}\right)$, the author first defines the exponential local martingale
\begin{equation}
\label{changeQmv}
M_{\nu,t}=\exp \left\{-\int_{0}^{t} \phi_s^{\top} \dd W_s-\int_{0}^{t} \nu_s^{\top} \dd B_s-\frac{1}{2} \int_{0}^{t}\left(|\phi_s|^{2}+|\nu_s|^{2}\right) \dd s\right\}.
\end{equation}
If $M_{\nu}$ is a true martingale, one can introduce the parametrized family of measures $\dd \mathbb{Q}_{\nu}=M_{\nu,T} \dd \mathbb{P}$ and the variance optimal martingale measure is the measure whose Girsanov kernel for the Brownian motions $B$ solves the following optimization problem
\begin{align}
\label{def:vmm}
\min _{\nu \in L_{\mathbb{F}}^{2}\left([0, T] ; \mathbb{R}^{d}\right)} \E\left[\frac{M_{\nu,T}}{S^0_T}\right]^{2}.
\end{align}
The link between the variance optimal martingale measure and the stochastic Riccati equation is presented in Theorem 4.1 of \cite{lim2004}: let $\nu^{\text{mv}}$ be the solution to the minimization problem \eqref{def:vmm} so that $\mathbb{Q}_{{\text{mv}}}:= \mathbb{Q}_{\nu^{\text{mv}}}$ is the variance optimal martingale measure. If the stochastic Riccati equation \eqref{eq:StochasticRiccati} admits a solution, then 
\begin{align}
\label{eq:GirsanovKernelB}
\nu^{\text{mv}}=-\frac{\Lambda_2}{L}.
\end{align}

The interpretation of the second BSDE is the following: $\tilde X^{\mathrm{mv}}$ represents the dynamics of a portfolio in a fictiously extended financial market. The initial value of $\tilde X^{\mathrm{mv}}$ provides the contingent claim price in the mean-variance hedging approach, see Theorem 4.2 in \cite{lim2004} stating that
\begin{align*}
\tilde X_t^{\mathrm{mv}}=\mathbb{E}^{\mathbb{Q}_{{\text{mv}}}}\left[\left. \tilde H\right| \mathcal{F}_{t}\right].
\end{align*}
This concludes the self-contained introduction for the mean-variance hedging.

\section{Local risk minimisation}\label{sec:LocalRiskTheory}
As mentioned above, self-financing trading strategies have constant cost equal to the initial wealth $\tilde V_0=y$. Hedging a contingent claim with a self-financing strategy means that one can guarantee the payment of the claim at time $T$  simply by investing the initial amount $\tilde V_0$. Using a non-self-financing  strategy carries whereas the risk associated to the cost process $\tilde C$.  We introduce the \textit{risk process} associated to a trading strategy $(\xi,\psi)$ via the following
\begin{align}
\label{eq:riskProcess}
R_t(\xi,\psi) := \E\left[ \left.\left( \tilde{C}_T- \tilde{C}_t\right)^2\right|\F_t\right], \ 0\leq t\leq T.
\end{align}
In  the approach proposed by \cite{Schweizer91}, one aims at minimizing the risk process with respect to small perturbations of the trading strategy $(\xi,\psi)$. The concept of small perturbation requires the introduction of several notations that we skip in the present treatment. We limit ourselves to the following description: the idea is to introduce a measure for the increase of quadratic risk by measuring the relative variation of $R$ over time partitions (see \cite[Equation (1.3)]{Schweizer08}). Such measure of the increase of quadratic risk over perturbations is employed in Definition 1.5 in \cite{Schweizer08} to properly introduce the concept of \textit{local risk minimizing strategy}. For our purposes, we will rely on Theorem 1.6 in \cite{Schweizer08}, suitably reformulated for the case of a single final payoff $\tilde H$. Such result states the equivalence between local risk minimizing strategies as in the above mentioned Definition 1.5 and strategies that reach the final payoff, while being mean-self-financing. Moreover, the associated cost process must be a martingale strongly orthogonal\footnote{Let $M, N$ be two $\mathbb{P}$-local martingales. We say that $M$ and $N$ are strongly orthogonal if  $[M, N]\equiv 0$ up to some evanescent set, implying that $MN$ is a $\mathbb{P}$-local martingale.} to the martingale component of the discounted asset process $\tilde{S}$.

\begin{thrm}[Theorem 1.6 in \cite{Schweizer08}] The followings are equivalent:
\begin{enumerate}
\item $(\xi,\psi)$ is local risk minimizing;
\item $(\xi,\psi)$ is such that $\tilde{V}_T=\tilde{H}$, is mean-self-financing and the cost process \eqref{eq:CostProcess} is strongly orthogonal to the martingale component of $\tilde{S}$.
\end{enumerate}
\end{thrm}

From now on, we will use the second item in the result above as definition of local risk minimizing strategies. From \cite[Proposition 5.2]{Schweizer08}, one has that the existence of a locally risk-minimizing strategy $(\xi^{\text{lr}},\psi^{\text{lr}})$ is equivalent to the existence of the so-called F\"ollmer-Schweizer (FS) decomposition of the discounted payoff. To be self-contained, we report the full proof as it can be found in \cite{Schweizer08}.

\begin{thrm}[Proposition 5.2 in \cite{Schweizer08}]
The discounted payoff $\tilde H$ admits a local risk minimizing strategy $(\xi^{\mathrm{lr}},\psi^{\mathrm{lr}})$ if and only if $\tilde H$ admits the representation
\begin{equation}\label{eq:FSdec}
\tilde H = h_0 + \int^T_0 \zeta^\top_t \dd \tilde S_t + \mathcal H_T,
\end{equation}
for some $h_0\in L^2_{\F_0}(\Omega; \R)$, $\zeta$ satisfying Assumption \ref{ass:strategy}~(1--2) and $\mathcal H=(\mathcal H_t)_{t\in [0,T]}$  a right-continuous, square integrable martingale  strongly orthogonal to the martingale component of $\tilde S$, such that  $\mathcal H_0=0$.
\end{thrm}

\begin{proof}
Given the decomposition \eqref{eq:FSdec} with $\mathcal H$ strongly orthogonal to the martingale component of $\tilde S$, the strategy $(\xi^{\text{lr}},\psi^{\text{lr}})$ defined by 
\begin{align*}
\xi^{\text{lr}}_t &  = \zeta_t,\\
\psi^{\text{lr}}_t & =\tilde  V^{\text{lr}}_t -  \xi^{\text{lr},\top}_t \tilde S_t,
\end{align*}
with $\tilde V^{\text{lr}}_t  =  h_0 + \int^t_0 \zeta^\top_s \dd \tilde S_s + \mathcal H_t$, has the following cost process:
\begin{align*}
\tilde{C}^{\mathrm{lr}}_t=\tilde V^{\text{lr}}_t-\int_0^t \xi^{\text{lr},\top}_s\dd \tilde{S}_s=h_0 + \int^t_0 \zeta^\top_s \dd \tilde S_s + \mathcal H_t-\int_0^t \xi^{\text{lr},\top}_s\dd \tilde{S}_s = h_0+\mathcal H_t.
\end{align*}
Hence the cost process is a square integrable martingale strongly orthogonal to the martingale component of $\tilde{S}$. We conclude that the strategy $(\xi^{\text{lr}}, \psi^{\text{lr}})$ is mean-self-financing and the associated wealth process $\tilde V^{\text{lr}}$ satisfies $\tilde{V}^{\text{lr}}_T=\tilde{H}$. Hence $(\xi^{\text{lr}}, \psi^{\text{lr}})$ is local risk minimizing. Conversely, if the strategy $(\xi^{\text{lr}},\psi^{\text{lr}})$ is local risk minimizing, then we can write $\tilde{H}$ as 
\begin{align*}
\tilde{H}=\tilde{V}^{\text{lr}}_T=\tilde{C}^{\mathrm{lr}}_T+\int_0^T \xi^{\text{lr},\top}_s\dd \tilde{S}_s=\tilde{C}^{\mathrm{lr}}_0+\int_0^T \xi^{\text{lr},\top}_s\dd \tilde{S}_s+\tilde{C}^{\mathrm{lr}}_T-\tilde{C}^{\mathrm{lr}}_0,
\end{align*}
and we obtain \eqref{eq:FSdec} by setting
\begin{align*}
h_0&=\tilde C^{\mathrm{lr}}_0,\\
\zeta&=\xi^{\text{lr}},\\
\mathcal{H}&=\tilde C^{\mathrm{lr}}-\tilde C^{\mathrm{lr}}_0,
\end{align*}
with $\mathcal{H}$ strongly orthogonal to the martingale component of $\tilde{S}$.
\end{proof}

Equipped with the result above, the objective of identifying a local risk minimizing strategy can be fulfilled by finding the F\"ollmer-Schweizer decomposition \eqref{eq:FSdec}. To achieve this, two alternative approaches can be found in the literature. Since we are working in the setting of continuous semimartingales, one first approach consists in computing the F\"ollmer-Schweizer decomposition from the Galtchouk-Kunita-Watanabe decomposition under the \textit{minimal martingale measure}. Here we follow the alternative route which consists in introducing a linear BSDE that we proceed to solve at a later step by means of deep learning methods. The result in full generality is provided by \cite[Proposition 1.1]{ekpq1997}; here we follow the analog arguments as in \cite{ah2012}.

We introduce a BSDE with terminal condition $\tilde{H}$ and with driver a function $f$ that we determine in the sequel:
\begin{equation}
\tilde{X}^{\mathrm{lr}}_t = \tilde H - \int_t^T\frac{1}{S^0_s}\eta^{\mathrm{lr}, \top}_{1,s} \dd W_s-\int_t^T\frac{1}{S^0_s}\eta^{\mathrm{lr}, \top}_{2,s} \dd B_s+\int_t^Tf(s,\tilde{X}^{\mathrm{lr}}_s,\eta_{1,s}^{\mathrm{lr}},\eta_{2,s}^{\mathrm{lr}})  \dd s.
\end{equation}
Let $(\tilde{X}^{\text{lr}}, \eta^{\text{lr}})\in L^{2}_{\FF}(\Omega;{C}([0,T];\R^m)) \times L^2_{\FF}\left([0, T];\R^{m+d}\right)$ be its solution with $\eta^{\text{lr}} := \left(\eta^{\text{lr}}_1, \eta^{\text{lr}}_2\right)$. We obtain the F\"ollmer-Schweizer decomposition of $\tilde H$ if we  assume that the driver $f$ can be chosen such that for some $\xi$ satisfying Assumption \ref{ass:strategy} (1--2) one has
\begin{align*}
\int_0^t\xi^\top_s\dd\tilde{S}_s=\int_0^t\frac{1}{S^0_s}\eta^{\mathrm{lr}, \top}_{1,s} \dd W_s-\int_0^tf(s,\tilde X_s^{\mathrm{lr}},\eta_{1,s}^{\mathrm{lr}},\eta_{2,s}^{\mathrm{lr}})  \dd s.
\end{align*}
However, from the dynamics \eqref{eq:tildeS} of $\tilde{S}$ it is clear that
\begin{align*}
\int_0^t\xi^\top_s\dd\tilde{S}_s=\int_0^t\xi^\top_s\operatorname{diag}(\tilde{S}_s)\sigma_s \dd W_s+\int_0^t\xi^\top_s\operatorname{diag}(\tilde{S}_s)(\mu_s-r_s\mathbbm{1})\dd s.
\end{align*}
Since $\tilde{S}$ is a special semimartingale, the decomposition above is unique, meaning that
\begin{align}
\frac{1}{S^0_t}\eta_{1,t}^{\mathrm{lr}}&=\operatorname{diag}(\tilde{S}_t)\sigma_t^\top\xi_t,\label{eq:xiLR}\\
f(t,\tilde{X}^{\mathrm{lr}}_t,\eta^{\mathrm{lr}}_{1,t},\eta^{\mathrm{lr}}_{2,t})&=-\frac{1}{S^0_t}\eta_{1,t}^{\mathrm{lr}, \top} \sigma_t^{-1}(\mu_t-r_t\mathbbm{1})=-\frac{1}{S^0_t}\eta_{1,t}^{\mathrm{lr}, \top}\phi_t.
\end{align}

To summarize the discussion, we have the following:

\begin{prop}\label{prop:FSBSDE}
The F\"ollmer-Schweizer decomposition of $\tilde{H}$ is given by
\begin{align}
\label{eq:FSdeco}
\tilde{H}=\tilde{X}^{\mathrm{lr}}_0+\int_0^T\frac{1}{S^0_s}\eta^{\mathrm{lr},\top}_{1,s}\left(\operatorname{diag}(\tilde{S}_s)\sigma_s\right)^{-1}\dd \tilde{S}_s+\int_0^T\frac{1}{S^0_s}\eta^{\mathrm{lr},\top}_{2,s}\dd B_s,
\end{align}
where $(\tilde{X}^{\mathrm{lr}}, \eta^{\mathrm{lr}})\in L^{2}_{\FF}(\Omega;{C}([0,T];\R^m)) \times L^2_{\FF}\left([0, T];\R^{m+d}\right)$ is the unique solution to the linear BSDE
\begin{equation}
\label{eq:localriskX}
\tilde{X}^{\mathrm{lr}}_t = \tilde H - \int_t^T\frac{1}{S^0_s}\eta^{\mathrm{lr}, \top}_{1,s} \dd W_s-\int_t^T\frac{1}{S^0_s}\eta^{\mathrm{lr}, \top}_{2,s} \dd B_s-\int_t^T\frac{1}{S^0_s}\eta_{1,s}^{\mathrm{lr}, \top} \phi_s \dd s.
\end{equation}
\end{prop}

\begin{oss}
\begin{enumerate}
\item Equations \eqref{eq:FSdeco} and \eqref{eq:localriskX} correspond, respectively,  to equations (1.17) and (1.18) in \cite{ekpq1997};
\item We stress that the decomposition is computed under the physical measure $\mathbb{P}$;
\item Let  $(\tilde{X}^{\mathrm{lr}}, \eta^{\mathrm{lr}})$ be the unique solution to \eqref{eq:localriskX}.
It follows from the arguments above and from the uniqueness of the F\"ollmer-Schweizer decomposition that for a local minimizing strategy $(\xi^{\mathrm{lr}},\psi^{\mathrm{lr}})$ with associated wealth and cost process $\tilde{V}^{\text{lr}}$ and $\tilde{C}^{\text{lr}}$, respectively, one has
$$
\tilde{C}^{\text{lr}}_t - \tilde{C}^{\text{lr}}_0 = \int^t_0 \eta^{\mathrm{lr}, \top}_{2,s} \dd B_s,
$$
and 
\begin{equation}
\label{eq}
\int^T_t \xi^{\text{lr},\top}_s\dd \tilde{S}_s = \int_t^T\eta^{\mathrm{lr}, \top}_{1,s} \dd W_s + \int_t^T\eta_{1,s}^{\mathrm{lr}, \top} \phi_s \dd s.
\end{equation}
Then, since $\tilde{V}^{\text{lr}}_T = \tilde H$, by substituting equation \eqref{eq} in  \eqref{eq:localriskX} and by using \eqref{eq:CostProcess}, one obtains for any $t\in [0,T]$ that
\begin{align}\label{eq:XeqV}
\tilde{X}^{\text{lr}}_t = \tilde{V}^{\text{lr}}_T - \int^T_t \xi^{\text{lr},\top}_s\dd \tilde{S}_s -  \tilde{C}^{\text{lr}}_T + \tilde{C}^{\text{lr}}_t =   \tilde{V}^{\text{lr}}_t,
\end{align}
namely, the risk minimizing strategy is such that the solution process $\tilde{X}^{\text{lr}}_t$ coincides with the value process $\tilde{V}^{\text{lr}}_t$ at any time $t\in [0,T]$.
\end{enumerate}
\end{oss}

We conclude the treatment of local risk minimization with the following corollary.

\begin{prop}
Under the preceeding assumptions, the optimal hedging portfolio process $\tilde X^{\mathrm{lr}}$ admits the following representation
\begin{align}
\tilde{X}^{\mathrm{lr}}_t=\mathbb{E}^{\mathbb{Q}_{\mathrm{lr}}}\left[\left.\tilde{H}\right|\mathcal{F}_t\right],
\end{align}
where the minimal martingale measure $\mathbb{Q}_{\mathrm{lr}}$ is defined by the following Radon-Nikodym derivative:
\begin{equation}
\label{changeQlr}
\left.\frac{\dd \mathbb{Q}_{\mathrm{lr}}}{\dd \mathbb{P}}\right|_{\mathcal{F}_t}:=\exp\left\{-\int_0^t\phi^\top_s \dd W_s-\frac{1}{2}\int_0^t\phi^\top_s\phi_s \dd s\right\}.
\end{equation}
\end{prop}

\begin{proof}
The result follows from the representation of the value process of the linear BSDE  \eqref{eq:localriskX} as a conditional expectation.
\end{proof}

\begin{oss}
	\label{oss: mmm}
It is interesting to compare the pricing approaches that are implicit in the two techniques considered:  mean-variance hedging and local risk minimization imply indeed two different choices for the pricing measure. We see that in the mean-variance hedging approach both Brownian motions $W$ and $B$ are transformed by the Girsanov change of measure which depends on the solution of the stochastic Riccati equation because of \eqref{eq:GirsanovKernelB}. In the local risk minimization approach only the Brownian motion $W$ is transformed: the only requirement on the measure is that $\tilde{S}$ is a martingale, hence the name \textit{minimal martingale measure}. 
\end{oss}


\section{The multidimensional Heston model}
\label{sec:multiHeston}
We introduce in this section the model that we shall use for the numerical examples in Section \ref{sec:numerics}. We point out that the validity of the numerical approach presented in this paper is not restricted to this particular model specification, which is for illustrative purposes. The choice of a specific model allows us to be concrete concerning the existence and uniqueness of solutions for the BSRE arising in the mean-variance approach, see Section \ref{sec:MeanVarTheory}. The main motivation for choosing (an extension of) the Heston model is to compare our deep quadratic hedging approach with results from the literature. For $m=d=1$, semi-explicit solutions are indeed available for the Heston model for the mean-variance hedging and for the local risk minimization in \cite{cerny2008} and \cite{heath2001numerical}, respectively. The model we propose is inspired by the multidimensional Heston specification considered by \cite{dcgg2013} in the context of foreign exchange markets. Following the notation of Section \ref{setting}, we set $m=d$, so that $\tilde S, W$ and $B$ are $m$-dimensional stochastic processes. 

Let $A = (A_{ij})_{i, j = 1}^{m}$ be a $m\times m$ matrix of coefficients, and let $\kappa =(\kappa_1, \ldots, \kappa_m)^\top$, $\theta =(\theta_1, \ldots, \theta_m)^\top$, $\sigma =(\sigma_1, \ldots, \sigma_m)^\top$ and $\rho =(\rho_1, \dots, \rho_m)^\top$ be  $m$-dimensional vectors of coefficients with $-1 \le \rho_i\le 1$ for each $i=1, \dots, m$. We set
\begin{align}
\label{eq:driftHeston}
\begin{aligned}
\mu_t&:=A\,\mathrm{diag}(Y^2_t)\bar{\mu}+\bar{r}^\top Y^2_t\mathbbm{1},\\
r_t&:=\bar{r}^\top Y^2_t,
\end{aligned}
\end{align}
for some $\bar{\mu},\bar{r}\in\mathbb{R}^m$ and  consider the following system of SDEs
\begin{equation}
\begin{cases}
\label{dSi} \dd \tilde S_t^i = \tilde S_t^i\left(\sum_{j=1}^m A_{ij}Y^{2,j}_t\bar{\mu}_j\dd t + \sum_{j=1}^mA_{ij}Y^j_t \dd W_t^j \right),\\
\dd Y_t^{2,i} = \kappa_i \left(\theta_i-Y_t^{2,i}\right)\dd t + \sigma_i Y_t^i \left(\rho_i \dd W_t^i + \sqrt{1-\rho_i^2}\dd B_t^i\right),
\end{cases} \quad \mbox{for } i=1, \dots, m,
\end{equation}
where $Y^j = \sqrt{Y^{2,j}}$. This representation shows that $\tilde S = (\tilde S^1, \dots, \tilde S^m)^\top$ is obtained by combining all the components of $Y^2 = (Y^{2, 1}, \dots, Y^{2,m})^\top$ and of $W = (W^1, \dots, W^m)^\top$ by means of the matrix $A$. We can also rewrite equation \eqref{dSi} in matrix form as
\begin{equation}
	\label{d_model}
	\begin{cases}
		\dd \tilde S_t 
		= \mathrm{diag}(\tilde S_t)\left(\left(A\,\mathrm{diag}(Y^2_t)\bar{\mu} \right)\dd t + A\,\mathrm{diag}\left(Y_t\right)\dd W_t \right),\\
		\dd Y^2_t = \mathrm{diag}(\kappa)\left(\theta- Y^2_t\right)\dd t + \mathrm{diag}(\sigma)\mathrm{diag}\left(Y_t\right)\left(\mathrm{diag}(\rho)\dd W_t + \mathrm{diag}(\sqrt{\mathbbm{1}-\rho^2})\dd B_t\right),
	\end{cases}
\end{equation}
where $\mathrm{diag}(\sqrt{\mathbbm{1}-\rho^2})$ denotes the diagonal matrix with elements $\sqrt{1-\rho_i^2}$, for $i=1, \dots, m$. 

\begin{oss}
\label{oss:H}
\begin{enumerate}
	\item The drift of the $\tilde S$-dynamics in equation \eqref{d_model} is a generalization of the Heston model for the specifications in \eqref{eq:driftHeston}, where $A\,\mathrm{diag}(Y^2_t)\bar{\mu} =\mu_t - r_t \mathbbm{1}$, for $\mu$ and $r$ as in equation \eqref{eq:driftHeston}. For $m=1$ we indeed retrieve a one-dimensional Heston model
	\begin{equation}
		\label{HestonModel}
		\begin{cases}
			\dd \tilde S_t = \tilde S_t\left(\mu Y^2_t\dd t +Y_t\dd W_t\right),\\
			\dd Y^2_t = \kappa \left(\theta -Y^2_t \right)\dd t + \sigma Y_t\left(\rho \dd W_t +\sqrt{1-\rho^2}\dd B_t \right),
		\end{cases}
	\end{equation}
	with  $\bar{\mu}=\bar{\mu}_1=\mu$ and $A=A_{11}=1$, where $\mu, \kappa, \theta, \sigma >0$ are real constants with $-1 \le \rho \le 1$ and $2\kappa \theta \ge \sigma^2$. In particular, the model \eqref{HestonModel} was proposed by \cite{cerny2008} as a modification of the classical Heston model \cite{heston1993}: since they work under a zero interest rate assumption, their market price of risk is then proportional to $Y$. 
	\item We notice that the vector process $Y^2$ in \eqref{d_model} is a superposition of Heston-type variances, namely of CIR processes. We then require the Feller condition to hold component-wise, namely we assume that
	\begin{equation*}
		2\kappa_i \theta_i \ge \sigma_i^2, \qquad \mbox{for } i = 1,\dots, m,
	\end{equation*}
	in order to guarantee strict positivity for all the components of $Y^2$. 
	\item We also notice that, if $A$ is a diagonal matrix, then $\tilde S$ becomes a superposition of mutually independent Heston models of the form of \eqref{HestonModel}.
\end{enumerate}
\end{oss}

\subsection{Existence and uniqueness results for the BSRE}\label{sec:existenceUniqueness}
The aim of the present section is to prove existence and uniqueness for the stochastic Riccati equation \eqref{eq:StochasticRiccati} in the case of the multidimensional Heston model \eqref{d_model}. \textcolor{black}{Proving that the stochastic Riccati equation admits a unique solution is in general challenging. For example, \cite{kt2002} derive an existence and uniqueness result for the BSRE under the assumption that all coefficients are bounded, and present an application to the problem of mean-variance hedging. However, as we proceed to show, the market price of risk that appears in the driver of the BSRE is unbounded in our setting, so that the results of \cite{kt2002} can not be applied. Despite the results of \cite{kt2002} have been extended under different sets of assumptions in recent years, most contributions allow for random but bounded coefficients, see e.g. \cite{yxz2023}. A notable exception to this is given by the recent Ph.D. thesis \cite{algoulity2021}. For these reasons and for the sake of concreteness, we will limit ourselves to prove that the specific BSRE that we consider in our numerical experiments is well posed.} In particular, we will adapt the approach of \cite{sz2015}  to our setting. 

Let us first notice that, from the definition of the drift and diffusion coefficients in \eqref{eq:driftHeston}, straightforward computations show that the market price of risk takes the form
\begin{align*}
\phi_t = \sigma^{-1}_t (\mu_t - r_t \mathbbm 1)=\mathrm{diag}\left(Y_t\right)\bar{\mu}.
\end{align*}
We further notice that $\phi_t^\top\phi_t = \bar{\mu}^\top \mathrm{diag}\left(Y^2_t\right)\bar{\mu}=\bar{\mu}^{2,\top}Y^2_t$, where $\bar{\mu}^2$ denotes the vector with components $\bar{\mu}_j^2$, $j=1,\ldots, m$.
As a first step, we obtain a closed-form solution for  the BSRE \eqref{eq:StochasticRiccati}, generalizing the results in \cite{sz2015} to the multidimensional case.

\begin{lem}
For the multidimensional Heston model \eqref{dSi}, a solution to the stochastic Riccati equation \eqref{eq:StochasticRiccati} is given by
\begin{align}
\label{eq:BSREsolution}
L_t=\exp\left\{\varphi(t,T)+\psi(t,T)^\top Y^2_t\right\},
\end{align} 
where $\varphi(\cdot,T):[0,T]\to\R$ and $\psi(\cdot, T):[0,T]\to\R^m$ satisfy the following system of Riccati ordinary differential equations (ODEs):
\begin{align}
\label{eq:RiccatiOde}
\begin{aligned}
&\frac{\partial \varphi}{\partial t}+\psi(t,T)^\top\, \mathrm{diag}\left(\kappa\right)\theta = 0,\quad \varphi(T,T)=0,\\
&\frac{\partial \psi}{\partial t}^\top-\psi(t,T)^\top\, \mathrm{diag}\left(\kappa\right)+\frac{1}{2}\psi(t,T)^\top\, \mathrm{diag}\left(\sigma^2\right)\, \mathrm{diag}\left(\psi(t,T)\right)-\bar{\mu}^{2,\top}\\
&\qquad-2\psi(t,T)^\top \, \mathrm{diag}\left(\sigma\right)\, \mathrm{diag}\left(\rho\right)\, \mathrm{diag}\left(\bar{\mu}\right)\\
&\qquad-\psi(t,T)^\top \, \mathrm{diag}\left(\sigma^2\right) \, \mathrm{diag}\left(\rho^2\right) \, \mathrm{diag}\left(\psi(t,T)\right)=0,\quad \psi(T,T)=0.
\end{aligned}
\end{align}
\end{lem}

\begin{proof}
We apply the It\^o formula to \eqref{eq:BSREsolution} and write
\begin{align*}
\dd L_t&=\frac{\partial L_t}{\partial t}+\frac{\partial L_t}{\partial Y_t^2}\dd Y^2_t+\frac{1}{2}\frac{\partial^2 L_t}{\partial (Y_t^2)^2}\dd \left\langle Y^2, Y^2 \right\rangle_t\\
&=\left(\frac{\partial \varphi}{\partial t}+\frac{\partial \psi^\top}{\partial t}Y^2_t\right)L_t\dd t+L_t\psi(t,T)^\top\left(\mathrm{diag}(\kappa)\left(\theta- Y^2_t\right)\dd t\right.\\
&\left.\quad + \mathrm{diag}(\sigma)\mathrm{diag}\left(Y_t\right)\left(\mathrm{diag}(\rho)\dd W_t + \mathrm{diag}(\sqrt{\mathbbm{1}-\rho^2})\dd B_t\right)\right)\\
&\quad +\frac{1}{2}L_t\psi(t,T)^\top \mathrm{diag}\left(\sigma^2\right)\mathrm{diag}\left(Y^2_t\right)\psi(t,T)\dd t\\
&=L_t\overbrace{\left(\frac{\partial \varphi}{\partial t}+\psi(t,T)^\top\mathrm{diag}(\kappa)\theta \right)}^{=0}\dd t+L_t\left(\frac{\partial \psi^\top}{\partial t}-\psi(t,T)^\top\mathrm{diag}(\kappa)\right.\\
&\quad \left.+\frac{1}{2}\psi(t,T)^\top \mathrm{diag}\left(\sigma^2\right)\mathrm{diag}\left(\psi(t,T)\right)\right)Y^2_t\dd t+\Lambda_{1,t}^\top\dd W_t+\Lambda_{2,t}^\top\dd B_t,
\end{align*}
where we defined
\begin{align}
\label{eq:solutionControls}
\begin{aligned}
\Lambda_{1,t}^\top&:=\psi(t,T)^\top \mathrm{diag}\left(\sigma\right) \mathrm{diag}\left(Y_t\right) \mathrm{diag}\left(\rho\right)L_t,\\
\Lambda_{2,t}^\top&:=\psi(t,T)^\top \mathrm{diag}\left(\sigma\right) \mathrm{diag}\left(Y_t\right) \mathrm{diag}\left(\sqrt{\mathbbm{1}-\rho^2}\right)L_t.
\end{aligned}
\end{align}
By focusing on the drift term, we observe that
\begin{align*}
& L_t\left(\frac{\partial \psi^\top}{\partial t}-\psi(t,T)^\top\mathrm{diag}(\kappa)+\frac{1}{2}\psi(t,T)^\top \mathrm{diag}\left(\sigma^2\right)\mathrm{diag}\left(\psi(t,T)\right)\right)Y^2_t\\
&= L_t\left(\frac{\partial \psi^\top}{\partial t}-\psi(t,T)^\top\mathrm{diag}(\kappa) +\frac{1}{2}\psi(t,T)^\top \mathrm{diag}\left(\sigma^2\right)\mathrm{diag}\left(\psi(t,T)\right)\right)Y^2_t \\
&\quad - \phi_t^\top\phi_t L_t - 2\phi^\top_t\Lambda_{1,t}- \frac{\Lambda_{1,t}^\top\Lambda_{1,t}}{L_t} +\phi_t^\top\phi_t L_t+ 2\phi^\top_t\Lambda_{1,t}+\frac{\Lambda_{1,t}^\top\Lambda_{1,t}}{L_t}\\
&=\phi_t^\top\phi_t L_t+ 2\phi^\top_t\Lambda_{1,t}+\frac{\Lambda_{1,t}^\top\Lambda_{1,t}}{L_t},
\end{align*}
due to \eqref{eq:RiccatiOde} and to the fact that $\phi_t^\top\phi_t=\bar{\mu}^{2,\top}Y^2_t$. Hence the proof is complete.
\end{proof}

\begin{oss}
It is immediate to observe that we can write the vector Riccati ODE in \eqref{eq:RiccatiOde} as a vector of scalar Riccati ODEs that can be solved independently of each other, namely
\begin{align}
	\label{eq:RiccatiOdej}
\frac{\partial \psi_j}{\partial t} - \psi_j(t,T)\kappa_j+\frac{1}{2}\psi_j^2(t,T)\sigma^2_j-\bar{\mu}^2_j-2\psi_j(t,T)\sigma_j\rho_j\bar{\mu}_j-\psi^2_j(t,T)\sigma^2_j\rho^2_j=0,
\end{align}
for $j=1,\ldots,m$, meaning that one can resort to a standard existence and uniqueness result for the scalar case.
\end{oss}

In the next step, we state a standard result on existence and uniqueness for the Riccati system \eqref{eq:RiccatiOdej}.

\begin{lem}
Let us assume that $\rho^2_j<\frac{1}{2}$ 
for all $j=1,\ldots,m$. Then, for every $j=1,\ldots,m$, there exists a unique solution $\psi_j(\cdot,T)$ to the Riccati ODE \eqref{eq:RiccatiOdej} with $\sup_{0\leq t\leq T}|\psi_j(t,T)|<\infty$. Moreover $\psi_j(t,T)\leq 0$ for all $t\in[0,T]$, $j=1,\ldots,m$.
\end{lem}

\begin{proof}The result follows from e.g. Lemma 10.12 in \cite{Filipovic2009} by setting, in his notation, $A=\sigma^2_j\left(\frac{1}{2}-\rho^2_j\right)$, $B=-(\kappa_j+2\bar{\mu}_j\sigma_j\rho_j)$, $C= \bar{\mu}^2_j$, and by applying the change of variable $t\mapsto T-t$. The condition $A>0$ holds true if and only if $\rho^2_j<\frac{1}{2}$. We also need to check $C^2>0$, but this is trivially satisfied. Finally, we need to check that $B^2+4AC\in\mathbb{C}\setminus\mathbb{R}_{-}$, which is also trivially satisfied.
\end{proof}

We can now prove uniqueness by adapting the approach of \cite{sz2015} to our setting.

\begin{prop}
The triple $(L,\Lambda_1,\Lambda_2)$ as provided in \eqref{eq:BSREsolution}, \eqref{eq:RiccatiOde} and \eqref{eq:solutionControls} is the unique solution to the BSRE \eqref{eq:StochasticRiccati}.
\end{prop}

\begin{proof}
Define $\mathscr{L}:=\log L$ and $Z_{i}:=\frac{\Lambda_{i}}{L}$, $i\in\{1,2\}$. An application of the It\^o formula shows that
\begin{align*}
\dd \mathscr{L}_t=\left(|\phi_t|^2+2\phi_t^\top Z_{1,t}+\frac{1}{2}|Z_{1,t}|^2-\frac{1}{2}|Z_{2,t}|^2\right) \dd t+ Z_{1,t}^\top\dd W_t+ Z_{2,t}^\top\dd B_t.
\end{align*}
Let us introduce the measure $\tilde{\mathbb{P}}$ defined via
\begin{align*}
\left.\frac{\dd \tilde{\mathbb{P}}}{\dd \mathbb{P}}\right|_{\mathcal{F}_t}:=\exp\left\{-2\int_0^t|\phi_s|^2\dd s-2\int_0^t\phi_s^\top\dd W_s\right\},
\end{align*}
which is a true martingale thanks to a component-wise application of Theorem 2.1 in \cite{Mijatovi2010}. Under $\tilde{\mathbb{P}}$, $\mathscr{L}$ has dynamics
\begin{align*}
\dd \mathscr{L}_t=\left(|\phi_t|^2+\frac{1}{2}|Z_{1,t}|^2-\frac{1}{2}|Z_{2,t}|^2\right) \dd t+ Z_{1,t}^\top\dd \tilde{W}_t+ Z_{2,t}^\top\dd \tilde{B}_t.
\end{align*}
We introduce then a second measure change defined via
\begin{align*}
\left.\frac{\dd \check{\mathbb{P}}}{\dd \tilde{\mathbb{P}}}\right|_{\mathcal{F}_t}:=\exp\left\{-\frac{1}{2}\int_0^t|Z_s|^2\dd s-\int_0^t Z_{1,s}^\top\dd \tilde{W}_s+\int_0^t Z_{2,s}^\top\dd \tilde{B}_s\right\},
\end{align*}
which is a true martingale due to the boundedness of $\psi_j$ $j=1,\ldots, m$, allowing us to apply Corollary A1 in \cite{sz2015}. The dynamics of $\mathscr{L}$ under $\check{\mathbb{P}}$ are
\begin{align*}
\dd \mathscr{L}_t=\left(|\phi_t|^2-\frac{1}{2}|Z_{1,t}|^2+\frac{1}{2}|Z_{2,t}|^2\right) \dd t+ Z_{1,t}^\top\dd \check{W}_t+ Z_{2,t}^\top\dd \check{B}_t.
\end{align*}
We now proceed by contradiction. Assume that there are two distinct solutions $(\mathscr{L},Z)$ and $(\mathscr{L}^\prime,Z^\prime)$, and define $(\Delta \mathscr{L}, \Delta Z)=(\mathscr{L}-\mathscr{L}^\prime,Z-Z^\prime)$. Under the measure $\check{\mathbb{P}}$, we have that
\begin{align*}
\dd \Delta \mathscr{L}_t=\left(-\frac{1}{2}|\Delta Z_{1,t}|^2+\frac{1}{2}|\Delta Z_{2,t}|^2\right)\dd t + \Delta Z_{1,t}^\top \dd \check{W}_t+\Delta Z_{2,t}^\top \dd \check{B}_t,
\end{align*}
which is a quadratic BSDE for which there exists a unique solution thanks to the results of \cite{kob2000}. Such a solution is given by $(0,\mathbf{0})$ with $\mathbf{0}\in\mathbb{R}^{d+m}$ being a vector of $0$'s. This implies $\mathscr{L}=\mathscr{L}^\prime$, $Z_{1}=Z_{1}^\prime$ and $Z_{2}=Z_{2}^\prime$, which is a contradiction. Hence the solution to the BSRE is unique.
\end{proof}

\section{{Our methodology}}
\label{sec:DeepQuadraticHedging}
We have seen in Section \ref{sec:MeanVarTheory} and Section \ref{sec:LocalRiskTheory} that both mean-variance hedging and local risk minimization can be treated from the point of view of BSDEs. 
Our next step is to solve these BSDEs by deep learning methods. In particular, we shall focus on the deep BSDE solver proposed by \cite{weinar2017}. Notice that in Section \ref{sec:MarketModel} we assumed that the coefficients in \eqref{eq:SDE S} are $\FF$-adapted. The algorithm of  \cite{weinar2017} that we employ introduces the further assumption that the market model is Markovian. However, due to our modelling choice, it is natural for us to apply a Markovian solver. Markovian diffusive models constitute a vast class that encompasses those which are most popular both from the point of view of the academic literature and industrial applications, such as diffusive stochastic volatility models.

Recently non-Markovian models such as the rough Heston model in \cite{eer2019} and the rough-Bergomi model of \cite{bfg2016} have been the subject of several contributions in the literature. In the context of such non-Markovian models, valution equations take the form of backward stochastic partial differential equations (BSPDEs) which can be numerically solved by suitable extensions of the original solver of \cite{weinar2017} as proposed in \cite{bqy2022}. Alternatively, \cite{jacquier2023deep} propose an approach based on curve-dependent PDEs, which arise from a Feynman-Kac theorem for rough volatility models proved in \cite{viens2019martingale}. More generally, we can say that the concrete mathematical structure of the model of choice will determine a certain variation of the reasoning we propose in the present paper.

 We also point out that other deep learning-based solvers for BSDEs (or associated PDEs) in the Markovian setting can be found in the literature, see for instance the dynamic programming based algorithms proposed in  \cite{hpw2021} and the deep splitting method in \cite{bbcjn2021}. We refer to \cite{bbcjn2021} for an extensive literature review. The solver proposed by E, Han and Jentzen in \cite{weinar2017} provides us all the tools necessary for developing our methodology, but we do not exclude that other solvers could be also used in the same context.

{\color{black}
\begin{oss}
The reason why we choose to work with the deep BSDE solver by E-Han-Jentzen is mainly due to the fact that it can naturally be applied to solve two BSDEs recursively, as we need to do in the framework of mean-variance hedging.  Moreover,  local deep algorithms, as the one presented in \cite{hpw2021}, usually require to solve as many optimization problems as the number of time steps. This translates into a high computational cost when the  discretization in time becomes more accurate. Since our approach requires a sequential application of the solver, we believe that a solver of this type is not the most suitable for our context. 
On the other hand, it is well known that  global algorithms, such as the one in \cite{weinar2017} that we use, can easily loose accuracy and get stuck in poor local minima. However, as we shall see in Section \ref{sec:numerics}, this does not prevent us from obtaining satisfactory results in all the performed tests under suitable choice of the learning parameters.
\end{oss}
}

We shall now provide a self-contained presentation of the deep BSDE solver proposed by \cite{weinar2017} as it is relevant to our setting. After that, we show how to apply the solver in the context of  mean-variance hedging and local risk minimization. Notice that the treatment of the solver is general and not restrictive to the Heston model considered in the numerical experiments.

\subsection{The deep BSDE solver by E, Han and Jentzen}\label{sec:solver}
We start from a general Forward-Backward stochastic differential equation (FBSDE). Let $\left(\Omega,\F,\p\right)$ be a probability space rich enough to support an $\R^{q}$-valued Brownian motion $\mathcal{W}=(\mathcal{W}_t)_{t\in[0,T]}$. Let $\mathbb{F}=(\F_t)_{t\in[0,T]}$ be the standard $\mathbb{P}$-augmentation of the filtration generated by $\mathcal{W}$.
Let us consider an FBSDE in the following general form:
\begin{align}
\mathcal{X} _ { t } &  = x + \int _ { 0 } ^ { t } b \left( s, \mathcal{X} _ { s } \right) \mathrm d s + \int _ { 0 } ^ { t } a \left( s, \mathcal{X} _ { s } \right) \mathrm d \mathcal{W} _ { s } , \quad x \in \mathbb { R } ^ {q}, \label{eq:forward}\\
\mathcal{Y}_{t}     & = \vartheta (\mathcal{X}_T) +\int_{t}^{T} h \left(s, \mathcal{X}_{s}, \mathcal{Y}_{s}, \mathcal{Z}_{s}\right) \mathrm d s- \int_{t}^{T} \mathcal{Z}_{s}^\top \mathrm d \mathcal{W}_{s}, \quad t \in[0, T], \label{eq:backward}
\end{align}
where the vector fields $b:[0,T]\times \mathbb{R}^{q} \to \mathbb{R}^{q}$, $a:[0,T]\times \mathbb{R}^{q}\to \mathbb{R}^{{q}\times {q}}$, $h:[0,T]\times \R^{q}\times \R \times \R^{q} \to \R$ and $\vartheta:\R^{q}\to \R$ satisfy
suitable assumptions ensuring existence and uniqueness results.

The FBSDE above is intimately linked with the following stochastic control problem:
\begin{align}\label{eq:min}
& \underset{y,\; \mathcal{Z}=(\mathcal{Z}_t)_{t\in [0,T]}}{\text{minimise}} \; \E\left[ \left| \vartheta(\mathcal{X}_T) - \mathcal{Y}_T\right|^2\right],\\
& \text{subject to } \begin{cases}
\mathcal{X} _ { t }  &\hspace{-0.2cm} =   x + \int _ { 0 } ^ { t } b \left( s, \mathcal{X} _ { s } \right) \mathrm d s + \int _ { 0 } ^ { t } a \left( s, \mathcal{X} _ { s } \right) \mathrm d \mathcal{W} _ { s } ,\\
\mathcal{Y}_{t}  
  &   \hspace{-0.2cm} = y  -\int_{0}^{t} h \big(s, \mathcal{X}_{s}, \mathcal{Y}_{s}, \mathcal{Z}_{s}\big) \mathrm d s +\int_{0}^{t} \mathcal{Z}_{s}^\top \, \mathrm d \mathcal{W}_{s}, \quad t \in[0, T].
\end{cases}\label{eq:dinXY}
\end{align}
More precisely, a solution $(\mathcal{Y},\mathcal{Z})$ to \eqref{eq:backward} is a  minimiser of the problem \eqref{eq:min}. 

The idea of the deep BSDE solver is to numerically solve a discretized version of the  optimal control problem \eqref{eq:min}--\eqref{eq:dinXY}.
For $N\in \mathbb N$, we introduce a time discretization  $0=t_0< t_1<\ldots <t_N=T$. Without loss of generality, we consider a uniform mesh with step $\Delta t$ such that $t_n = n\Delta t$, $n=0,\ldots,N$, and define $\Delta \mathcal{W}_{n} := \mathcal{W}_{t_{n+1}} - \mathcal{W}_{t_n}$. We then consider an Euler-Maruyama discretization of the system \eqref{eq:min}--\eqref{eq:dinXY},  i.e.
\begin{align}
\overline{\mathcal{X}}_{{n+1}}    & =  \overline{\mathcal{X}}_{n}  + b(t_n,\overline{\mathcal{X}}_{n}) \Delta t + a(t_n,\overline{\mathcal{X}}_{n}) \Delta \mathcal{W}_{n}, && \overline{\mathcal{X}}_{0} = x, \label{eq:Eforward}\\
\overline{\mathcal{Y}}_{{n+1}} &=  \overline{\mathcal{Y}}_{n} - h(t_n, \overline{\mathcal{X}}_{n}, \overline{\mathcal{Y}}_{n}, \overline{\mathcal{Z}}_{n}) \Delta t + \overline{\mathcal{Z}}_{n}^\top\Delta \mathcal{W}_{n}, && \overline{\mathcal{Y}}_0 = y,\label{eq:Ebackward}
\end{align}
for $n=0,\ldots, N-1$. The deep BSDE solver consists in approximating, at each time step $n$,  the control process $\overline{\mathcal{Z}}_n$ in \eqref{eq:Ebackward} by means of an artificial neural network (ANN), namely by a function $\cN^{\mathcal Z}_n:\R^q\to \R^q$ of the form $ \cN^{\mathcal Z}_n(x) = \mathcal{L}_{\ell}^n \circ \varrho \circ \mathcal{L}^n_{\ell-1} \circ \ldots \circ \varrho \circ \mathcal{L}^n_{1}(x)$,
where all $\mathcal{L}^n_j$, for $j=1,\ldots,\ell$ and $n=0, \ldots, N-1$, are affine transformations and  $\varrho$, called \textit{activation function}, is a univariate function that is applied component-wise to vectors. Equation \eqref{eq:Ebackward} then becomes
\begin{equation}
	\widehat{\mathcal{Y}}_{{n+1}} =  \widehat{\mathcal{Y}}_{n} - h(t_n, \overline{\mathcal{X}}_{n}, \widehat{\mathcal{Y}}_{n}, \cN^{\mathcal{Z}}_n(\overline{\mathcal{X}}_{n})) \Delta t + \cN^{\mathcal{Z}, \top}_n(\overline{\mathcal{X}}_{n})\Delta \mathcal{W}_{n}, \qquad \qquad\quad\widehat{\mathcal{Y}}_0 = y.\label{eq:EbackwardHat}
\end{equation}
Since we are approximating  the control process $\overline{\mathcal{Z}}_n$ with a different ANN at each time step $n$, we need in practice a family $\left(\cN^{\mathcal{Z}}_n\right)_{n=0}^{N-1}$ of ANNs.

By denoting with $\mathcal{P}((\mathcal{N}^{{\mathcal{Z}}}_n)_{n=0}^{N-1})$ the set of all the parameters, i.e. weights and biases, of the ANNs family, the stochastic control problem \eqref{eq:min} reduces to 
\begin{align}\label{eq:ocpLMV}
	\underset{y,\mathcal{P}((\mathcal{N}^{{\mathcal{Z}}}_n)_{n=0}^{N-1})}{\text{minimise}} \; \E\left[ \left| \vartheta(\overline{\mathcal{X}}_N) - \widehat{\mathcal{Y}}_{{N}}\right|^2\right].
\end{align}
By stochastic gradient descent, the deep BSDE solver trains the ANNs and finds the optimal initial value $\widehat{y}$ and the set of parameters ${\mathcal{P}}((\widehat{\mathcal{N}}^{{\mathcal{Z}}}_n)_{n=0}^{N-1})$ characterizing the optimal family of networks $(\widehat{\cN}^{\mathcal{Z}}_n)_{n=0}^{N-1}$.

\subsection{Details on the parametrization of the solver}\label{sec:paramSolver}
In this section we link the setup of the deep BSDE solver with the market model from Section \ref{sec:MarketModel}. With respect to the general presentation of the solver, we set $q = m+d$ and $\mathcal{W}=(W,B)^\top \in\mathbb{R}^{q}$. Concerning the forward SDE \eqref{eq:forward}, the coefficients appearing in the dynamics of the risky assets are $\mathbb{F}$-measurable, meaning that they will depend in general on both the Brownian motions $W$ and $B$. To capture this, we introduce an additional $\mathbb{R}^d$-dimensional risk factor process $Y=(Y_t)_{t\in [0,T]}$ following the dynamics
\begin{equation}
	\label{eq:dynY}
	\dd Y_t  = \gamma(t, Y_t)\dd t + \Gamma(t, Y_t)\dd\mathcal{W}_t, \qquad \qquad \quad Y_0 = y_0,
\end{equation}
with $\gamma:[0, T]\times \R^d \to \R^d$ and $\Gamma:[0, T]\times\R^{d}\to \R^{d\times q}$ of the form $\Gamma = (\Gamma^W, \Gamma^B)$ for $\Gamma^W:[0,T]\times \R^d \to \R^{d\times m}$ and $\Gamma^B:[0,T]\times \R^d \to \R^{d\times d}$. The coefficients $r,\mu^i$ and $\sigma^{i,j}$ in \eqref{eq:tildeS} are assumed to be sufficiently regular functions of $Y$, for $i,j=1,\ldots,m$. We omit such dependence to ease the notation.

The forward process $\mathcal{X}$ in \eqref{eq:forward} collects then both the $m$ risky assets $\tilde S^1,\ldots,\tilde S^m$ in its first $m$ components, and, additionally, the factor process $Y$, so that $\mathcal{X}=(\tilde S^1,\ldots,\tilde S^m,Y^1,\ldots,Y^d)^\top$. From equations \eqref{eq:tildeS} and \eqref{eq:dynY}, the precise form of the vector fields in \eqref{eq:forward} is then the following: for the diffusion matrix we set
\begin{align}\label{adef}
	a(t,\mathcal{X}_t)= 
	\begin{pmatrix}
		\operatorname{diag}(\tilde S_t) \sigma_t & \mathbf 0 \\
		\Gamma^W(t, Y_t)  & \Gamma^B(t, Y_t)
	\end{pmatrix},
\end{align} 
where $\mathbf 0 \in\R^{m\times d}$ is a matrix of $0$'s, and, similarly, for the drift term we set
\begin{align}\label{bdef}
	b(t,\mathcal{X}_t)=
	\begin{pmatrix}
		\operatorname{diag}(\tilde S_t)\left(\mu_t-r_t\mathbbm{1}\right)\\
		\gamma(t, Y_t)
	\end{pmatrix}.
\end{align}
We assume the  coefficients $\gamma$ and $\Gamma$ to be sufficiently regular so that the SDE  \eqref{eq:forward} with $a$ and $b$ given by \eqref{adef} and \eqref{bdef} admits a unique strong solution.

The backward process $\mathcal{Y}$ in \eqref{eq:backward} is given, case by case, by the BSDE that the quadratic hedging approach selected requires to solve. We shall be more precise about this in the next two paragraphs.

\subsubsection{Deep mean-variance hedging}\label{sec:deepMV}
As illustrated in Section \ref{sec:MeanVarTheory}, there are two BSDEs that we need to solve in the mean-variance hedging case. 
In a first step we discretize \eqref{eq:StochasticRiccati} forward in time,  which is the stochastic Riccati equation parametrizing the variance optimal martingale measure. We get
\begin{equation*}
\begin{cases}
\overline{{L}}_{{n+1}} =  \overline{{L}}_{n} +\left(|\phi_n|^2 \overline{L}_n + 2\phi^\top_n \overline{\Lambda}_{1,n} +\frac{\overline{\Lambda}^\top_{1,n}\overline{\Lambda}_{1,n}}{\overline{L}_n}\right)\Delta t + \overline{\Lambda}_{n}^\top\Delta \mathcal{W}_{n},\qquad \text{for $n=0,\ldots, N-1$},\\
 \overline{L}_0 = y_L,
\end{cases}
\end{equation*}
where  $\phi_n := \phi_{t_n}$, $\overline\Lambda^\top_n = (\overline\Lambda^\top_{1,n},\overline\Lambda^\top_{2,n})$ and $y_L\ge0$ is the initial value to be determined.
Next, in order to apply the solver, we assume that $\overline{\Lambda}_n^\top = \left(\overline{\Lambda}_{1,n}^\top,
\overline{\Lambda}_{2,n}^\top\right)$ is parametrized at each time step $n$ by an ANN $\mathcal{N}^{{\Lambda}}_n$ of the form
\begin{align*}
\mathcal{N}^{{\Lambda}}_n=\left(\begin{array}{c}
\mathcal{N}^{{\Lambda}}_{1,n} \\ 
\mathcal{N}^{{\Lambda}}_{2,n}
\end{array} \right),
\end{align*}
with $\mathcal{N}^{{\Lambda}}_{1,n}:\mathbb{R}^{q}\to \mathbb{R}^m$ and $\mathcal{N}^{{\Lambda}}_{2,n}:\mathbb{R}^{q}\to \mathbb{R}^d$.  The discretized BSDE takes then the form
\begin{equation}\label{eq:deepLMV}
\begin{cases}
\widehat L_{{n+1}} =  \widehat L_{n} +\left( |\phi_n|^2 \widehat L_n + 2\phi^\top_n \mathcal{N}^{{\Lambda}}_{1,n}{(\overline{\mathcal X}_n)} +\frac{\left(\mathcal{N}^{{\Lambda}}_{1,n}{(\overline{\mathcal X}_n)}\right)^\top\mathcal{N}^{{\Lambda}}_{1,n}{(\overline{\mathcal X}_n)}}{\widehat L_n}\right)\Delta t + \left(\mathcal{N}^{{\Lambda}}_{n}{(\overline{\mathcal X}_n)}\right)^\top\Delta \mathcal{W}_{n},\\
\hfill{\text{for $n=0,\ldots, N-1$}},\\
\widehat L_0 = y_L,
\end{cases}
\end{equation}
where $(\overline{\mathcal{X}}_n)^N_{n=0}$ is the discretized forward process \eqref{eq:Eforward}. 
As a consequence of the previous approximations, taking into account the terminal condition of the stochastic Riccati BSDE, the stochastic control problem \eqref{eq:min} involves a minimization over the initial value $y_L$ and the set of parameters of the ANNs that we denote by $\mathcal{P}((\mathcal{N}^{{\Lambda}}_n)_{n=0}^{N-1})$. We then need to solve
\begin{align}\label{eq:ocpLMV2}
\underset{y_L,\mathcal{P}((\mathcal{N}^{{\Lambda}}_n)_{n=0}^{N-1})}{\text{minimise}} \; \E\left[ \left| 1 - \widehat{{L}}_{{N}}\right|^2\right].
\end{align}
Let $\widehat y_L$ and $(\widehat{\mathcal{N}}^{{\Lambda}}_n)_{n=1}^{N-1}$ be respectively the initial condition and the ANNs obtained from  the above minimization problem via the deep solver. 
With an abuse of notation, we denote with $(\widehat{L}_n)^N_{n=0}$ the approximated solution to \eqref{eq:StochasticRiccati} which is obtained by \eqref{eq:deepLMV} with $y_L=\widehat y_L$ and ${\mathcal{N}}^{\Lambda}=\widehat{\mathcal{N}}^{{\Lambda}}$.

Equipped with such numerical solution, we focus now on the portfolio dynamics in the fictiously extended financial market. Once again, we perform an Euler-Maruyama discretization, meaning that we consider
\begin{equation*}
\begin{cases}
\overline{{X}}_{{n+1}}^{\mathrm{mv}}=  \overline{{X}}_{n}^{\mathrm{mv}}  + \frac{1}{S^0_n}\left(\phi^{\top}_n\overline{\eta}^{\mathrm{mv}}_{1,n}-\frac{\overline{\Lambda}^\top_{2,n}\overline{\eta}_{2,n}^{\mathrm{mv}}}{\overline{L}_n}\right) \Delta t+\frac{1}{S^0_n}\overline{\eta}^{\mathrm{mv}, \top}_{n}\Delta \mathcal W_n,\qquad \text{for $n=0,\ldots, N-1$},\\
\overline{{X}}^{\mathrm{mv}}_{{0}}    = y_X^{\mathrm{mv}},
\end{cases}
\end{equation*}
where $\overline{\eta}_{n}^{\mathrm{mv}, \top} = \left(\overline{\eta}_{1,n}^{\mathrm{mv}, \top}, \overline{\eta}_{2,n}^{\mathrm{mv}, \top}\right)$ and $y_X^{\mathrm{mv}}\ge0$ is the initial value to be determined.
Next, we introduce a second family of ANNs $\mathcal{N}^{{\eta^{\mathrm{mv}}}}_n:\mathbb{R}^{q}\to\mathbb{R}^{q}$ of the form
\begin{align*}
\mathcal{N}^{{\eta^{\mathrm{mv}}}}_n=\left(\begin{array}{c}
\mathcal{N}^{{\eta^{\mathrm{mv}}}}_{1,n} \\ 
\mathcal{N}^{{\eta^{\mathrm{mv}}}}_{2,n}
\end{array} \right),
\end{align*}
with $\mathcal{N}^{{\eta^{\mathrm{mv}}}}_{1,n}:\mathbb{R}^{q}\to \mathbb{R}^m$ and $\mathcal{N}^{{\eta^{\mathrm{mv}}}}_{2,n}:\mathbb{R}^{q}\to \mathbb{R}^d$, depending on the set of parameters denoted by  ${\mathcal{P}((\mathcal{N}^{{\eta^{\mathrm{mv}}}}_n)_{n=1}^{N-1})}$. Such networks allow us to introduce the following approximation:
\begin{equation}\label{eq:deepXMV}
\begin{cases}
\widehat X^{\mathrm{mv}}_{{n+1}} =  \widehat X^{\mathrm{mv}}_{n} +\frac{1}{S^0_n}\left(\phi^{\top}_n\mathcal{N}^{{\eta^{\mathrm{mv}}}}_{1,n}(\overline{\mathcal X}_n)-\frac{\left(\widehat{\mathcal{N}}^{{\Lambda}}_{2,n}(\overline{\mathcal X}_n)\right)^\top \mathcal{N}^{{\eta^{\mathrm{mv}}}}_{2,n}(\overline{\mathcal X}_n)}{\widehat L_n}\right) \Delta t+\frac{1}{S^0_n}\left(\mathcal{N}^{{\eta^{\mathrm{mv}}}}_{n}(\overline{\mathcal X}_n)\right)^\top\Delta \mathcal{W}_{n},\\
\hfill{\text{for $n=0,\ldots, N-1$}},\\
\widehat X^{\mathrm{mv}}_0=y_X^{\mathrm{mv}},
\end{cases}
\end{equation}
where $\widehat{\mathcal{N}}^{{\Lambda}}_{2,n}$ and ${\widehat L_n}$ are results of the first optimization problem \eqref{eq:ocpLMV2}. 
The associated stochastic control problem \eqref{eq:min} involves a minimization over the initial value $y_X^{\mathrm{mv}}$ and the networks' parameters ${\mathcal{P}((\mathcal{N}^{{\eta^{\mathrm{mv}}}}_n)_{n=0}^{N-1})}$:
\begin{align}\label{eq:ocpXMV}
\underset{y_X^{\mathrm{mv}},\mathcal{P}((\mathcal{N}^{{\eta^{\mathrm{mv}}}}_n)_{n=0}^{N-1})}{\text{minimise}} \; \E\left[ \left| \tilde{H} - \widehat{X}^{\mathrm{mv}}_{N}\right|^2\right].
\end{align}
We denote by $\widehat{y}_X^{\mathrm{mv}}$ and $(\widehat{\mathcal{N}}^{{\eta^{\mathrm{mv}}}}_n)_{n=0}^{N-1}$ the corresponding optimal solutions  computed with the deep solver. With abuse of notation, from now on  $(\widehat{X}^{\mathrm{mv}}_n)^N_{n=0}$ will denote the discrete time process which is obtained from \eqref{eq:deepXMV} by substituting $y_X^{\mathrm{mv}} = \widehat y_X^{\mathrm{mv}}$ and ${\mathcal{N}}^{{\eta^{\mathrm{mv}}}} = \widehat{\mathcal{N}}^{{\eta^{\mathrm{mv}}}}$, and that we take as the approximated solution to \eqref{hedgingproblemX}.


The approximated optimal trading strategy and the corresponding value process can then be recursively computed starting from equation \eqref{eq:xdyn} and \eqref{eq:ximv} as follows:
\begin{equation*}
	\begin{cases}
		\widehat V^{\mathrm{mv}}_0 = \widehat{y}_X^{\mathrm{mv}},\\
		\widehat\xi^{\mathrm{mv}}_n=\text{diag}(\overline{S}_n)^{-1}\left(\left(\sigma_n^{-1}\right)^{\top}\left[\phi_n+\frac{\widehat{\mathcal{N}}^{\Lambda}_{1,n}(\overline{\mathcal X}_n)}{\widehat{L}_n}\right](\widehat X^{\mathrm{mv}}_n- \widehat V^{\mathrm{mv}}_n)+\left(\sigma_n^{-1}\right)^{\top} \widehat{\mathcal{N}}^{{\eta^{\mathrm{mv}}}}_{1,n}(\overline{\mathcal X}_n)\right),\\
		\widehat{\psi}^{\mathrm{mv}}_n = \widehat V^{\mathrm{mv}}_n -\text{diag}(\overline{S}_n) \widehat\xi^{\mathrm{mv}}_n,\\
		{\widehat V}^{\mathrm{mv}}_{n+1} = {\widehat V}^{\mathrm{mv}}_{n} + \big\{\sum^m_{i=1} \left[\mu^i_n -r_n\right]  \widehat\xi^{\mathrm{mv},i}_n\overline{S}^i_n \big\} \Delta t  + \sum^m_{i=1} \widehat\xi^{\mathrm{mv},i}_n \overline{S}^i_n \sum^m_{j=1} \sigma^{ij}_n  \Delta W^j_n\vspace{0.2cm},\\
		\hfill{\text{for $n=0,\ldots, N-1$}}.
	\end{cases}
\end{equation*}


\subsubsection{Deep local risk minimization}\label{sec:deepLRM}
\label{sec:deepLR}
The approximation of local risk minimizing strategies by means of the deep solver poses less challenges with respect to the mean-variance hedging approach. In fact, thanks to Proposition \ref{prop:FSBSDE}, we know that to find the F\"ollmer-Schweizer decomposition we need to solve the linear BSDE \eqref{eq:localriskX}. By discretizing the backward dynamics forward in time by the Euler-Maruyama scheme, we get 
\begin{align}
\begin{cases}
\overline{{X}}^{\mathrm{lr}}_{{n+1}} =  \overline{{X}}^{\mathrm{lr}}_{n} +\frac{1}{S^0_n}\overline{\eta}^{\mathrm{lr}, \top}_{1,n}\phi_n\Delta t + \frac{1}{S^0_n}\overline{\eta}_{n}^{\mathrm{lr}, \top}\Delta \mathcal{W}_{n},\qquad \text{for $n=0,\ldots, N-1$},\\
 \overline{X}^{\mathrm{lr}}_0 = y^{\mathrm{lr}}_X,
\end{cases}
\end{align}
where $\overline{\eta}^{\mathrm{lr}, \top}_{n} = \left(\overline{\eta}^{\mathrm{lr}, \top}_{1,n}, \overline{\eta}^{\mathrm{lr}, \top}_{2,n}\right)$ and $y_X^{\mathrm{lr}}\ge0$ is the initial value to be determined. We introduce a family of ANNs $\mathcal{N}^{\eta^{\mathrm{lr}}}_n:\mathbb{R}^{q}\to\mathbb{R}^{q}$, $n=0,\ldots,N-1$, of the form
\begin{align*}
\mathcal{N}^{\eta^{\mathrm{lr}}}_n=\left(\begin{array}{c}
\mathcal{N}^{\eta^{\mathrm{lr}}}_{1,n} \\ 
\mathcal{N}^{\eta^{\mathrm{lr}}}_{2,n}
\end{array} \right)
\end{align*}
with $\mathcal{N}^{\eta^{\mathrm{lr}}}_{1,n}:\mathbb{R}^{q}\to \mathbb{R}^m$ and $\mathcal{N}^{\eta^{\mathrm{lr}}}_{2,n}:\mathbb{R}^{q}\to \mathbb{R}^d$, and consequently the approximation
\begin{align}
	\label{eq:localriskXdisc}
\begin{cases}
\widehat{{X}}_{{n+1}}^{\mathrm{lr}} =  \widehat{{X}}_{n}^{\mathrm{lr}} +\frac{1}{S^0_n}\left(\mathcal{N}^{\eta^{\mathrm{lr}}}_{1,n}(\overline{\mathcal X}_n)\right)^\top\phi_n\Delta t + \frac{1}{S^0_n}\left(\mathcal{N}^{\eta^{\mathrm{lr}}}_n(\overline{\mathcal X}_n)\right)^\top\Delta \mathcal{W}_{n},\qquad \text{for $n=0,\ldots, N-1$},\\
\widehat{X}^{\mathrm{lr}}_0 = y_X^{\mathrm{lr}}.
\end{cases}
\end{align}
The associated stochastic control problem \eqref{eq:min} then becomes
\begin{align}\label{eq:ocpXLR}
\underset{y_X^{\mathrm{lr}},\mathcal{P}\left((\mathcal{N}^{\eta^{\mathrm{lr}}}_n)_{n=0}^{N-1}\right)}{\text{minimise}} \; \E\left[ \left| \tilde{H} - \widehat{X}^{\mathrm{lr}}_{N}\right|^2\right],
\end{align}
where $\mathcal{P}((\mathcal{N}^{\eta^{\mathrm{lr}}}_n)_{n=0}^{N-1})$ is the set of parameters of the ANNs. We denote by $\widehat{y}_X^{\mathrm{lr}}$ and $(\widehat{\mathcal{N}}^{{\eta^{\mathrm{lr}}}}_n)_{n=0}^{N-1}$ the optimal solution to the problem above, and with $(\widehat{X}^{\mathrm{lr}}_n)^N_{n=0}$ the discrete time process obtained from \eqref{eq:localriskXdisc} by substituting $y_X^{\mathrm{lr}} = \widehat y_X^{\mathrm{lr}}$ and ${\mathcal{N}}^{{\eta^{\mathrm{lr}}}} = \widehat{\mathcal{N}}^{{\eta^{\mathrm{lr}}}}$, and which approximates the solution to \eqref{eq:localriskX}.

The approximated optimal trading strategy can then be  computed from \eqref{eq:xiLR}:
\begin{equation}
	\begin{cases}
		\widehat \xi^{\text{lr}}_n = \frac{1}{S^0_n}\left(\operatorname{diag}(\overline{S}_n)\sigma_n\right)^{-1} \widehat{\mathcal{N}}^{\eta^{\mathrm{lr}}}_{1,n}\left(\overline{\mathcal X}_n\right),\\
		\widehat \psi^{\text{lr}}_n = \widehat{X}^{\mathrm{lr}}_n - \operatorname{diag}(\overline{S}_n) \widehat \xi^{\text{lr}}_n ,\\
		\hfill{\text{for $n=0,\ldots, N-1$}}.
	\end{cases}
\end{equation}
This concludes the treatment of the deep BSDE solver approach for quadratic hedging. 

\subsubsection{Comparison with the Deep Hedging approach of \cite{bgtw2019}}
We conclude this section with some remarks allowing us to put the above proposed methodology in context. The approach which is the closest to our procedure is the so-called \textit{deep hedging} developed in \cite{bgtw2019}. Here the authors consider the profit and loss (PL) of a trader subject to transaction costs. In our notation, this would be given by 
\begin{align*}
    \mathrm{PL}_T\left(\tilde{H}, y, \xi \right):=-\tilde{H}+y+(\xi \cdot \tilde S)_T -\mathcal{C}_T(\xi),
\end{align*}
where $(\xi \cdot \tilde S)_T$ denotes the discrete-time stochastic integral representing the gains and losses from trading according to the strategy $\xi$, and $\mathcal{C}_T(\xi)$ represents, e.g., transaction costs. As their market model is incomplete, they resort to indifference pricing in order to determine the initial price $y$. Given a convex risk measure $\varrho$, they consider the following minimization problem with respect to a random variable $X$, representing the financial payoff:
\begin{align*}
    \pi(X):=\inf _{ \xi\in \mathcal{A}^{\text{ip}}} \varrho\left(X+(\xi \cdot \tilde S)_T-\mathcal{C}_T(\xi)\right),
\end{align*}
where $\mathcal{A}^{\mathrm{ip}}$ is the set of admissible strategies for their indifference pricing setting. The price $y$ is defined as the amount of cash that the trader needs to charge in order to be indifferent between $-\tilde H$ and $0$, i.e.  $y=\pi(-\tilde H)-\pi(0)$. The strategies are directly parametrized by means of artificial neural networks. Choosing the expectation of the square of $\mathrm{PL}_T(\tilde{H}, y, \xi)$ as the convex risk measure $\varrho$  brings the methodology close to ours.  There is however a significant difference: the utility indifference price is available only at time zero as a real number. Hence it is not possible to obtain the dynamic properties of the price process or simulate future scenarios of the price which are needed for the computing, e.g., CVA. Our approach to local risk minimization and mean-variance hedging allows to obtain the BSDEs that are satisfied by the price process. We consider this as a useful feature for risk management applications.





{\color{black}
\section{Convergence analysis}
We discuss in this section the convergence of our method. To do this, we start by writing the linear BSDEs \eqref{hedgingproblemX} and  \eqref{eq:localriskX} in the general form 
\begin{align}\label{eq:linearBSDE}
	\begin{cases}
		\dd \tilde{X}_t = \Theta^\top_t \tilde \eta_{t} \dd t +\tilde\eta^{\top}_{t} \dd {\mathcal W}_t,\\
		\tilde{X}_T = \tilde{H},
	\end{cases}
\end{align}
where $\mathcal W = (W, B)^\top\in \R^q$, $\tilde \eta := \eta/S^0 \in \R^q$, and $\Theta \in \R^q$ is defined as $\Theta := (\phi, -\Lambda_2/L)^\top$ in the case of \eqref{hedgingproblemX} and $\Theta := (\phi, \mathbf{0})^\top$ in the case of \eqref{eq:localriskX}, with $\mathbf{0}\in\R^d$ a vector of $0$'s. Notice that equation \eqref{eq:linearBSDE} is a linear BSDE with unbounded coefficients due to the definition of $\Theta$. Moreover, both the process $\Theta$ and the terminal condition $\tilde H$ depend on the risk factors (i.e. on the forward process $\mathcal X =(\tilde S^1, \ldots, \tilde S^m, Y^1,\ldots, Y^d)$) and, in the mean-variance case, $\Theta$ also depends on the solution of the BSRE \eqref{eq:StochasticRiccati}. To ease the notation, we remove for this section  the $\,\tilde{\ }\,$'s appearing in \eqref{eq:linearBSDE}, and directly write the solution and the terminal condition as $(X,\eta)$ and $H$, respectively.

We notice that, when the deep solver is applied to \eqref{eq:linearBSDE}, the input process $\Theta$ and the terminal condition $H$ are replaced by the piecewise-constant approximation $\bar \Theta$ and $\bar H$, respectively. In particular, $\bar \Theta$ is defined starting from the Euler approximation of $\mathcal X$ and, in the mean-variance case, also from the solution $(\bar L, \bar \Lambda_2)$  provided by the solver applied to equation \eqref{eq:StochasticRiccati}. The terminal condition $\bar H$ is defined starting from the Euler approximation of $\mathcal X$. 

The first result that we present estimates then the error introduced on the solution to \eqref{eq:linearBSDE} when one approximates  $(\Theta, H)$ by means of $(\bar \Theta, \bar H)$. In what follows, we consider the following Dol\'eans-Dade exponentials:
\begin{equation}
    \label{eq:Dexp}
    \mathcal E_t := \exp\left(-\int^t_0  \frac{| \Theta_s|^2}{2}\dd s + \Theta_s^ \top \dd \mathcal W_s\right),\quad \text{and} \quad \bar{\mathcal E}_t := \exp\left(-\int^t_0 \frac{|\bar \Theta_s|^2}{2}\dd s + \bar \Theta_s^ \top \dd \mathcal W_s\right).
\end{equation}

We  work under the following assumption.

\begin{assumption}\label{ass:convNew}
{$H, \bar H\in L^2_{\F_T}(\Omega;\R)$} and $\Theta, \bar \Theta$ are two progressively-measurable square-integrable processes such that $\int^t_0 \Theta^ \top_s \dd\mathcal W_s$ and $\int^t_0 \bar \Theta^ \top_s \dd\mathcal W_s $ are bounded mean oscillation (BMO) martingales, i.e.  it holds almost surely that
$$
\sup_{\tau\in \mathcal T_{[0,T]}}\E\left[\int^T_\tau |\Theta_s|^ 2 \dd s \Big| \mathcal F_\tau \right] <\infty,\quad \text{and} \quad \sup_{\tau\in \mathcal T_{[0,T]}}\E\left[\int^T_\tau |\bar\Theta_s|^ 2 \dd s \Big| \mathcal F_\tau \right] <\infty, 
$$
{where $\mathcal T_{[0,T]}$ denotes the family of $\mathbb F$-stopping times valued in $[0,T]$.}
\end{assumption}

\begin{prop}[Corollary 2.1.3 in \cite{DosReisThesis}]\label{prop:estTheta}
Let Assumptions \ref{ass:convNew} be satisfied, and let {$\mathcal E\in L^{\epsilon}_{\FF}([0,T];\R)$} and {$\bar{\mathcal E}\in  L^{\bar\epsilon}_{\FF}([0,T];\R)$}, for some $\epsilon,\bar \epsilon>1$. Then, for a given $p>1$ such that $\frac{1}{\max(\epsilon,\bar \epsilon)} + \frac1p =1$, there exists a constant $C\geq 0$ such that 
\begin{align*}
&\E\left[\sup_{t\in [0,T]}  |X_t - \bar X_t|^2 \right] 
\leq C \left\{ \E\left[ | H - \bar H|^{2p^ 2} \right]^ {1/{p^ 2}}+ \E\left[ \left(\int^ T_0 |\Theta_s - \bar \Theta_s|^ 2 \dd s\right)^ {2  p^ 2} \right]^ {1/{2p^ 2}}\right\}. 
\end{align*}
\end{prop}

\begin{oss}
    When dealing with local risk minimization, since $\Theta = (\phi, \mathbf{0})^\top$, the  process  $|\Theta - \bar \Theta|$ appearing on the right-hand side of the estimate above only depends on the difference between the market price of risk $\phi$ and its approximation obtained after the Euler discretization of the market dynamics. In the mean-variance framework, an additional contribution to the error comes from  the difference between $-\frac{\Lambda_2}{L}$ and $-\frac{\bar \Lambda_2}{\bar L}$, where $(\bar L, \bar \Lambda_2)$ comes from the deep solver solution to the BSRE \eqref{eq:StochasticRiccati}. This motivates the fact that in the numerical experiments in Section \ref{sec:numerics}, the relative errors in Table \ref{table:resultsMV} for the mean-variance hedging are overall higher than those in Table 
    \ref{table:resultsLR} for the local risk minimization.
\end{oss}

The second step is to discuss the solver approximation error. Error bounds for the deep BSDE solver have been established in \cite{HanLon18} under the assumption of Lipschitz continuity of the driver of the BSDE. Here, we will make use of such results applied to a suitable transformation of the  BSDE. Considering then equation \eqref{eq:linearBSDE} with $(\Theta, \tilde H)$ replaced by $(\bar \Theta, \bar H)$, we aim at providing estimates for the difference between $\bar X$, the exact solution, and $\widehat X$, the solution obtained with the solver.
In particular, one can show that $\bar X$ can be written as $\bar X_t = \bar{\mathcal E}^{-1}_t Z_t$, with $\bar{\mathcal E}$ defined in equation \eqref{eq:Dexp} and $Z$ (together with the control process $\zeta$) solution to the linear BSDE with no drift
\begin{align}\label{eq:linearBSDE2}
	\begin{cases}
		\dd Z_t = \zeta^{\top}_{t} \dd {\mathcal W}_t,\\
		Z_T = \bar{\mathcal E}_T\bar {H}.
	\end{cases}
\end{align} 

We denote by  $(\widehat Z, \widehat \zeta)$ the solution to  \eqref{eq:linearBSDE2} provided by the deep BSDE solver and extended by piecewise constant interpolation to the whole time interval $[0, T]$. 
Then, one clearly has that
\begin{align}
\E\left[|\bar X_t  - \widehat X_t| \right]& \leq   \E\left[\bar{\mathcal E}^{-1}_t |Z_t - \widehat Z_t|\right] + \E\left[\bar {\mathcal E}^{-1}_t |\widehat Z_t -  \bar{ \mathcal E}_t\widehat X_t| \right]\nonumber\\
& \leq \E\left[\bar{\mathcal E}_t^{-2}\right]^{1/2}\E\left[|Z_t - \widehat Z_t|^2\right]^{1/2} + \E\left[\bar{\mathcal E}^{-2}_t\right]^{1/2} \E\left[|\widehat Z_t -  \bar{\mathcal E}_t\widehat X_t|^2 \right]^{1/2}.\label{eq:XtoZ}
\end{align}

The quantity $\E[\bar{\mathcal E}^{-2}]$ in \eqref{eq:XtoZ} can be shown to be bounded under suitable Kazamaki's type conditions.

\begin{assumption}\label{ass:Expbound}
There exists $\beta> 2+\sqrt{2}$ such that 
$$
\sup_{t\in [0,T]}\E\left[\exp\left(\frac{\beta}{2} \int^t_0 \bar \Theta_s^ \top \dd \mathcal W_s\right)\right]<\infty,\quad \text{and} \quad \sup_{t\in [0,T]}\E\left[\exp\left(\frac{2\beta}{\beta^2 - 4 \beta+2} \int^t_0 |\bar \Theta_s|^2 \dd s\right)\right]<\infty.
$$
\end{assumption}

\begin{lem}\label{lemexp}
   Let  Assumption \ref{ass:Expbound}  be satisfied. Then we have that
   $$
   \E\left[\sup_{t\in [0,T]}\bar{\mathcal E}^{-2}_t\right]<\infty.
   $$
\end{lem}
\begin{proof}
    The result follows by \cite{GashiLi2} by observing that, under Assumption \ref{ass:Expbound},  Theorem 2.6  in \cite{GashiLi2} ensures the boundedness of the second moment of $\bar{\mathcal E}^{-1}$.
\end{proof}
Moreover, from  \cite[Theorem 1]{HanLon18} there exists a constant $C\geq 0$, such that
\begin{equation}\label{hanlong}
\sup_{t\in [0,T]} \E\left[|Z_t - \widehat Z_t|^2\right] \leq C \left(\Delta t + \E\left[|\bar  H -\widehat Z_N|^2\right]\right),
\end{equation}
with $\Delta t$ the time step for the Euler discretization, and $\widehat Z_N$ the value of the approximated solution at the terminal time. 
By putting together Lemma \ref{lemexp} with \eqref{eq:XtoZ} and \eqref{hanlong}, we obtain the following result providing an a posteriori error estimate.

\begin{prop}\label{prop:covSolver}
   Let  Assumption \ref{ass:Expbound}  be satisfied. Then,  there exists a constant $C\geq 0$ such that 
   $$
   \sup_{t\in [0,T]}\E\left[|\bar X_t  - \widehat X_t| \right] \leq C \left(\Delta t^{1/2} + \E\left[|\bar  H -\widehat Z_N|^2\right]^{1/2} + \sup_{t\in [0,T]}\E\left[|\widehat Z_t -  \bar{\mathcal E}_t\widehat X_t|^2 \right]^{1/2}\right).
   $$
\end{prop}

From Proposition \ref{prop:estTheta} and Proposition \ref{prop:covSolver}, we finally derive the theorem below.

\begin{thrm}\label{teo:conv}
Let Assumptions \ref{ass:convNew} and \ref{ass:Expbound} be satisfied, and let {$\mathcal E\in L^{\epsilon}_{\FF}([0,T];\R)$} and {$\bar{\mathcal E}\in L^{\bar\epsilon}_{\FF}([0,T];\R)$}, for some $\epsilon,\bar \epsilon>1$. Then, for a given $p>1$ such that $\frac{1}{\max(\epsilon,\bar \epsilon)} + \frac1p =1$, there exists a constant $C\geq 0$ such that 
\begin{align*}
   \sup_{t\in [0, T]} \E\left[| X_t  - \widehat X_t| \right] \leq & \,C \Bigg(\Delta t^{1/2} + \E\left[|\bar  H -\widehat Z_N|^2\right]^{1/2} + \sup_{t\in [0, T]}\E\left[|\widehat Z_t -  \bar{\mathcal E}_t\widehat X_t|^2 \right]^{1/2}\\
   & + \E\left[ | H - \bar H|^{2p^ 2} \right]^ {1/{2 p^ 2}}+ \E\left[ \left(\int^ T_0 |\Theta_s - \bar \Theta_s|^ 2 \dd s\right)^ {2  p^ 2} \right]^ {1/{4p^ 2}}\Bigg).
\end{align*}
\end{thrm}

\begin{oss}
The performed convergence analysis shows a bound on the error that depends on the time discretization step, on the quantities estimating the goodness of the optimization routine and on the neural network explanatory power. As a consequence, a higher number of time steps does not automatically lead to a smaller error as the leading contribution to the error might also come from the second term. Notice also that having a neural network for each time step means that, for higher numbers of time steps, the optimization problem to solve is more complex. Looking at the numerical results in Section \ref{sec:numerics}, we then observe that for the local risk minimization (Table \ref{table:resultsLR}) the error decreases with the number of time steps, while in the mean-variance case (Table \ref{table:resultsMV}) for higher numbers of time steps the error may also be higher. This is due to the fact that the solver is applied twice in the latter case, hence the accumulation of the optimization error has a higher impact.
\end{oss}

}

\section{Numerical experiments}
\label{sec:numerics}
We present here the numerical results for mean-variance hedging and local risk minimization under the multidimensional Heston model introduced in Section \ref{sec:multiHeston}, by means of the deep BSDE solver presented in Section \ref{sec:DeepQuadraticHedging}. The code for the experiments is available at \url{https://github.com/silvialava/Deep_quadratic_hedging.git}.

Observe that for the multidimensional Heston model \eqref{d_model}, according with the notation in Section \ref{sec:paramSolver}, the forward process $\mathcal{X}$ is of the form $\mathcal{X} = (\tilde S, Y^2)$ with coefficients
\begin{align*}
	a(t,(\tilde S_t, Y_t^2))= 
	\begin{pmatrix}
		\mathrm{diag}(\tilde S_t)A\,\mathrm{diag}\left(Y_t\right)& \mathbf{0} \\
		\mathrm{diag}(\sigma)\mathrm{diag}\left(Y_t\right)\mathrm{diag}(\rho)  & \mathrm{diag}(\sigma)\mathrm{diag}\left(Y_t\right)\mathrm{diag}(\sqrt{\mathbbm{1}-\rho^2})
	\end{pmatrix}
\end{align*} 
where $\mathbf{0} \in\R^{m\times m}$ is a matrix of $0$'s, and
\begin{align*}
	b(t,(\tilde S_t, Y_t^2))=
	\begin{pmatrix}
		\mathrm{diag}(\tilde S_t)A \,\mathrm{diag}(Y^2_t)\bar{\mu}\\
		\mathrm{diag}(\kappa)\left(\theta- Y^2_t\right)
	\end{pmatrix}.
\end{align*}

Given a portfolio of $m\ge 1$ risky assets, for a strike price $K$ and a terminal time $T$, we aim at hedging and pricing a European-type call option, whose payoff function $g: \R_+^m \to \R_+$ is of the form
\begin{equation*}
	g(x) := \max \left(\sum_{i=1}^m x_i - mK, 0\right), \quad \mbox{for } x\in\R_+^m \mbox{ with } x = (x_1, \dots, x_m)^\top,
\end{equation*}
so that the discounted contingent claim $\tilde H$ becomes
\begin{equation}
	\label{contclaim}
	\tilde H = e^{-\int_0^T r_s\dd s}\,\max \left(\sum_{i=1}^m  S_T^i - mK, 0\right).
\end{equation}

In the one-dimensional case \eqref{HestonModel}, semi-explicit solutions can be computed by following \cite{cerny2008} for the mean-variance hedging and \cite{heath2001numerical} for the local risk minimization. We briefly present the two approaches in Appendix \ref{sec:benckmarks} as we shall use them as benchmarks for the deep BSDE solver. Both \cite{cerny2008} and \cite{heath2001numerical} rely on a two-dimensional partial differential equation (PDE) which we solve by adapting \cite{in2010adi}, see Appendix \ref{sec:PDEsolver} for details. These allow us to compare the entire contingent claim price path, as well as the trading strategies paths in $[0, T]$.

For higher values of $m$ (indicatively $m\geq 2$), however, solving these PDEs is not numerically feasible due to the curse of dimensionality. To test the accuracy of the deep BSDE solver one can alternatively perform a change of measure as illustrated, respectively, in \eqref{changeQmv} and \eqref{changeQlr}, and estimate the contingent claim price via Monte Carlo simulations under the variance optimal martingale measure and under the minimal martingale measure, respectively.

However, from equations  \eqref{changeQmv} and \eqref{eq:GirsanovKernelB}, we observe that the Radon-Nikodym derivative for the change of measure from $\p$ to the optimal martingale measure $\Q_{{\text{mv}}}$ depends on the solution of the stochastic Riccati equation \eqref{eq:StochasticRiccati}, hence it depends on the optimal solution that is found with the deep solver, see Section \ref{sec:deepMV}. Using this approximated density for Monte Carlo simulations under $\Q_{{\text{mv}}}$ would then lead to biased results.

To overcome this issue, we shall consider for all the experiments a model where the matrix $A$ is diagonal. As pointed out in Remark \ref{oss:H}, this leads to an $\tilde S$ being the superposition of mutually independent Heston models. We then apply component-wise to the vector processes $\tilde S$ and $Y^2$ the change of measure proposed by \cite[Section 4]{cerny2008} for the one-dimensional Heston model, which is presented in equation \eqref{HestonQ}, and we run the Monte Carlo routine by simulating $\tilde S$ and $Y^2$ directly under the variance optimal martingale measure $\Q_{{\text{mv}}}$.

On the other hand, for the minimal martingale measure $\Q_{\text{lr}}$, the Radon-Nikodym derivative \eqref{changeQlr} only depends on the market price of risk $\phi$. Then either we simulate under $\p$ and multiply by the density process \eqref{changeQlr}, or we simulate directly under $\Q_{\text{lr}}$ by performing the corresponding change of measure, namely
\begin{equation}
	\begin{aligned}
		&\dd W_t^{\text{lr}} = \dd W_t + \phi_{t}\dd t,\\
		&\dd B_t^{\text{lr}} = \dd B_t,
	\end{aligned}
\end{equation}
which, in the multidimensional Heston model \eqref{d_model}, leads to
\begin{equation}
	\label{HestonQLR}
	\begin{cases}
	\dd \tilde S_t 
	= \mathrm{diag}(\tilde S_t) A\,\mathrm{diag}\left(Y_t\right)\dd W_t^{\text{lr}},\\
	\dd Y^2_t = \mathrm{diag}(\kappa)\left(\theta- Y^2_t\right)\dd t - \mathrm{diag}(\sigma)\mathrm{diag}\left(Y_t\right)\mathrm{diag}(\rho)\phi_t \dd t
	\\ \qquad\qquad \qquad +\mathrm{diag}(\sigma)\mathrm{diag}\left(Y_t\right)\left(\mathrm{diag}(\rho)\dd W_t^{\text{lr}}  + \mathrm{diag}(\sqrt{\mathbbm{1}-\rho^2})\dd B_t^{\text{lr}} \right).
\end{cases}
\end{equation}

Without loss of generality, in equation \eqref{eq:driftHeston} we set $\bar{r}\equiv 0$ for all the experiments, so that the risk-free interest rate is $r_t = 0$ for each $t\ge 0$.

\subsection{Setup}
We conduct several numerical experiments on the multivariate Heston model to test the performance of the deep BSDE solver for quadratic hedging: both model and solver configuration details are in Table \ref{conf}, where $\mathbbm{1}\in \R^m$ denotes a vector of $1$'s. 

In particular, we consider four different portfolio dimensions, respectively $m=1, 5, 20, 100$. As we already mentioned above, the matrix $A$ is taken to be diagonal: this allows us to simulate under the variance optimal martingale measure by performing for each risky asset a one-dimensional change of measure as in \eqref{HestonQ}. Moreover, we consider all the model parameters to be constant vectors, so that also the one-dimensional case fits smoothly into the setting. We fix the strike price to $K=100.0$ and the terminal time is $T=1.0$.

We consider three different discretization grids, respectively $N=10$, $N=50$ and $N=100$, which means we have $N=10, 50, 100$ neural networks that must be trained. Each of them is built with $4$ inner layers, where each layer has a number of nodes depending on the portfolio dimension, namely $2m+20$, and with the rectified linear unit (ReLU) as activation function. For the training of the neural networks, we set to $8000$ the total number of stochastic gradient descent iterations with initial learning rate equal to $5\cdot 10^{-2}$. After $4000$ iterations (Partial iterations in the table), the learning rate is reduced to $5\cdot 10^{-3}$ to improve the convergence of the method.

Finally, we need to define an initial guess for $y_L$ and $y_X^{\text{mv}}$ for the deep mean-variance hedging, see Section \ref{sec:deepMV}, and an initial guess for $y_X^{\text{lr}}$ for the deep local risk minimization, see Section \ref{sec:deepLR}. Since we know that $0< L\le 1$, we set the range for $y_L$ to $[0.5, 2.0]$, while for $y_X$ we consider, respectively, the $95\%$ and the $105\%$ of the corresponding Monte Carlo price simulation (MC in the table).

We stress the fact that both model and solver configurations are kept constant in all the experiments, unless explicitly stated. This allows to compare the performance of the BSDE solver consistently. However, we point out that one can aim at improving even further our numerical results by tuning the hyper-parameters of the deep solver in a tailor-made manner for each single experiment.

\begin{table}[tp]
	\centering
	\begin{tabular}{lrrlr}
		\toprule
		\multicolumn{2}{l}{\large \textbf{Model configuration}} &&\multicolumn{2}{l}{\large \textbf{Deep solver configuration}}\\
		\midrule
		$m$ & $1, 5, 20, 100$& &$N$ & $10, 50, 100$\\
		$A$&$\mathrm{diag}(\mathbbm{1})$&&Number of layers& $4$\\
		$\bar{\mu}$&$0.1\cdot \mathbbm{1}$&&Number of nodes& $2m+20$\\
		$\kappa$&$0.5 \cdot \mathbbm{1}$&&Activation function&ReLU\\
		$\theta$&$0.05\cdot \mathbbm{1}$&&Total iterations&$8000$\\
		$\sigma$&$0.1\cdot \mathbbm{1}$&&Partial iterations&$4000$\\
		$\rho$&$-0.45\cdot \mathbbm{1}$&&Initial learning rate&$5\cdot 10^{-2}$\\
		$S_0$&$100.0\cdot \mathbbm{1}$&&Second learning rate& $5\cdot 10^{-3}$\\
		$Y^2_0$&$0.025\cdot \mathbbm{1}$&&$y_L$ initial range& $[0.5, 2.0]$\\
		$K$&$100.0$&&$y_X$ initial range&MC $\cdot [0.95, 1.05]$\\
		$T$&$1.0$&&Batch size& $128$\\
		\bottomrule
	\end{tabular}
\caption{Model configuration and deep solver configuration for the numerical experiments. 
	\label{conf}}
\end{table}
	
\subsection{Deep mean-variance hedging results}
The numerical results for the deep mean-variance hedging approach are presented in Table \ref{table:resultsMV}, and in
 Figure \ref{MV_logloss} we show the evolution of the logarithmic loss as a function of the number of deep solver iterations. Here the black curves represent the log-loss for the first BSDE and follow the black grid on the left-hand side. The red curves represent the log-loss for the second BSDE and follow the red grid on the right-hand side. 

For each portfolio dimension considered, we compute the Monte Carlo (MC) price by simulating $10^5$ samples with $100$ points of time discretization under the variance optimal martingale measure as presented in equation \eqref{HestonQ}. Moreover, we compute the initial value of the opportunity process $L$ ($L$ value in the table) by means of formula \eqref{definitionLCK} together with equations \eqref{chi0} and \eqref{chi1}. For $t=0$, $L_0$ is indeed purely deterministic and, in the one-dimensional case, the value given by \eqref{definitionLCK} can be considered the exact value for the opportunity process. In the multidimensional case, since the vector $Y^2$ is a superposition of mutually independent processes, we can compute the value of $L$ as the products of the $m$ values that one would get if computing $L$ for each one-dimensional component. In other words, we consider
\begin{equation*}
	L_0 =\prod_{i=1}^{m}\exp\left(\chi_0^i(0)+\chi_1^i(0)Y_0^{2,i}\right),
\end{equation*}
where each $\chi_0^i$ and $\chi_1^i$ is computed as in equations \eqref{chi0} and \eqref{chi1}, for $i=1, \dots, m$.

We then train and run the BSDE solver for solving recursively the two BSDEs as in Section \ref{sec:deepMV}, obtaining, respectively, an estimate of the initial value of the opportunity process (BSDE solver $L$ value) and an estimate of the call option price (BSDE solver price). For both of them, we report the training time and the relative error, which is computed in the standard way starting from the previously simulated $L$ value and MC price. For the portfolio in dimension $1$, we also compute the option price with the {\v{C}}erný  and Kallsen approach in \cite{cerny2008} (see Section \ref{meanvariance1dim}) and the corresponding relative error.
	
From Table \ref{table:resultsMV} we observe that in all the experiments (except for the two corresponding to $m=100$ with $N=50$ and $N=100$ which we shall comment later) the accuracy for the process $L$ is in the third or forth decimal, and the error for the option price is always below $1\%$. These results are confirmed by the evolution of the logarithmic loss in Figure \ref{MV_logloss}. We observe indeed that the BSDE solver loss for the first equation (black curves) is in average of the order of $10^{-20}$. Different picture appears for the BSDE solver loss for the second equation (red curves) where we observe values of positive order. However, one needs to take into account that the loss is computed in the form of absolute value, and not as a relative value. For the first BSDE, the value of $L$ to be found is between $0$ and $1$, but for the second BSDE the value of $X$ to be found is of the order of $10^2$ (since we take $S_0=K=100.0$) and is expected to grow with the dimension of the problem (because of the contingent claim definition \eqref{contclaim}). For these reasons, we observe the log-loss increasing with $m$ and, most importantly, not converging to $0$ as one may expect.

Exception is made for the two experiments with portfolio dimension $m=100$ and time grid $N=50$ and $N=100$. Here the relative error for the option price is between $1\%$ and $2\%$. Remember that in the mean-variance hedging the second BSDE depends on the solution of the first BSDE. In these two cases, the solver is failing to solve the first BDSE as we can notice from the values obtained for $L$ (BSDE solver $L$ value) which exhibit a significant discrepancy with respect to the true solution. Quite surprisingly, the error for the second BSDE is still relatively low (below $2\%$). 

By observing the log-loss behavior in these two cases, we see that the solver does not converge for the first BSDE. We then rerun the experiments by decreasing the initial learning rate to $1\cdot 10^{-3}$ and the second learning rate to $5\cdot 10^{-4}$. The corresponding results are in Figure \ref{MV_extrafigure}. We see that the new results are in line with all the rest, and the log-loss is now converging as expected.
	
For the portfolio of dimension $1$, we also compute the call option price and hedging strategies evolution, namely units of cash account and shares of risky asset, both with the BSDE approach and with the approach in \cite{cerny2008}. We report them in Figure \ref{MV10}, \ref{MV50} and \ref{MV100}, respectively for $N=10$, $N=50$ and $N=100$. We observe that not only the solver is capturing the option price path, but also provides good approximated hedging strategies. 

{\color{black} Despite it is not feasible to obtain a benchmark for the option price evolution for higher portfolio dimensions, we can however compute the evolution of the process $L$, which we can use as benchmark for the solver solution to the BSRE \eqref{eq:StochasticRiccati}. Indeed, exploiting the particular structure of the Heston model that we chose for the experiments, similarly as for $L_0$ above, we obtain the evolution of $L$ by the product
\begin{equation*}
	L_t =\prod_{i=1}^{m}\exp\left(\chi_0^i(t)+\chi_1^i(t)Y_t^{2,i}\right),
\end{equation*}
where each $\chi_0^i$ and $\chi_1^i$ is computed as in equations \eqref{chi0} and \eqref{chi1}, for $i=1, \dots, m$.
Hence, for all the portfolio dimensions considered, we report in Figure \ref{plot:L} the approximated solution to the BSRE \eqref{eq:StochasticRiccati} obtained with the deep solver against the exact solution estimated with the formula above. Here we notice that the solver solution well captures the behaviour of the BSRE dynamics for the given paths.}
	
To have a clearer picture of the solver performance when the number of time steps is increased, we report in Figure \ref{MSE} the Mean Squared Error (MSE) as a function of time for the option price, for the units of cash account and for the shares of risky asset in dimension one. For the option price, we observe that the MSE is increasing with time. However, we also see a clear improvement when increasing the number of time steps, $N$. This means that, despite the error is accumulating over time, one can still control it by decreasing the mesh size in the time discretization.

We do not observe the same clear behavior for the units of cash account and the shares of risky asset. However, we point out that, while the optimization is done to approximate the BSDE, hence to approximate the option price process, the strategies are in some sense a by-product of such optimization. So the fact that the MSE in these two cases is not monotone is to be expected. For these reasons, we also report as a dashed line the mean over time of the MSE. Here we observe that the mean of MSE decreases from $N=10$ to $N=50$, while the line corresponding to $N=100$ is not visible because overlapping with the $N=50$ one. For the sake of completeness, we report here their values: for the units of cash accounts, the mean MSE for $N=50$ is $5.33\cdot 10^{-4}$ and the mean MSE for $N=100$ is $6.73\cdot 10^{-4}$;  for the  shares of risky asset, the mean MSE for $N=50$ is $4.92\cdot 10^{-4}$ and the mean MSE for $N=100$ is $6.44\cdot 10^{-4}$. Hence it gets slightly worse for $N=100$. We point out that when $N$ increases one has a higher number of ANNs to train, so that  the same number of iterations used for lower values of $N$  may be not sufficient to guarantee an improvement of the results.

\begin{table}[tp]
	\centering
	\resizebox{0.65\textwidth}{!}{\begin{tabular}{rrrr}
		\multicolumn{4}{c}{\huge \textbf{Mean-variance hedging}}\\
		&&& \\
		\toprule
		\multicolumn{1}{r}{\large \textbf{Portfolio dimension:} $\boldsymbol{1}$}&\multicolumn{2}{r}{\large \textbf{MC price:}}& $\boldsymbol{6.837}$\\
		&\multicolumn{2}{r}{\large \textbf{$\boldsymbol{L}$ value:}}& $\boldsymbol{0.99984}$\\
		\midrule
		\textbf{Time steps}&$\boldsymbol{10}$ & $\boldsymbol{50}$&$\boldsymbol{100}$ \\
		\midrule
		\textbf{BSDE solver $\boldsymbol{L}$ value}& $0.99969$& $0.99970$& $0.99969$ \\
		\textbf{Relative error (\%)}& $0.01476$&  $0.01434$&$0.01493$ \\
		\textbf{1st training time (s)}&  $82$& $576$& $1048$\\
		\hdashline
		\textbf{BSDE solver price}& $\boldsymbol{6.830}$& $\boldsymbol{6.854}$& $\boldsymbol{6.838}$ \\
		\textbf{Relative error (\%)}& $\boldsymbol{0.105}$&  $\boldsymbol{0.246}$&$\boldsymbol{0.0250}$ \\
		\textbf{2nd training time (s)}&  $1015$& $3270$& $5785$\\
		\hdashline
		\textbf{PDE price}&  $6.853$& $6.853$& $6.853$ \\
		\textbf{Relative error (\%)}& $0.245$&  $0.233$&$0.232$ \\
		\midrule
		\multicolumn{1}{r}{\large \textbf{Portfolio dimension:} $\boldsymbol{5}$}&\multicolumn{2}{r}{\large \textbf{MC price:}}& $\boldsymbol{15.298}$\\
		&\multicolumn{2}{r}{\large \textbf{$\boldsymbol{L}$ value:}}& $\boldsymbol{0.99848}$\\
		\midrule
		\textbf{Time steps}&$\boldsymbol{10}$ & $\boldsymbol{50}$&$\boldsymbol{100}$ \\
		\midrule
		\textbf{BSDE solver $\boldsymbol{L}$ value}& $0.99848$& $0.99848$& $0.99848$ \\
		\textbf{Relative error (\%)}& $0.00028$&  $0.00040$&$0.00019$ \\
		\textbf{1st training time (s)}&  $384$& $2360$& $3830$\\
		\hdashline
		\textbf{BSDE solver price}& $\boldsymbol{15.329}$& $\boldsymbol{15.363}$& $\boldsymbol{15.383}$ \\
		\textbf{Relative error (\%)}& $\boldsymbol{0.201}$&  $\boldsymbol{0.422}$&$\boldsymbol{0.554}$ \\
		\textbf{2nd training time (s)}&  $1702$& $6065$& $10345$\\
		\midrule
		\multicolumn{1}{r}{\large \textbf{Portfolio dimension:} $\boldsymbol{20}$}&\multicolumn{2}{r}{\large \textbf{MC price:}}& $\boldsymbol{30.560}$\\
		&\multicolumn{2}{r}{\large \textbf{$\boldsymbol{L}$ value:}}& $\boldsymbol{0.99393}$\\
		\midrule
		\textbf{Time steps}&$\boldsymbol{10}$ & $\boldsymbol{50}$&$\boldsymbol{100}$ \\
		\midrule
		\textbf{BSDE solver $\boldsymbol{L}$ value}& $0.99401$& $0.99396$& $0.99394$ \\
		\textbf{Relative error (\%)}& $0.00797$&  $0.00335$&$0.00140$ \\
		\textbf{1st training time (s)}&  $1396$& $\textcolor{black}{2494}$& $\textcolor{black}{5106}$\\
		\hdashline
		\textbf{BSDE solver price}& $\boldsymbol{30.612}$& $\boldsymbol{30.788}$& $\boldsymbol{30.800}$ \\
		\textbf{Relative error (\%)}& $\boldsymbol{0.171}$&  $\boldsymbol{0.748}$&$\boldsymbol{0.786}$ \\
\textbf{2nd training time (s)}&  $7704$& $13843$& $28510$\\
		\midrule
		\multicolumn{1}{r}{\large \textbf{Portfolio dimension:} $\boldsymbol{100}$}&\multicolumn{2}{r}{\large \textbf{MC price:}}& $\boldsymbol{68.831}$\\
		&\multicolumn{2}{r}{\large \textbf{$\boldsymbol{L}$ value:}}& $\boldsymbol{0.97002}$\\
		\midrule
		\textbf{Time steps}&$\boldsymbol{10}$ & $\boldsymbol{50}$&$\boldsymbol{100}$ \\
		\midrule
		\textbf{BSDE solver $\boldsymbol{L}$ value}& $0.97044$& $0.12426$& $0.27013$ \\
		\textbf{Relative error (\%)}& $0.02936$&  $87.19$&$72.15$ \\
		\textbf{1st training time (s)}&  $1757$& $9860$& $20917$\\
		\hdashline
		\textbf{BSDE solver price}& $\boldsymbol{68.168}$& $\boldsymbol{69.720}$& $\boldsymbol{67.602}$ \\
		\textbf{Relative error (\%)}& $\boldsymbol{0.964}$&  $\boldsymbol{1.291}$&$\boldsymbol{1.706}$ \\
\textbf{2nd training time (s)}&  $4516$& $21843$& $40253$\\		
		\bottomrule
	\end{tabular}}
	\caption{Mean-variance hedging results for different portfolio dimensions and different number of total time steps in the discretization grid. For each configuration, we compute the Monte Carlo (MC) price by simulating $10^5$ samples under the variance optimal martingale measure, and we use it to compute the relative error in the classical way. For the portfolio with one risky asset, we report the price obtained with the benchmark approach via PDE presented in Appendix \ref{meanvariance1dim}. \label{table:resultsMV}}
\end{table}

\begin{figure}[tp]
	\resizebox{1\textwidth}{!}{
		\begin{tabular}{@{}>{\centering\arraybackslash}m{0.04\textwidth}@{}>{\centering\arraybackslash}m{0.32\textwidth}@{}>{\centering\arraybackslash}m{0.32\textwidth}@{}>{\centering\arraybackslash}m{0.32\textwidth}@{}}
			\multicolumn{4}{c}{\huge \textbf{Mean-variance hedging}}\\
			&&& \\
			&\textbf{\large $\boldsymbol{N = 10}$} & \textbf{\large $\boldsymbol{N = 50}$} & \textbf{\large  $ \boldsymbol{N =100}$}\\
			&&& \\
			\begin{turn}{90}\textbf{\large $\boldsymbol{m = 1}$}\end{turn}&\includegraphics[width=0.31\textwidth, height =0.25\textwidth]{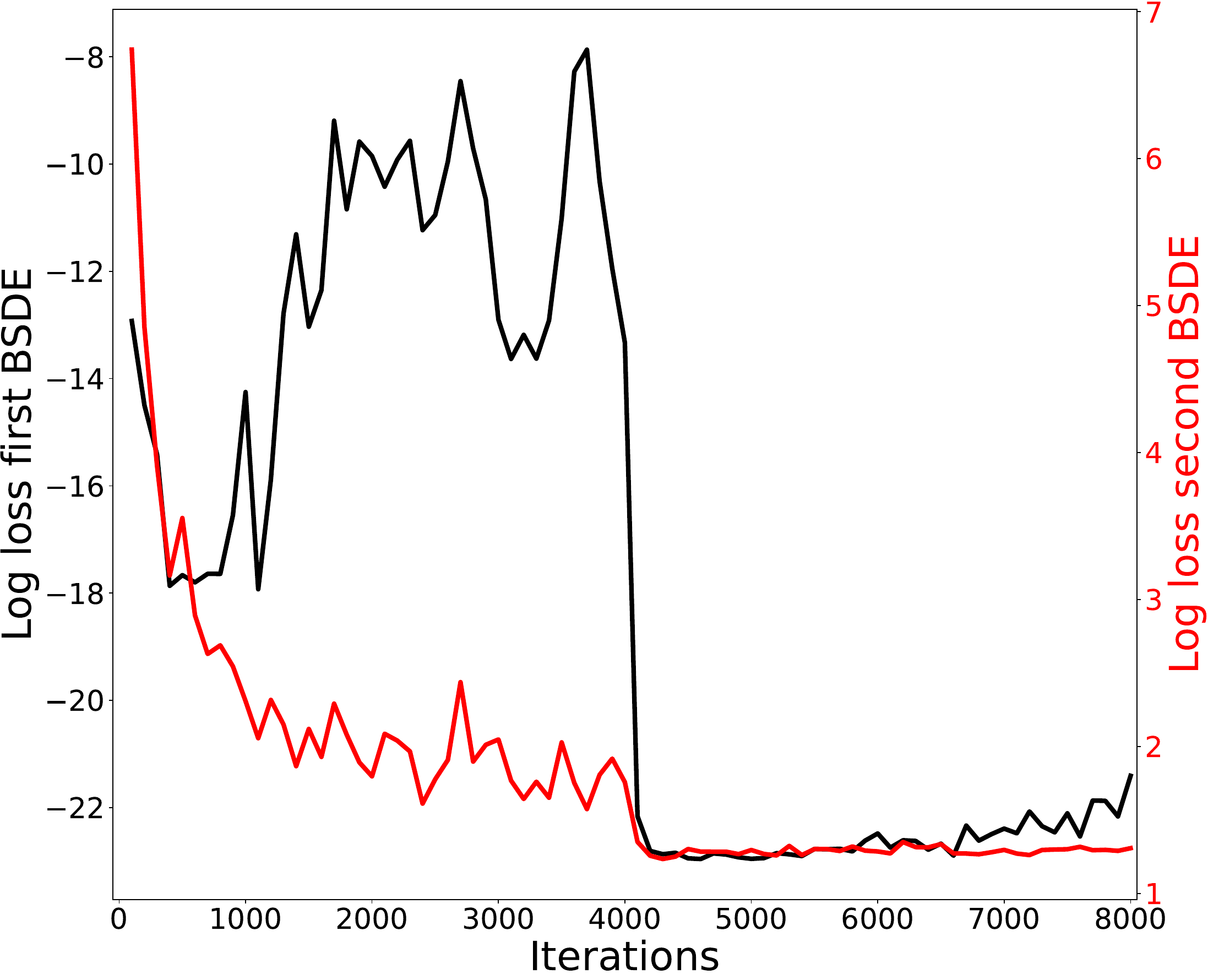} &
			\includegraphics[width=0.31\textwidth, height =0.25\textwidth]{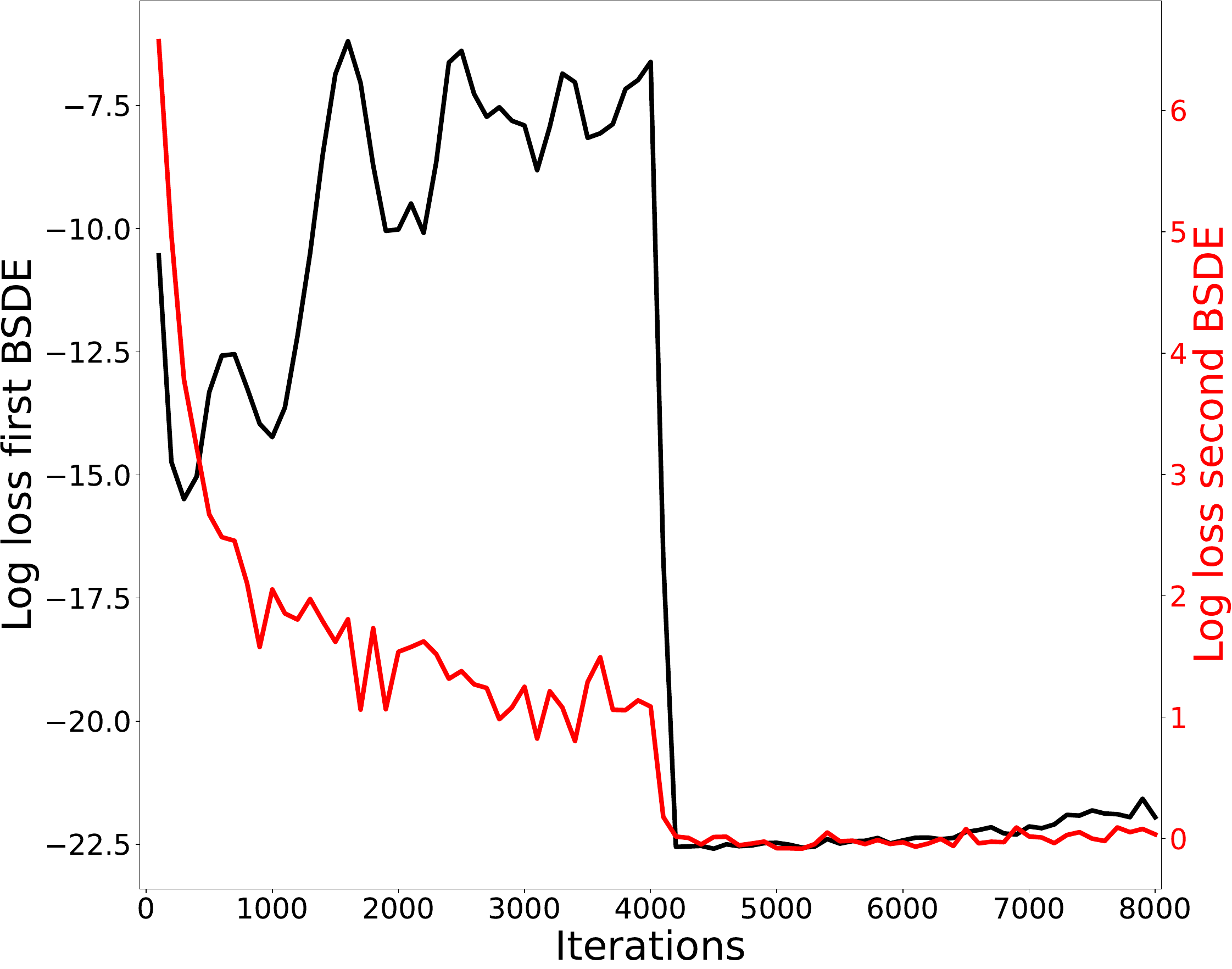}&
			\includegraphics[width=0.31\textwidth, height =0.25\textwidth]{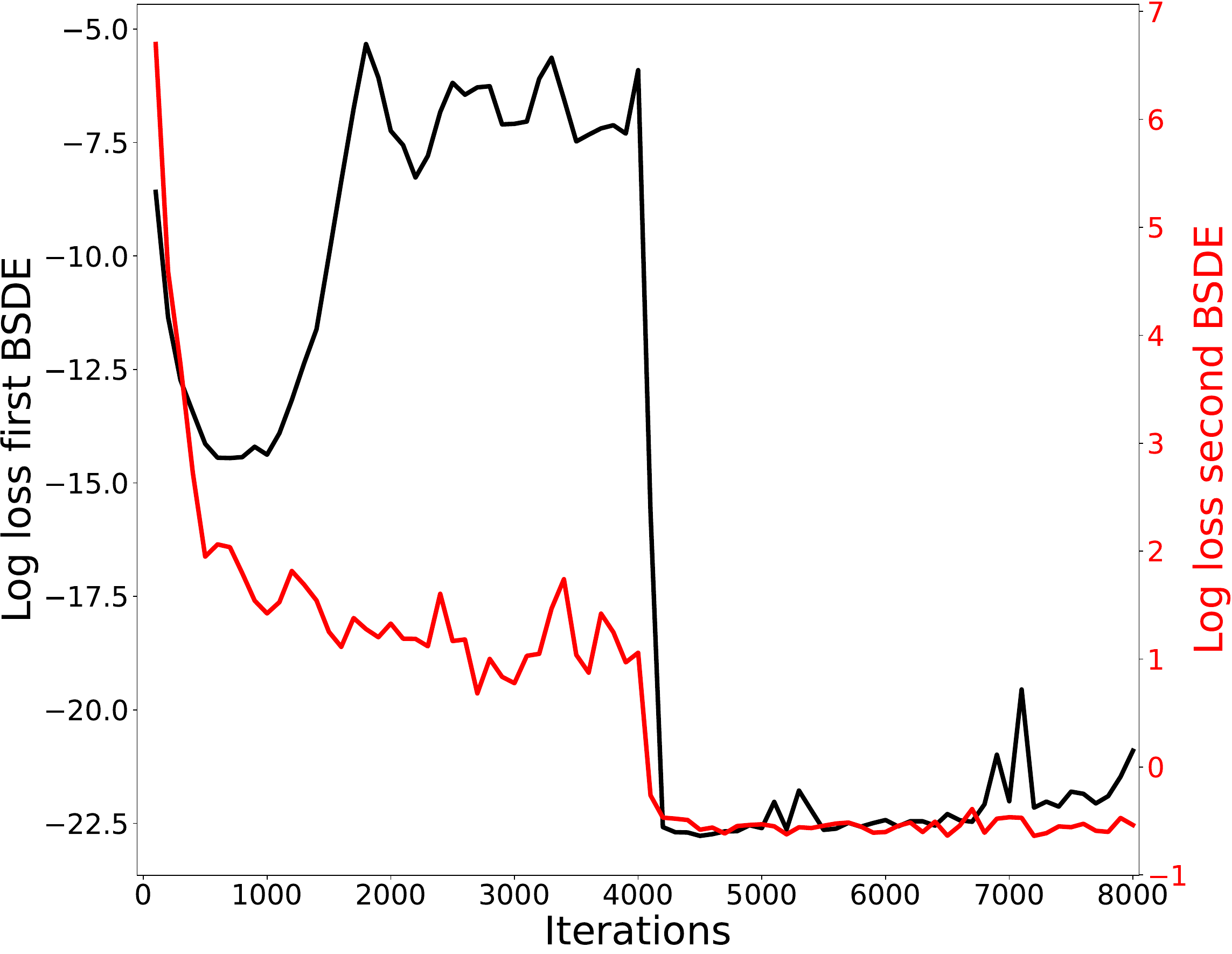} \\	
			\begin{turn}{90}\textbf{\large $\boldsymbol{m = 5}$}\end{turn}&\includegraphics[width=0.31\textwidth, height =0.25\textwidth]{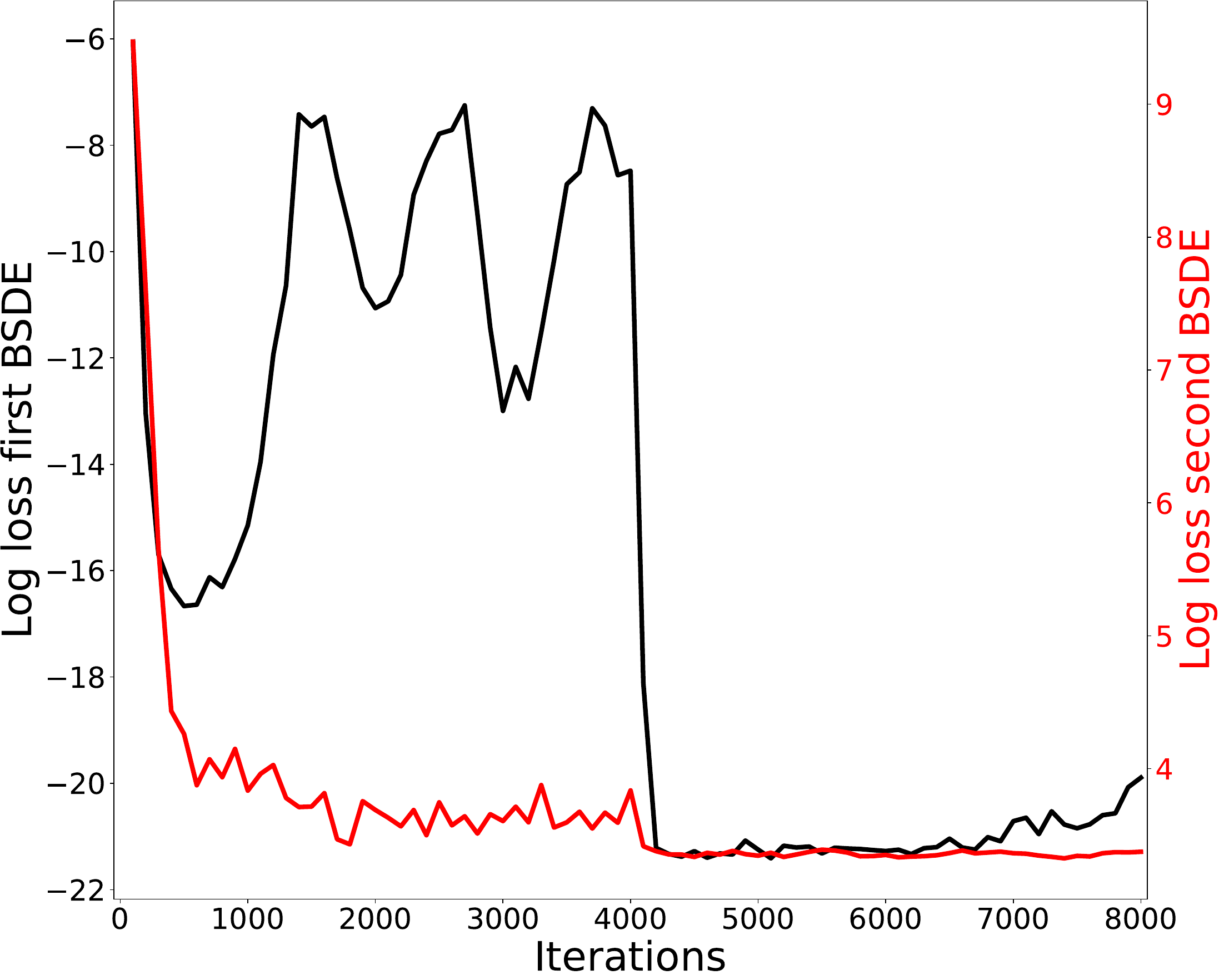} &
			\includegraphics[width=0.31\textwidth, height =0.25\textwidth]{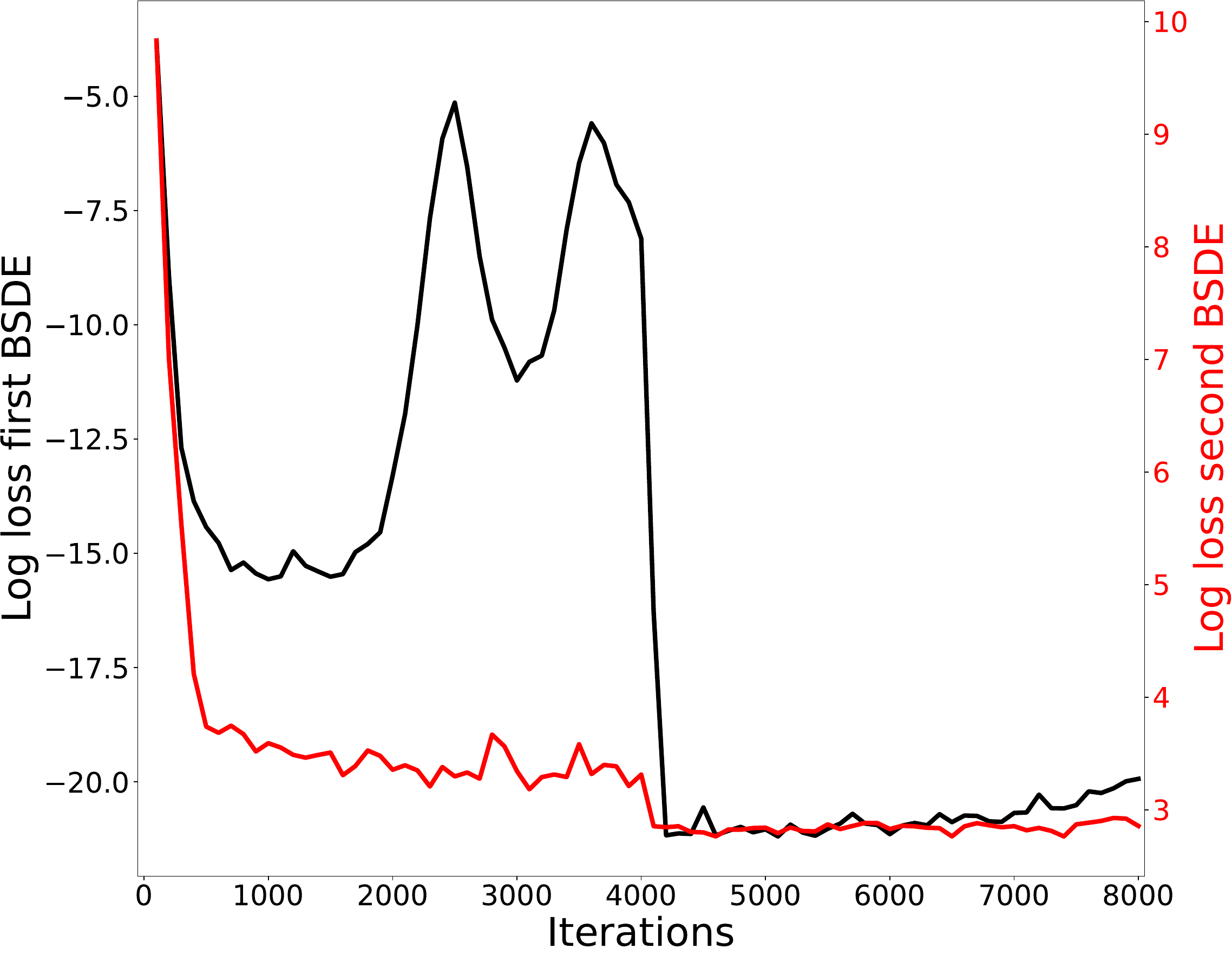}&
			\includegraphics[width=0.31\textwidth, height =0.25\textwidth]{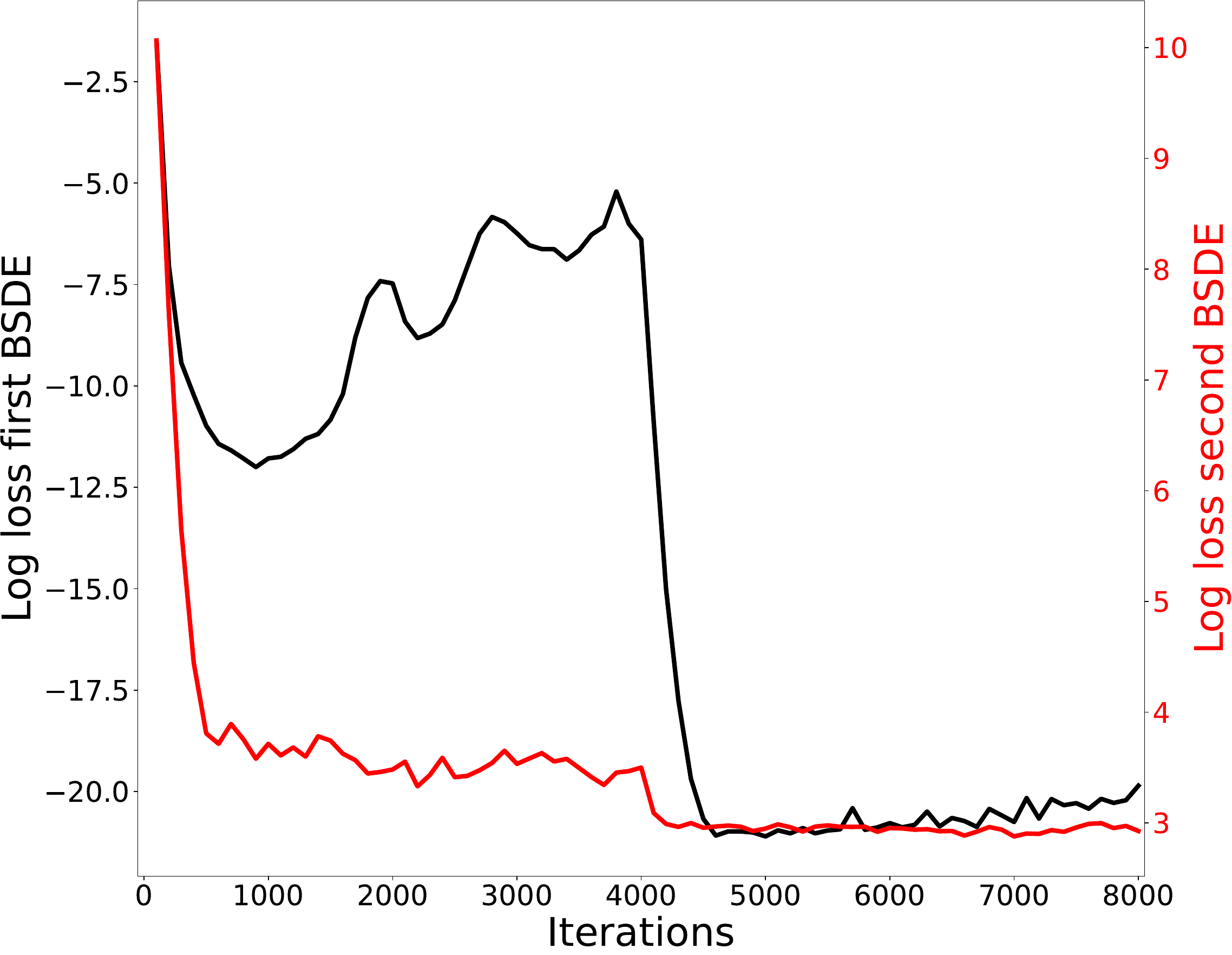} \\	
			\begin{turn}{90}\textbf{\large $\boldsymbol{m = 20}$}\end{turn}&\includegraphics[width=0.31\textwidth, height =0.25\textwidth]{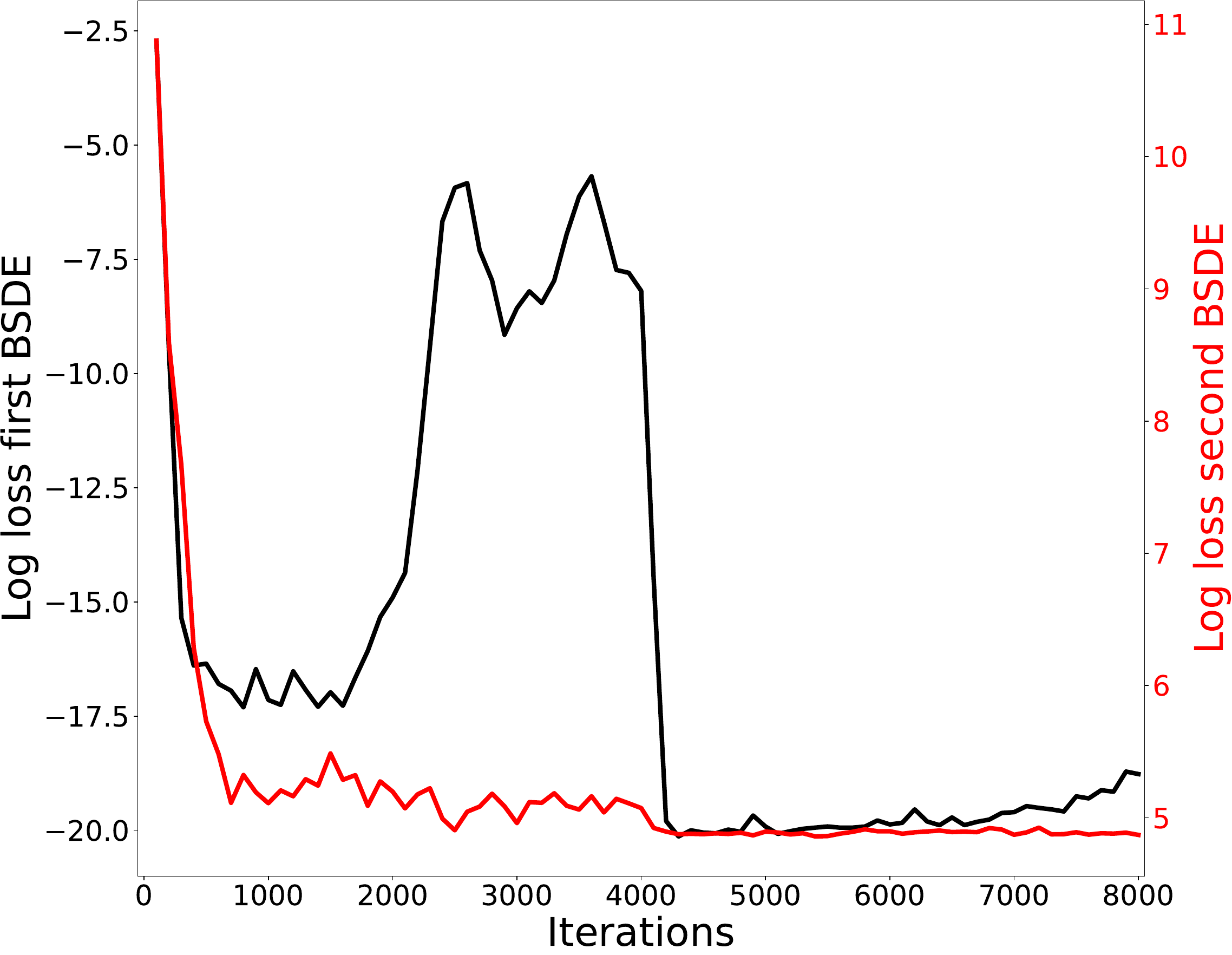}&
			\includegraphics[width=0.31\textwidth, height =0.25\textwidth]{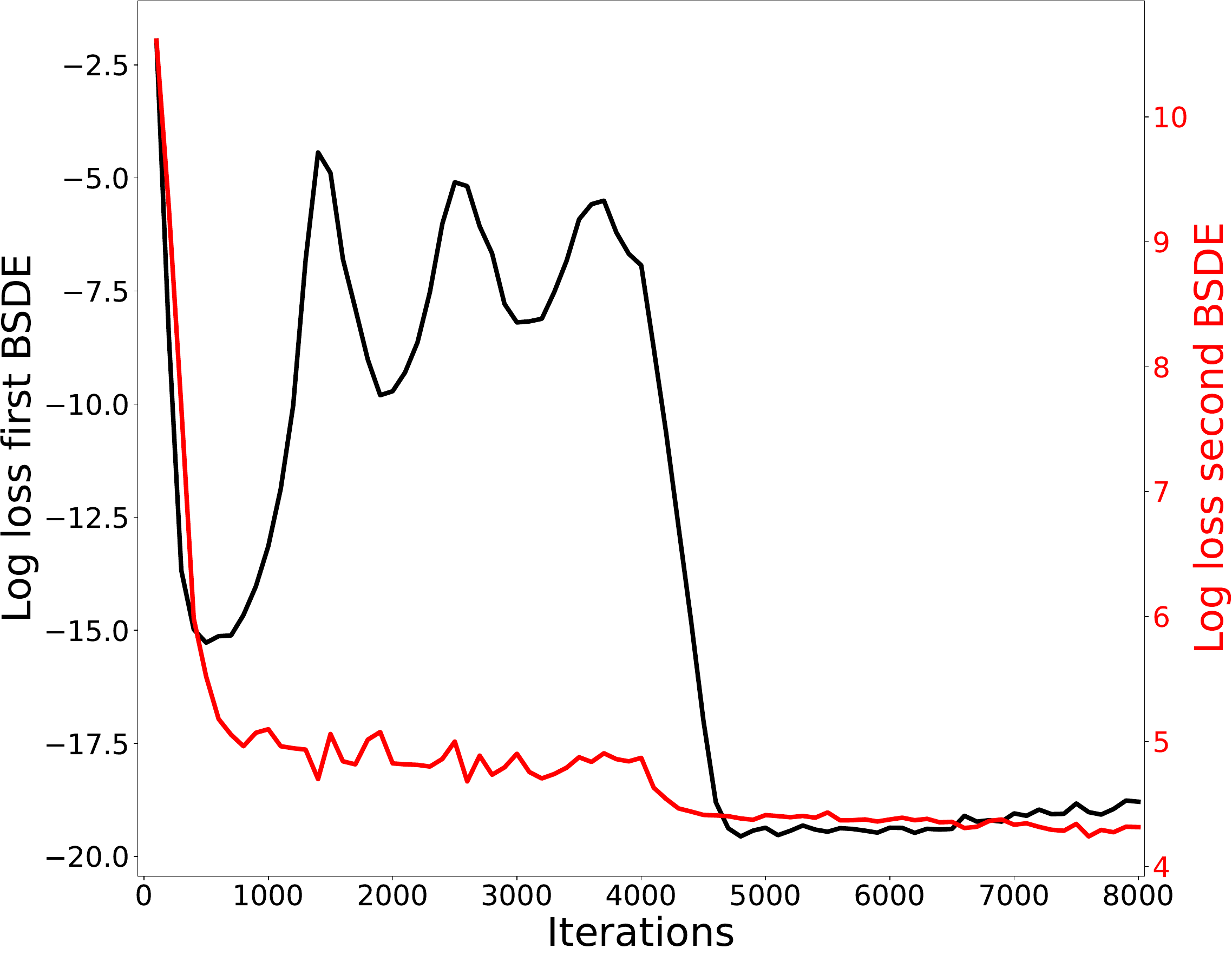} &	
			\includegraphics[width=0.31\textwidth, height =0.25\textwidth]{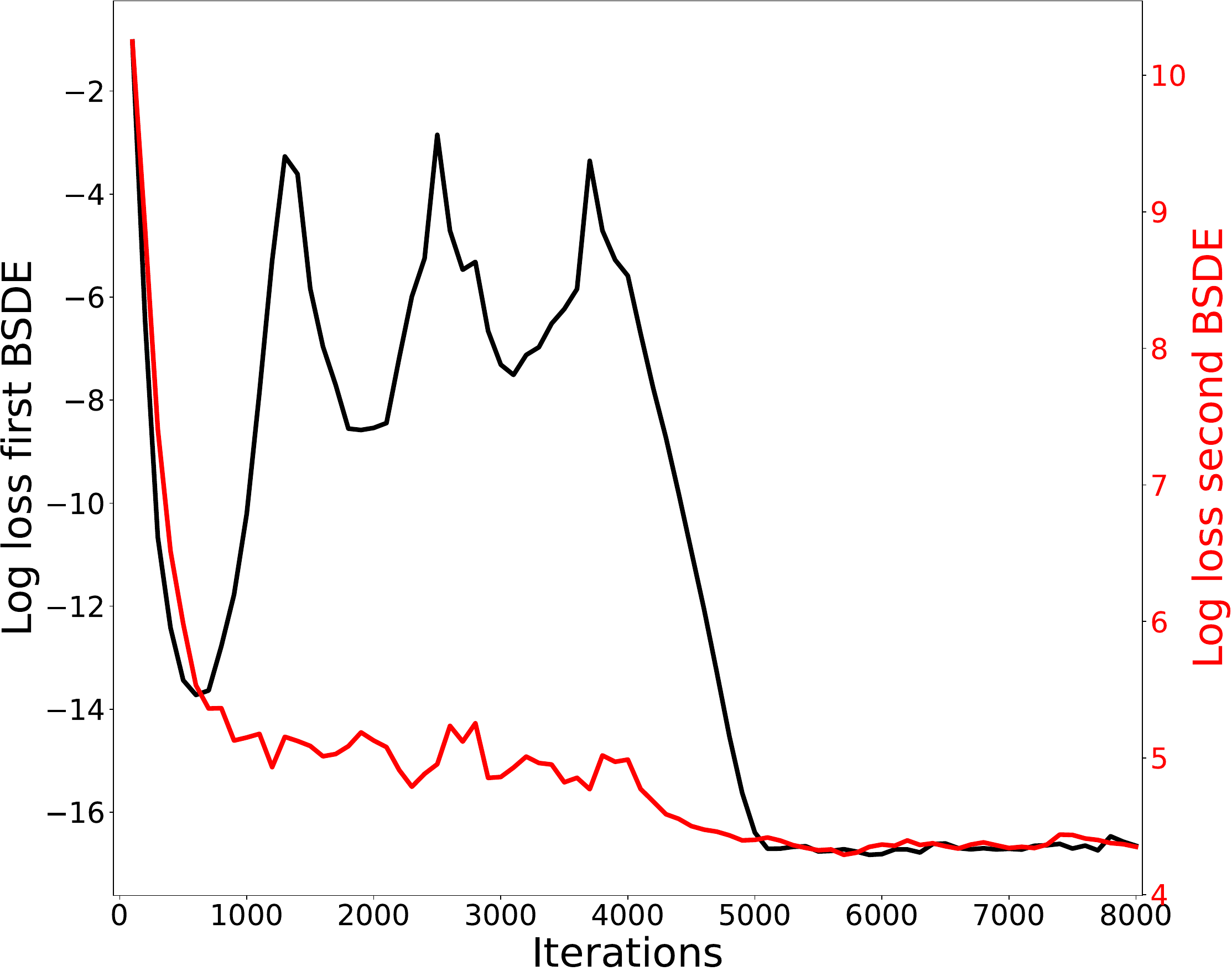}\\
			\begin{turn}{90}\textbf{\large $\boldsymbol{m = 100}$}\end{turn}&\includegraphics[width=0.31\textwidth, height =0.25\textwidth]{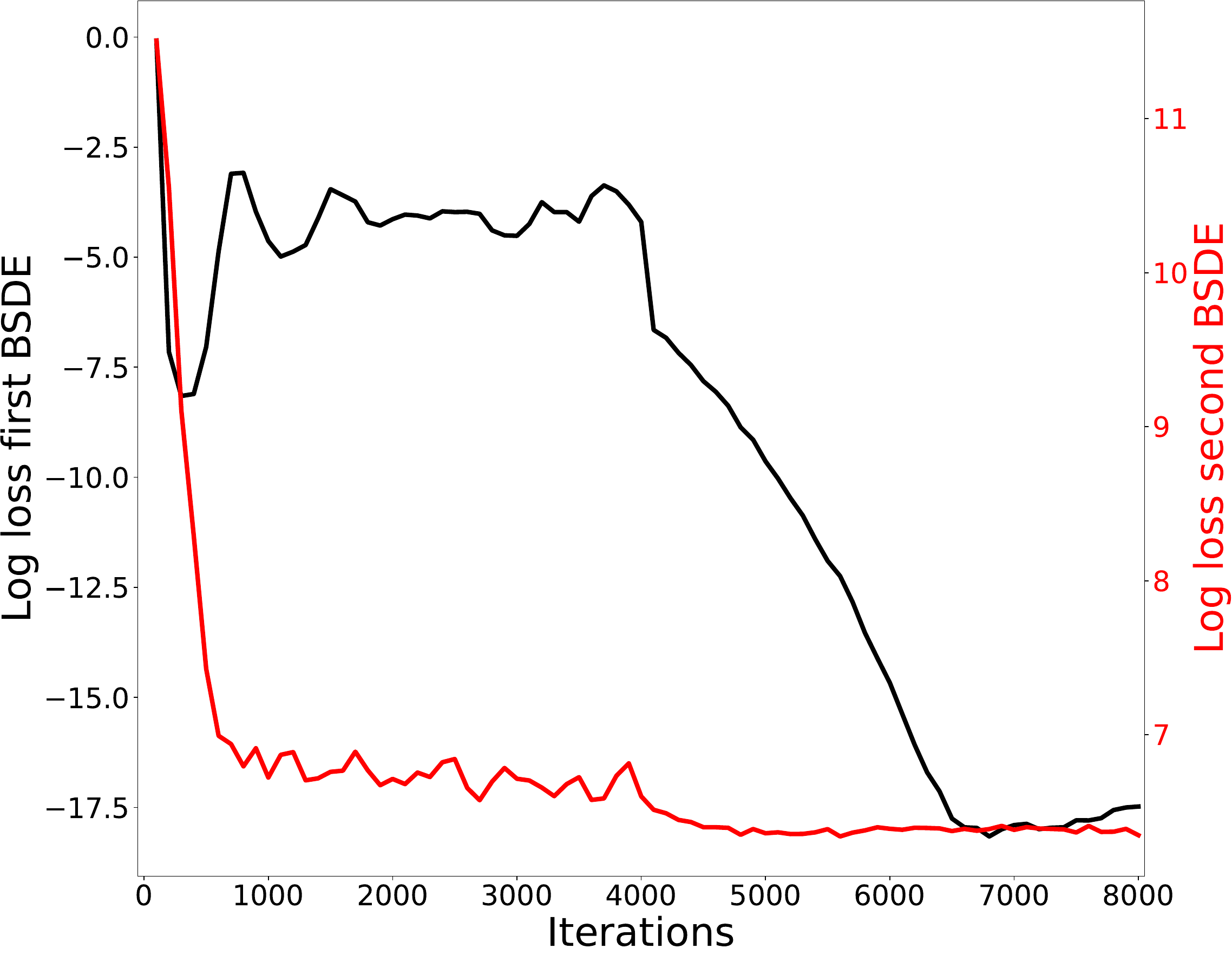}&
			\includegraphics[width=0.31\textwidth, height =0.25\textwidth]{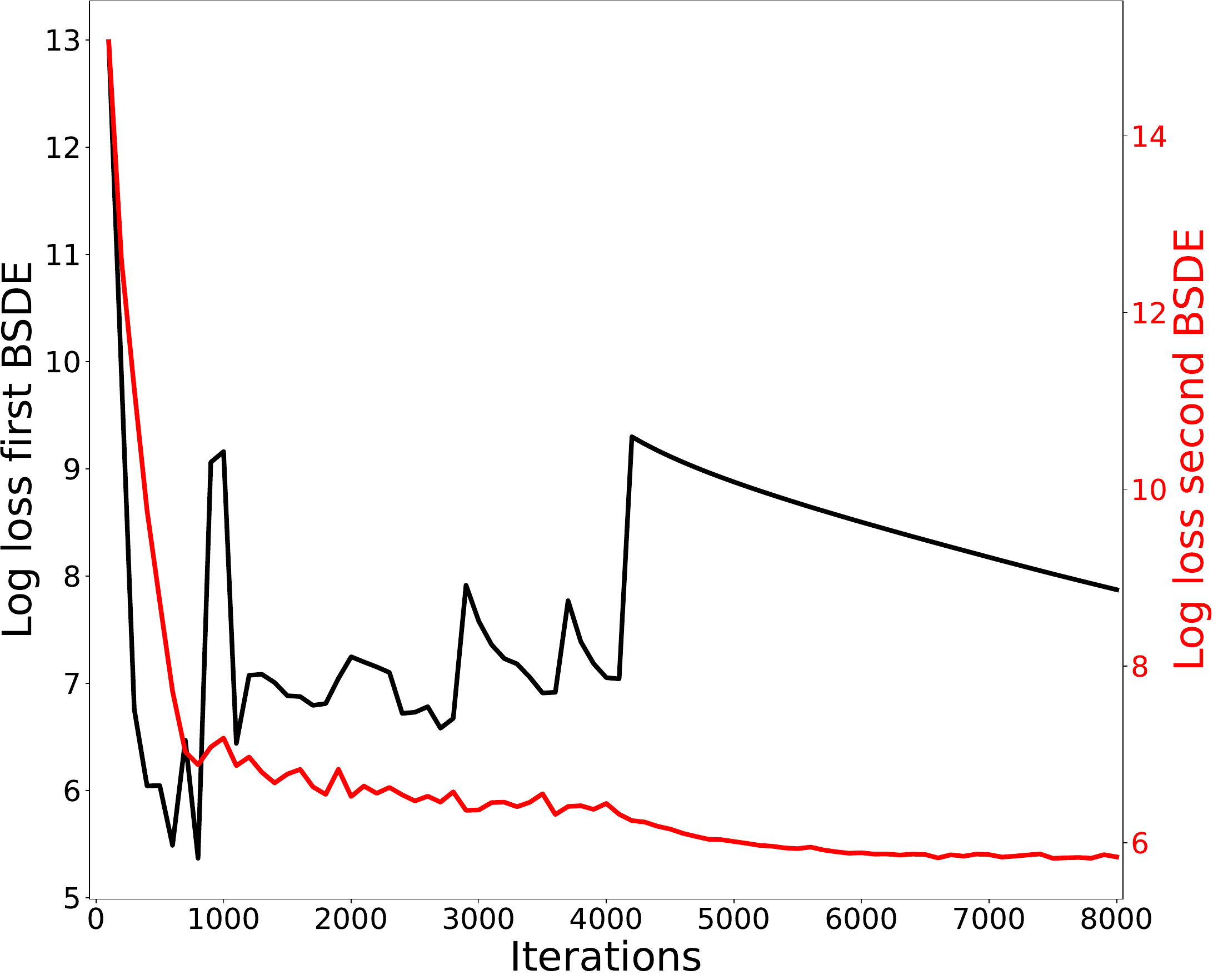} &	
			\includegraphics[width=0.31\textwidth, height =0.25\textwidth]{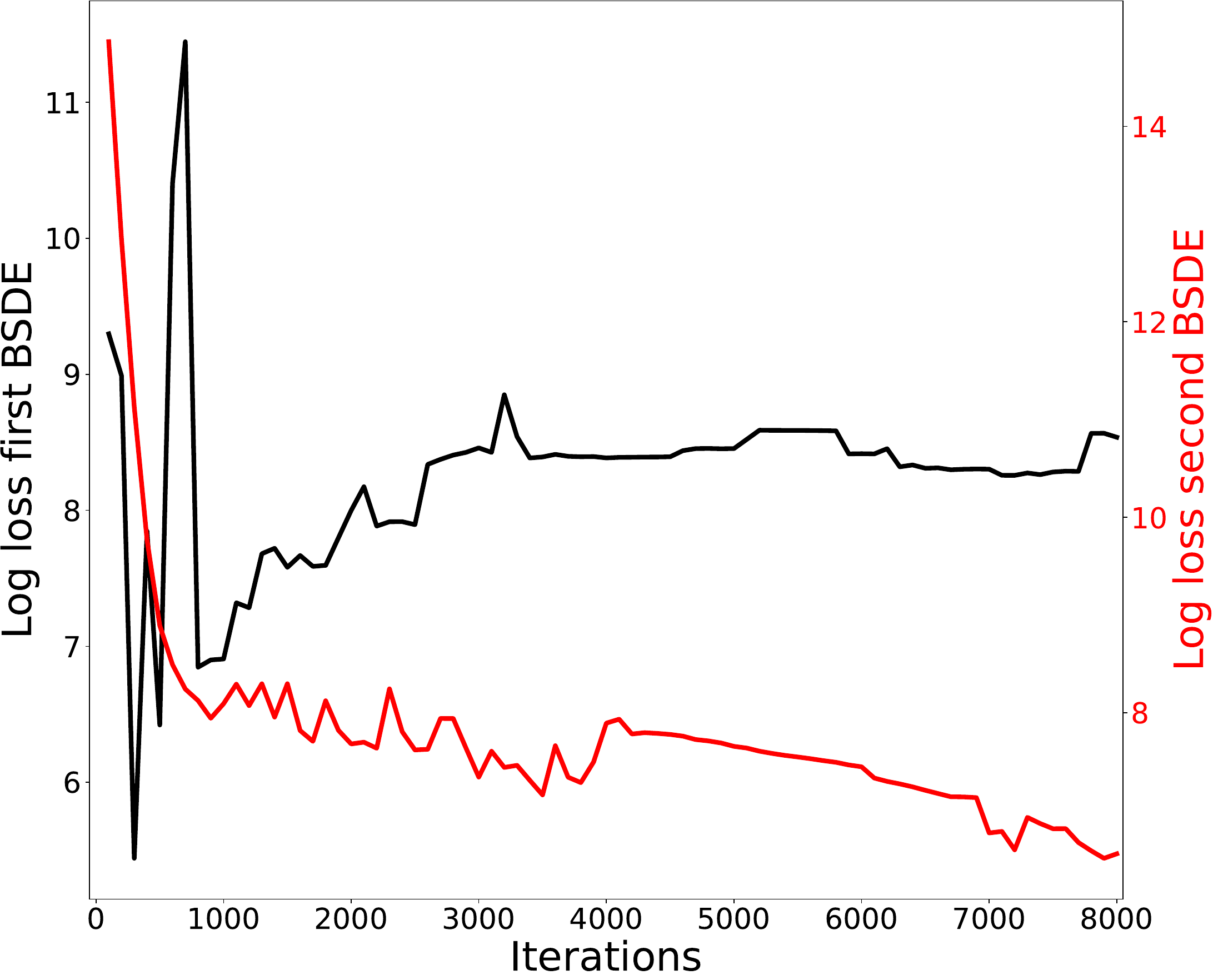}\\
			
		\end{tabular}
	}
	\caption{Logarithm of the loss functional as a function of the iteration number for the different experiment configurations presented in Table \ref{table:resultsMV}. The black curves represent the log-loss for the first BSDE and follow the black grid on the left-hand side. The red curves represent the log-loss for the second BSDE and follow the red grid on the right-hand side. \label{MV_logloss}}
\end{figure}

\begin{figure}[tp]
	\resizebox{1\textwidth}{!}{
		\begin{tabular}{@{}>{\centering\arraybackslash}m{0.5\textwidth}@{}>{\centering\arraybackslash}m{0.5\textwidth}@{}}
			 \textbf{\large $\boldsymbol{N = 50}$} & \textbf{\large $\boldsymbol{N = 100}$}\\
			\includegraphics[width=0.47\textwidth, height =0.42\textwidth]{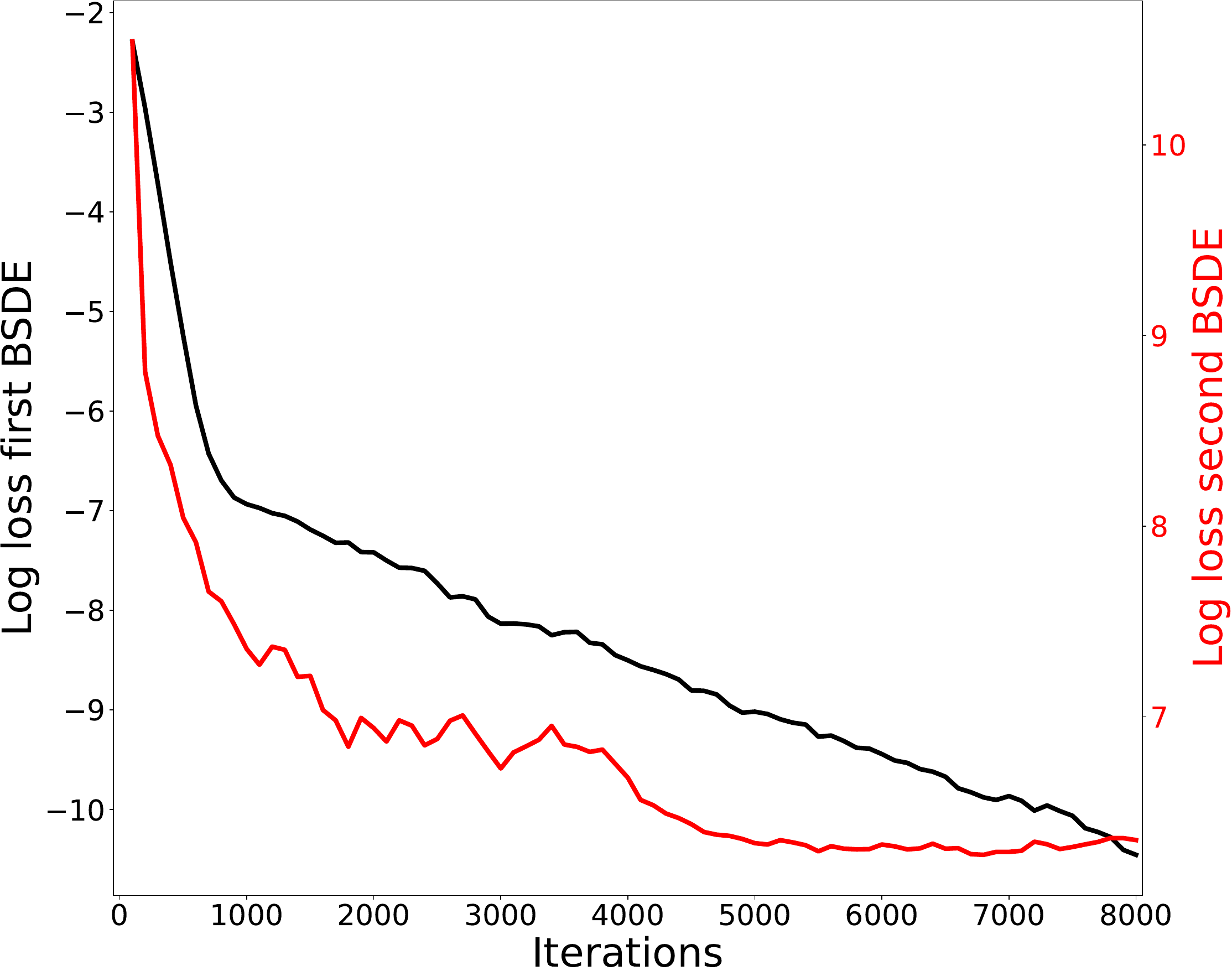} &	
			\includegraphics[width=0.47\textwidth, height =0.42\textwidth]{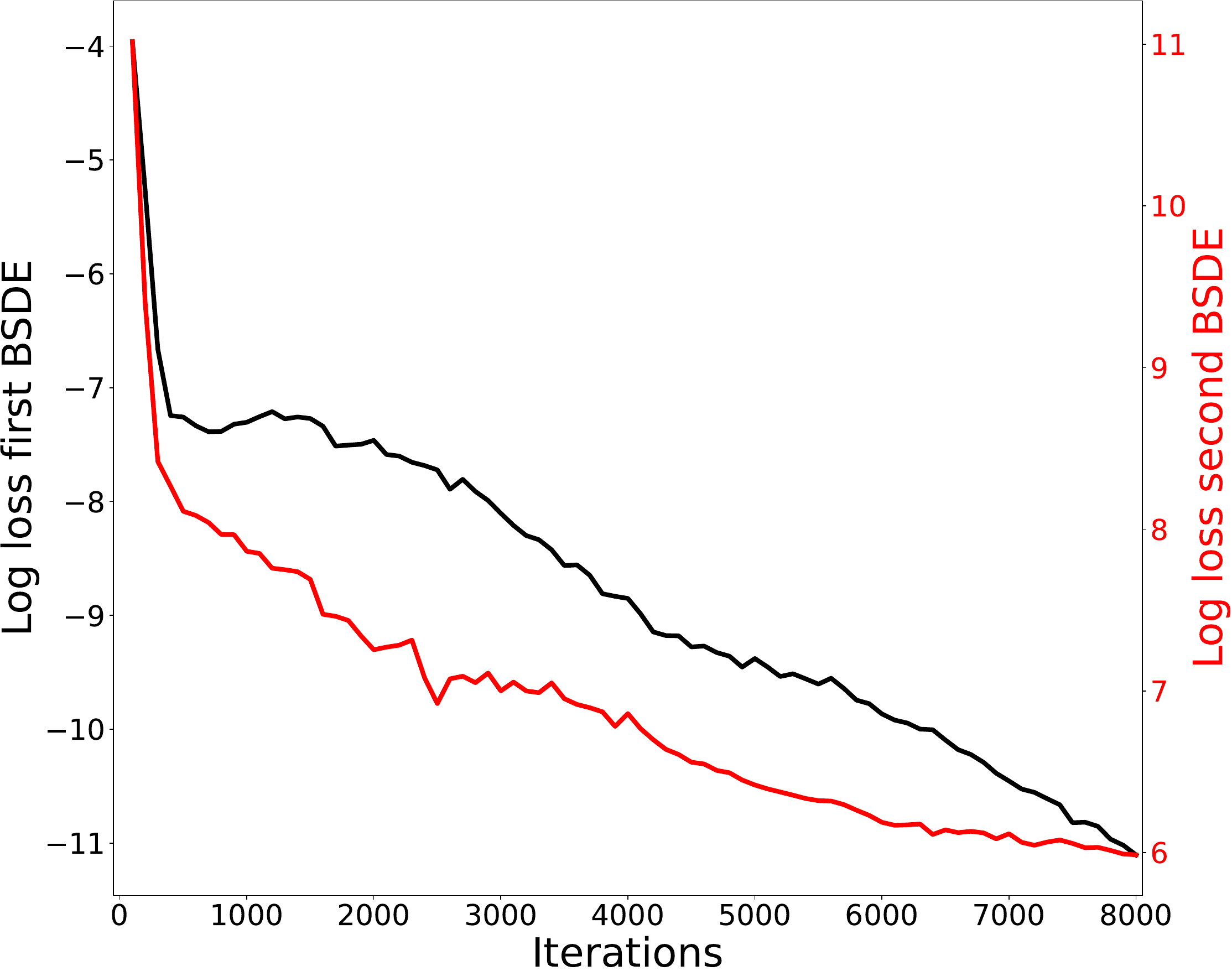}\\
			& \\
			\multicolumn{1}{l}{\textbf{BSDE solver L value:} $0.97007 \,(0.00489\, \%)$} & \multicolumn{1}{l}{\textbf{BSDE solver L value:} $0.97024\,(0.0230 \,\%)$}\\
			\multicolumn{1}{l}{\textbf{BSDE solver price: }$68.892 \,(0.0878\, \%)$} & \multicolumn{1}{l}{\textbf{BSDE solver price:} $68.910 \,(0.114 \,\%)$}\\
		\end{tabular}
	}
	\caption{Mean-variance hedging results with portfolio dimension $m = 100$, when we reduce the learning rates values and improve the experiments in Table \ref{table:resultsMV} and Figure \ref{MV_logloss}. \label{MV_extrafigure}}
\end{figure}

\begin{figure}[tp]
	\resizebox{1\textwidth}{!}{
		\begin{tabular}{@{}>{\centering\arraybackslash}m{0.5\textwidth}@{}>{\centering\arraybackslash}m{0.5\textwidth}@{}}
			\multicolumn{2}{c}{\huge \textbf{Mean-variance hedging: $\boldsymbol{N = 10}$}}\\
			& \\
			\includegraphics[width=0.48\textwidth, height =0.35\textwidth]{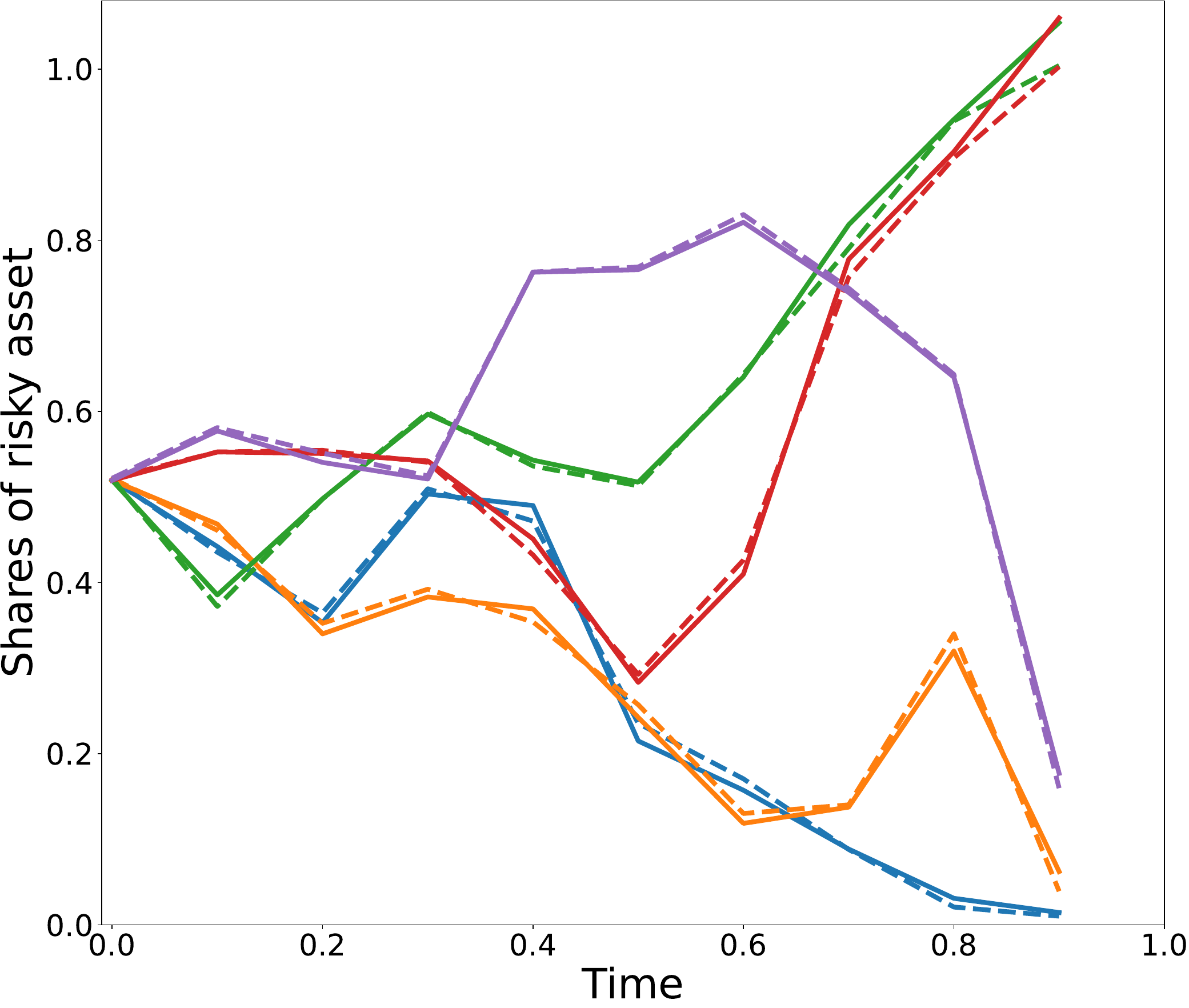} &
			\includegraphics[width=0.48\textwidth, height =0.35\textwidth]{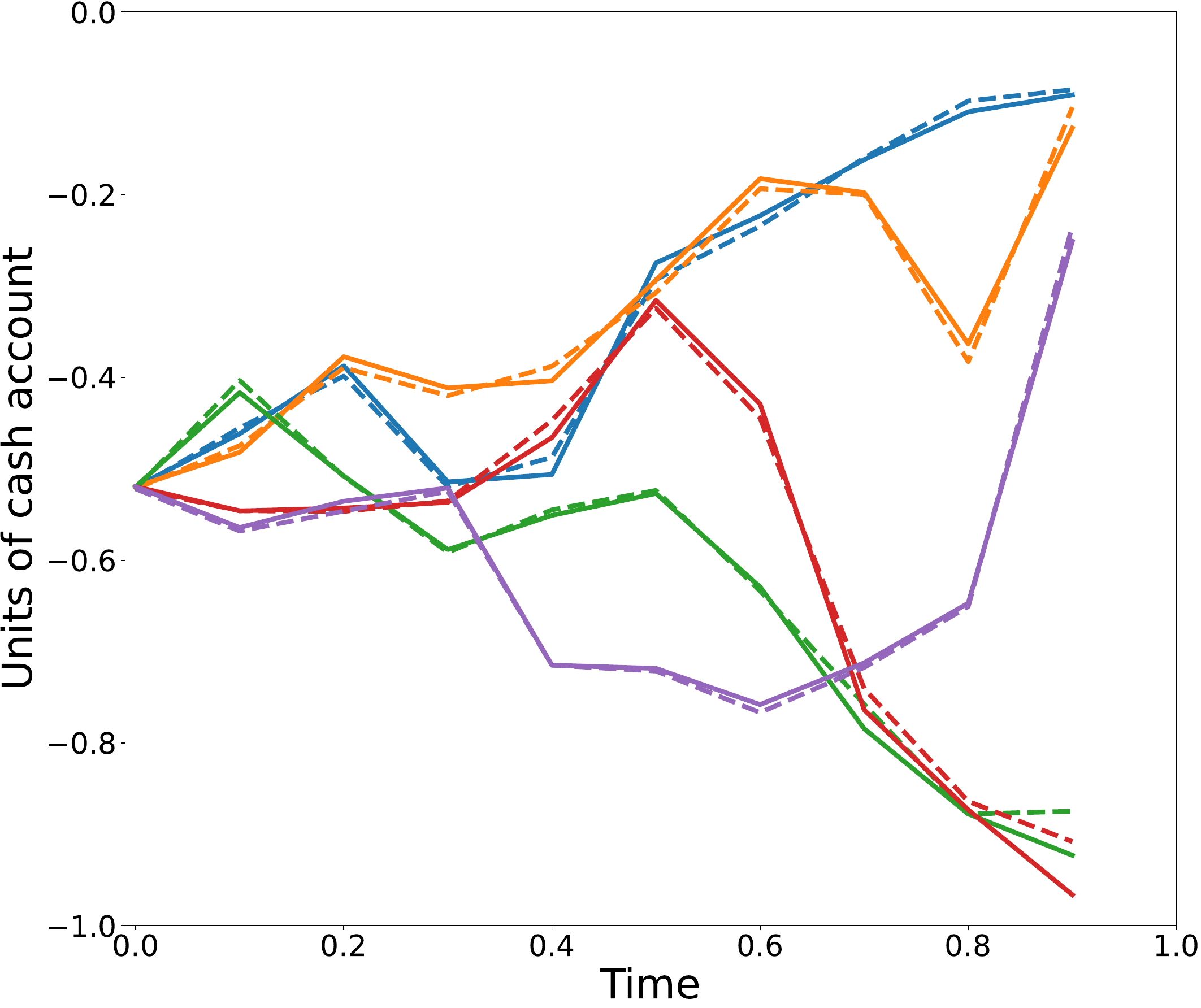}\\
			\multicolumn{2}{c}{\includegraphics[width=0.48\textwidth, height =0.35\textwidth]{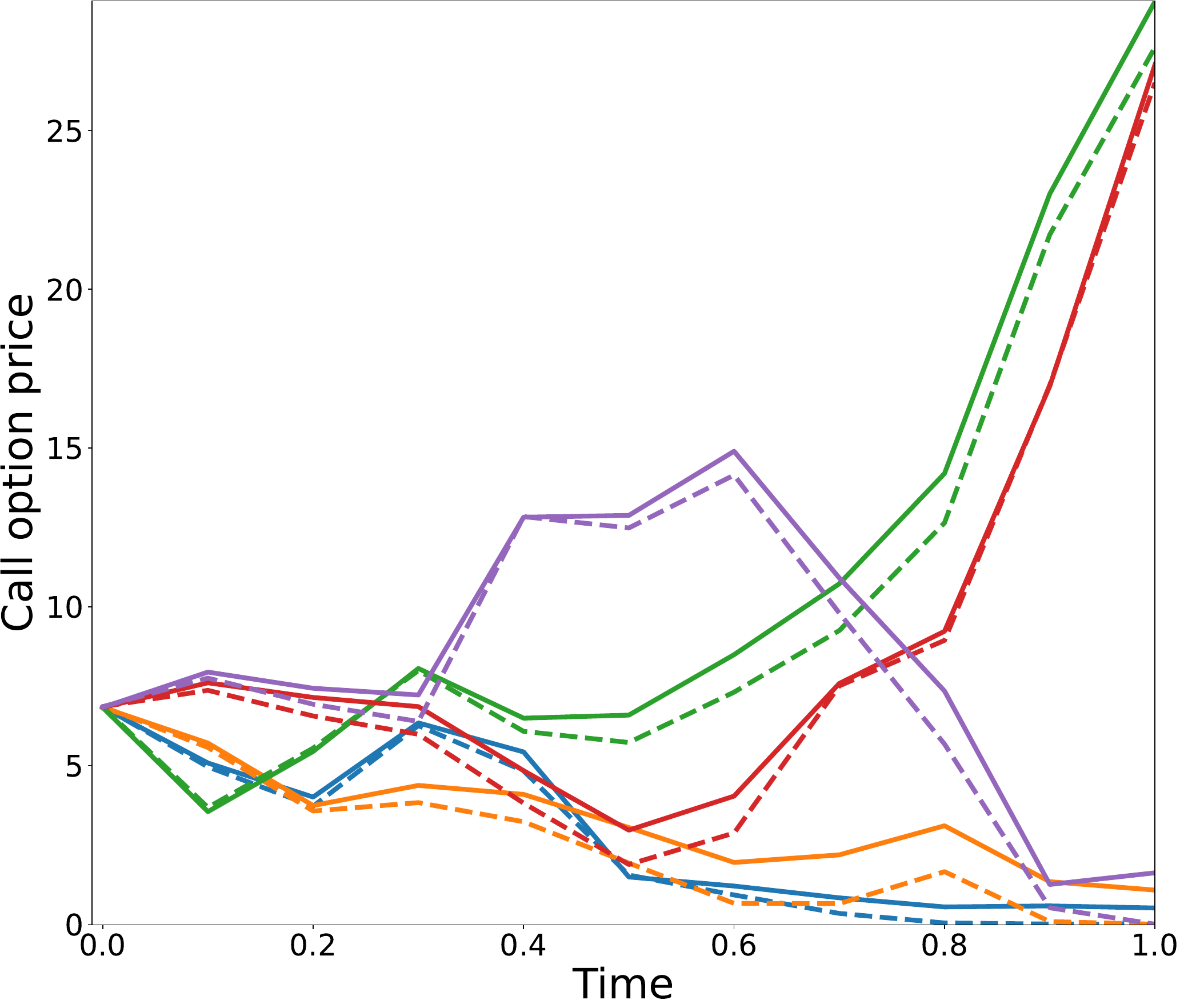}} \\
		\end{tabular}
	}
	\caption{Deep solver solution (solid line) and benchmark solution (dashed line) for the mean-variance hedging in a $10$ points discretization grid for $5$ random samples  in the interval $[0, 1]$. Upper panel: the shares of risky asset (left) and the units of cash account (right); lower panel: the call option price. \label{MV10}}
\end{figure}

\begin{figure}[tp]
	\resizebox{1\textwidth}{!}{
		\begin{tabular}{@{}>{\centering\arraybackslash}m{0.5\textwidth}@{}>{\centering\arraybackslash}m{0.5\textwidth}@{}}
			\multicolumn{2}{c}{\huge \textbf{Mean-variance hedging: $\boldsymbol{N = 50}$}}\\
			& \\
			\includegraphics[width=0.48\textwidth, height =0.35\textwidth]{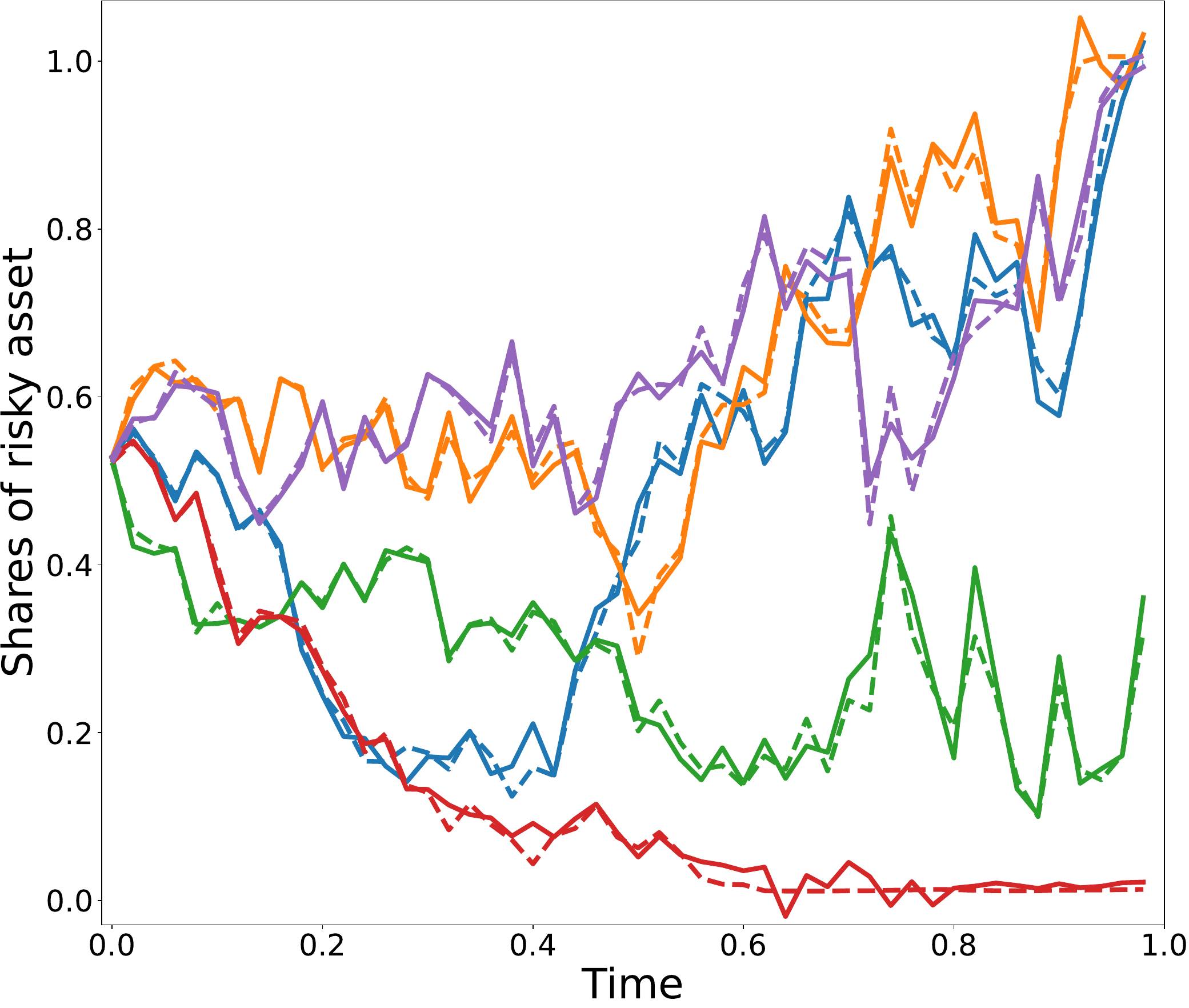} &
			\includegraphics[width=0.48\textwidth, height =0.35\textwidth]{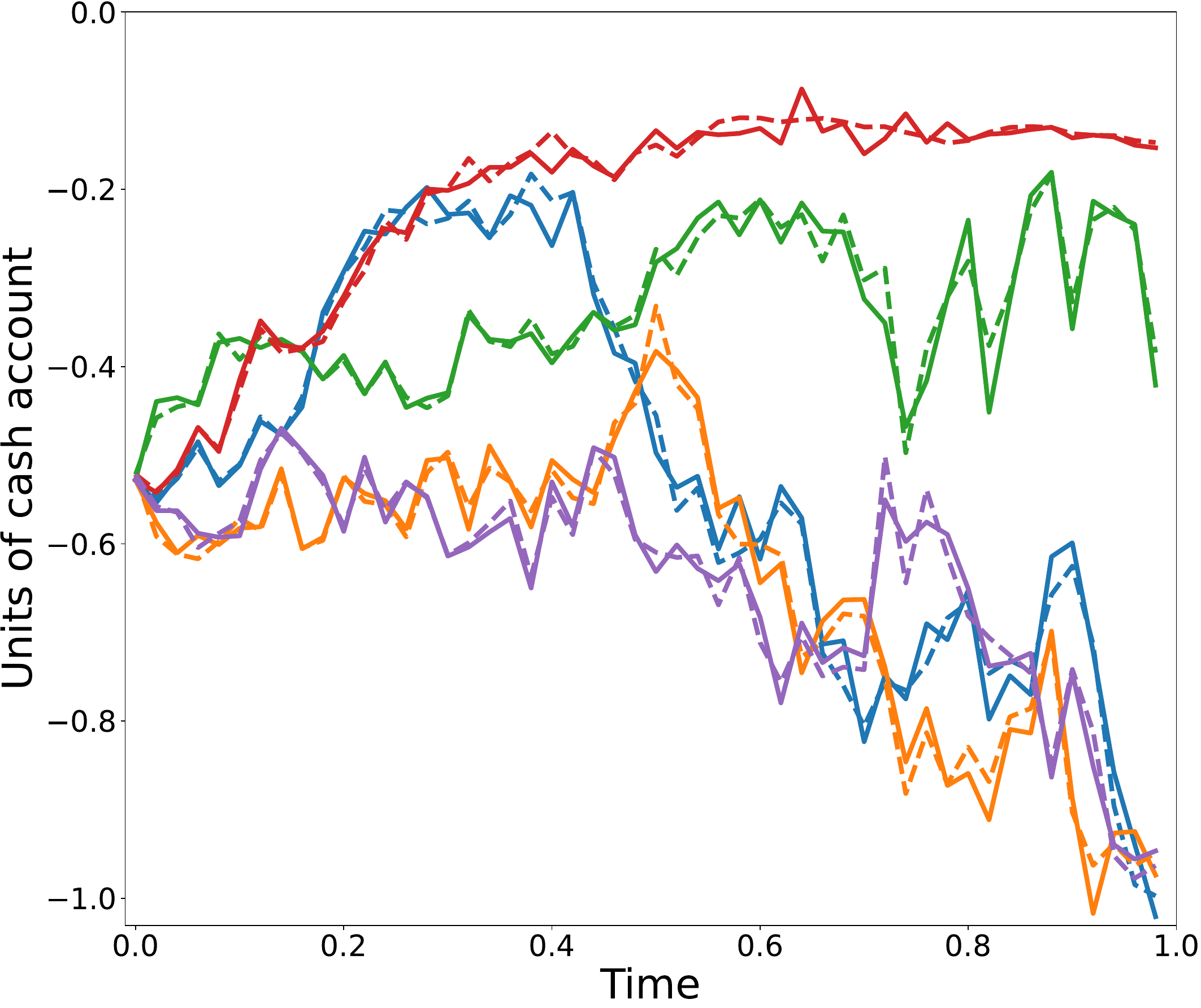}\\
			\multicolumn{2}{c}{\includegraphics[width=0.48\textwidth, height =0.35\textwidth]{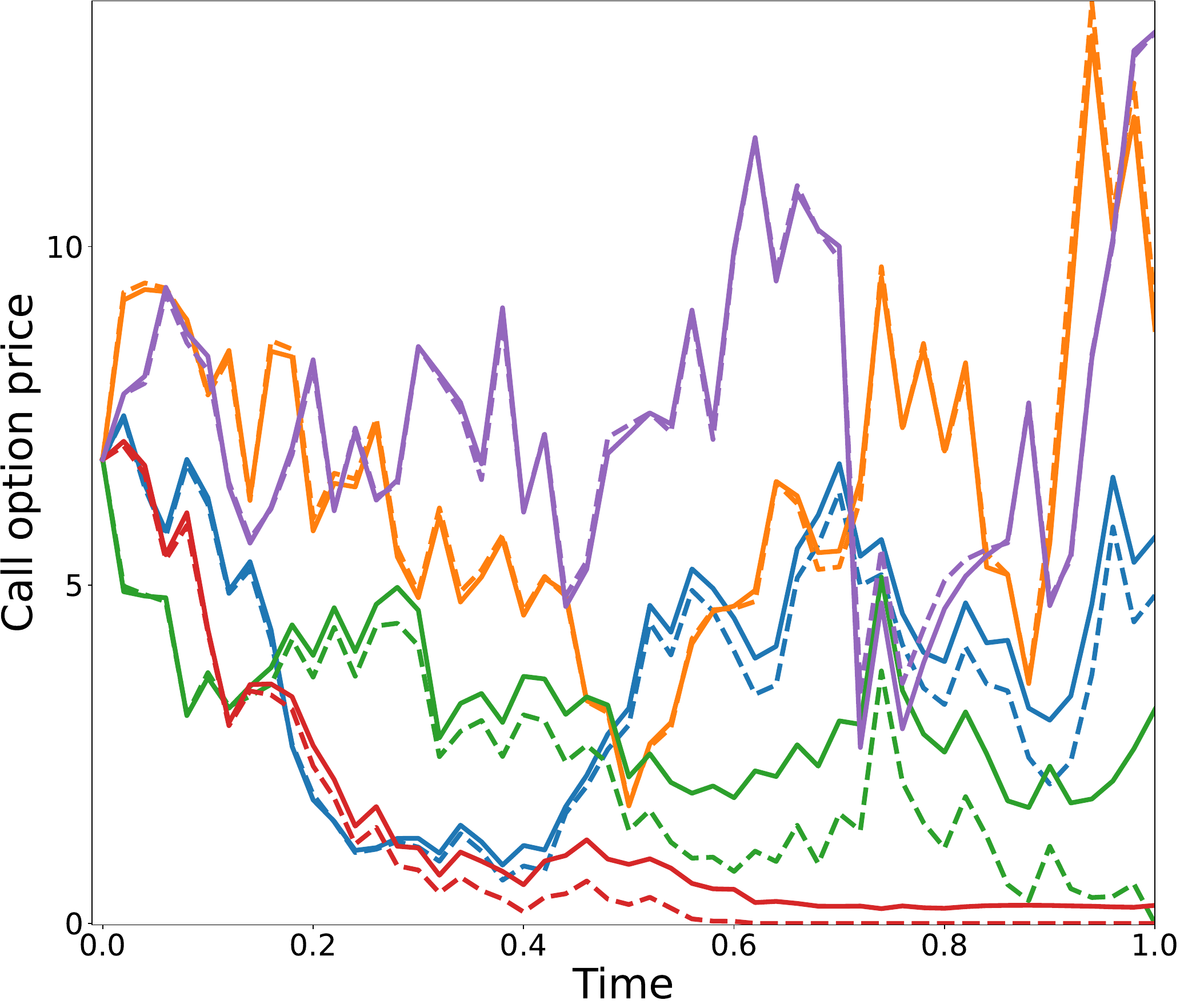}} \\
		\end{tabular}
	}
	\caption{Deep solver solution (solid line) and benchmark solution (dashed line) for the mean-variance hedging in a $50$ points discretization grid for $5$ random samples  in the interval $[0, 1]$. Upper panel: the shares of risky asset (left) and the units of cash account (right); lower panel: the call option price. \label{MV50}}
\end{figure}

\begin{figure}[tp]
	\resizebox{1\textwidth}{!}{
		\begin{tabular}{@{}>{\centering\arraybackslash}m{0.5\textwidth}@{}>{\centering\arraybackslash}m{0.5\textwidth}@{}}
			\multicolumn{2}{c}{\huge \textbf{Mean-variance hedging: $\boldsymbol{N = 100}$}}\\
			& \\
			\includegraphics[width=0.48\textwidth, height =0.35\textwidth]{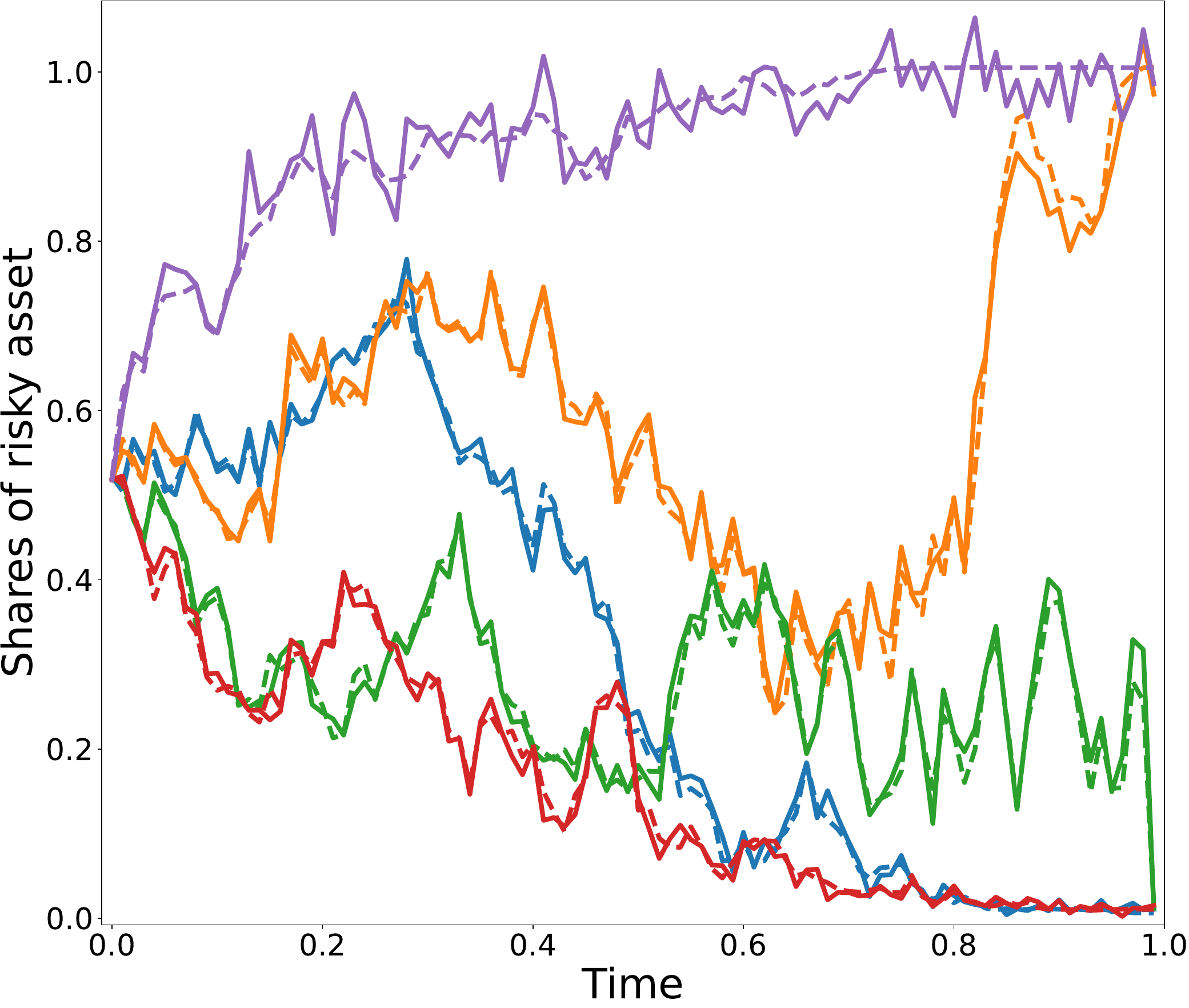} &
			\includegraphics[width=0.48\textwidth, height =0.35\textwidth]{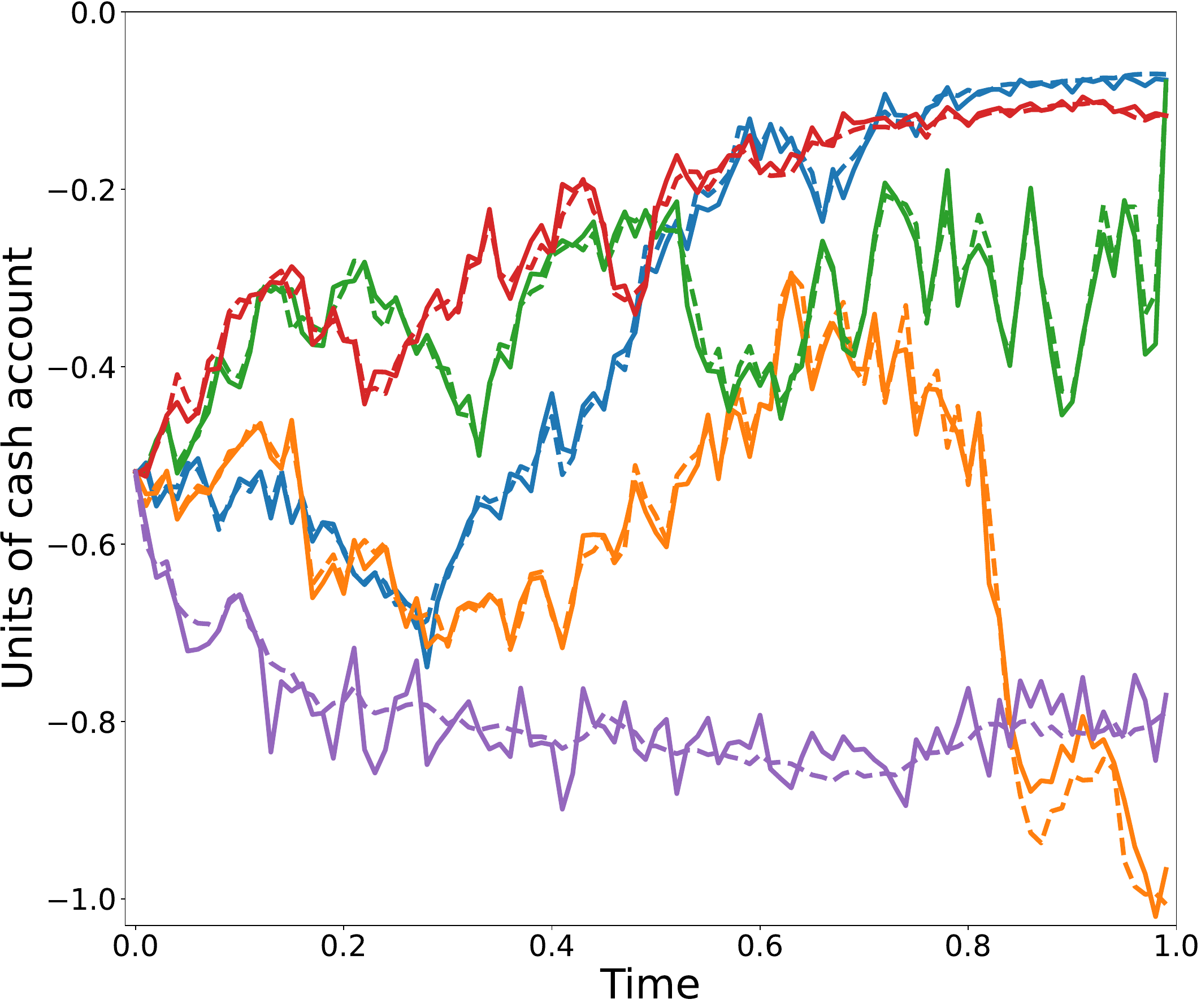}\\
			\multicolumn{2}{c}{\includegraphics[width=0.48\textwidth, height =0.35\textwidth]{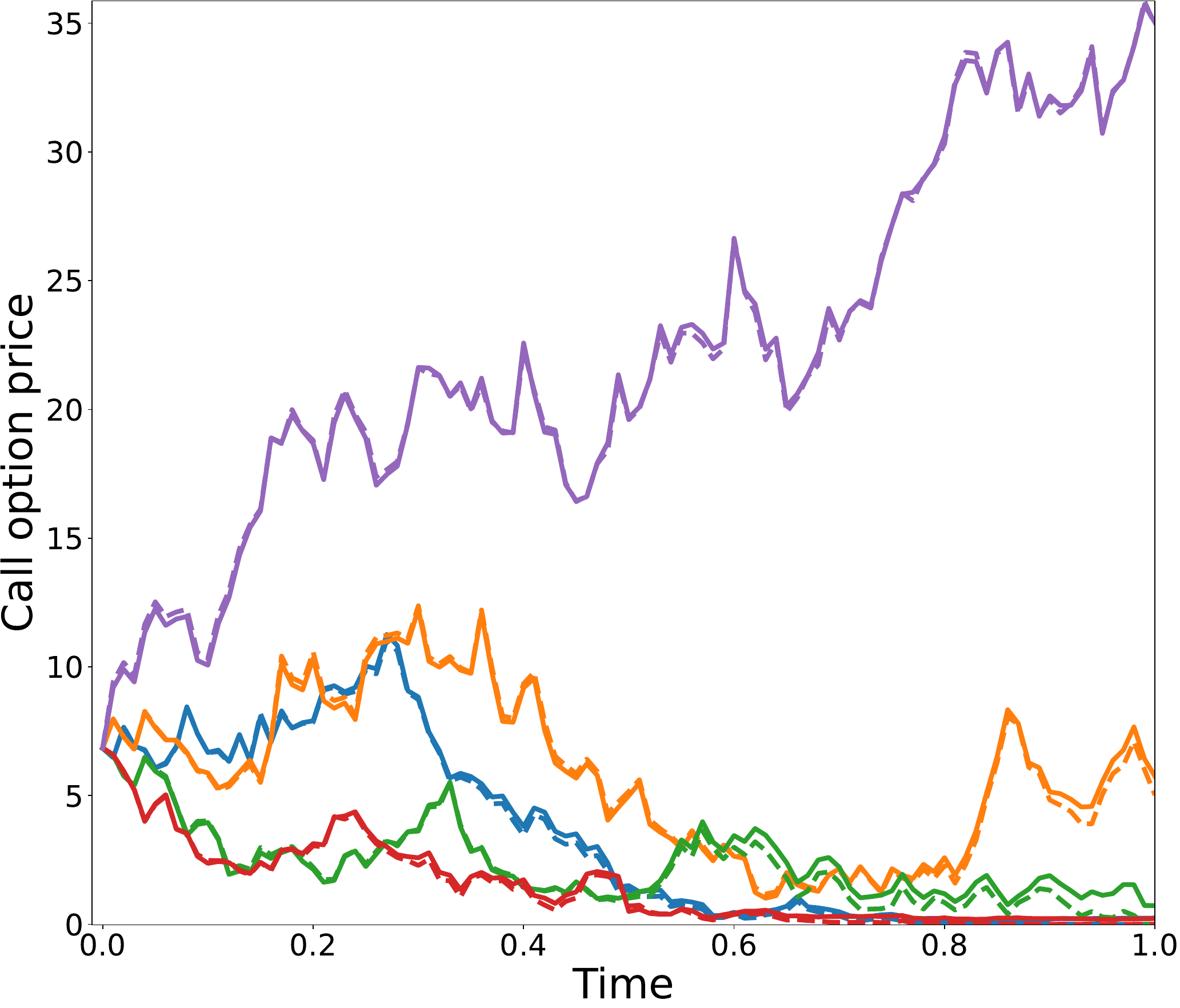}} \\
		\end{tabular}
	}
	\caption{Deep solver solution (solid line) and benchmark solution (dashed line) for the mean-variance hedging in a $100$ points discretization grid for $5$ random samples  in the interval $[0, 1]$. Upper panel: the shares of risky asset (left) and the units of cash account (right); lower panel: the call option price. \label{MV100}}
\end{figure}

\begin{figure}[tp]
	\resizebox{1\textwidth}{!}{
		\begin{tabular}{@{}>{\centering\arraybackslash}m{0.04\textwidth}@{}>{\centering\arraybackslash}m{0.32\textwidth}@{}>{\centering\arraybackslash}m{0.32\textwidth}@{}>{\centering\arraybackslash}m{0.32\textwidth}@{}}
			\multicolumn{4}{c}{\huge \textbf{Deep BSRE solution}}\\
			&&& \\
			&\textbf{\large $\boldsymbol{N = 10}$} & \textbf{\large $\boldsymbol{N = 50}$} & \textbf{\large  $ \boldsymbol{N =100}$}\\
			&&& \\
			\begin{turn}{90}\textbf{\large $\boldsymbol{m = 1}$}\end{turn}&\includegraphics[width=0.31\textwidth, height =0.25\textwidth]{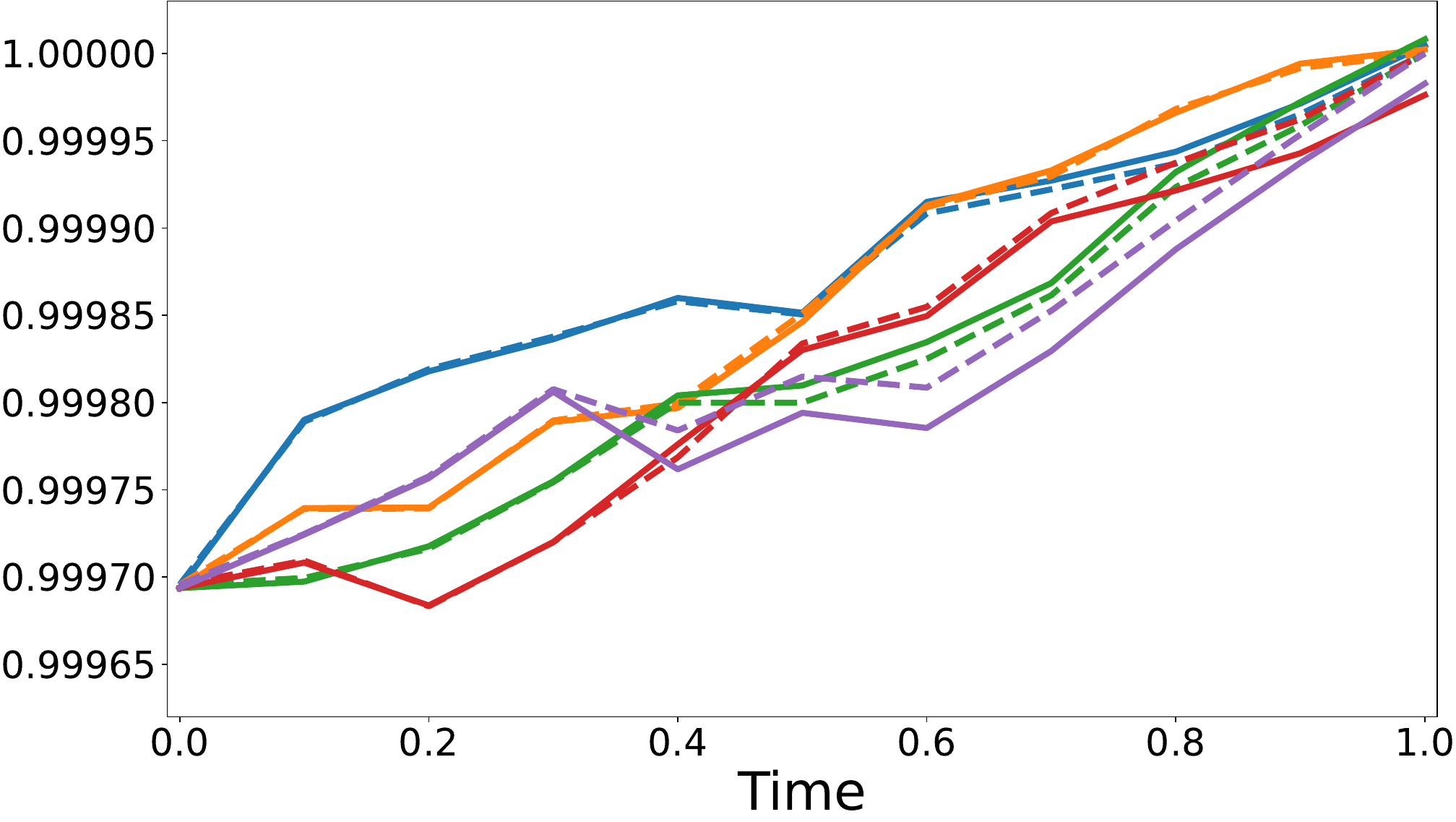} &
			\includegraphics[width=0.31\textwidth, height =0.25\textwidth]{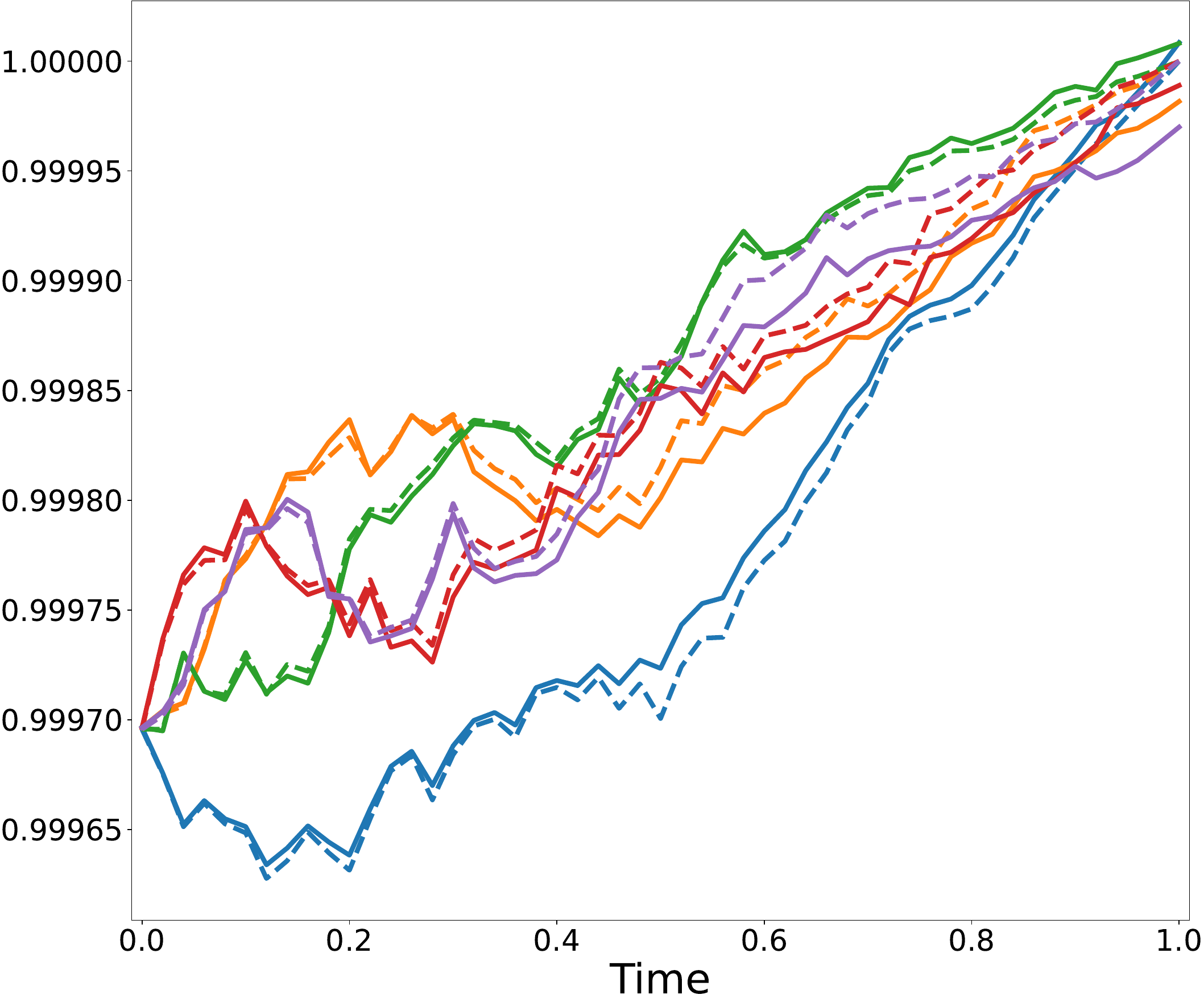}&
			\includegraphics[width=0.31\textwidth, height =0.25\textwidth]{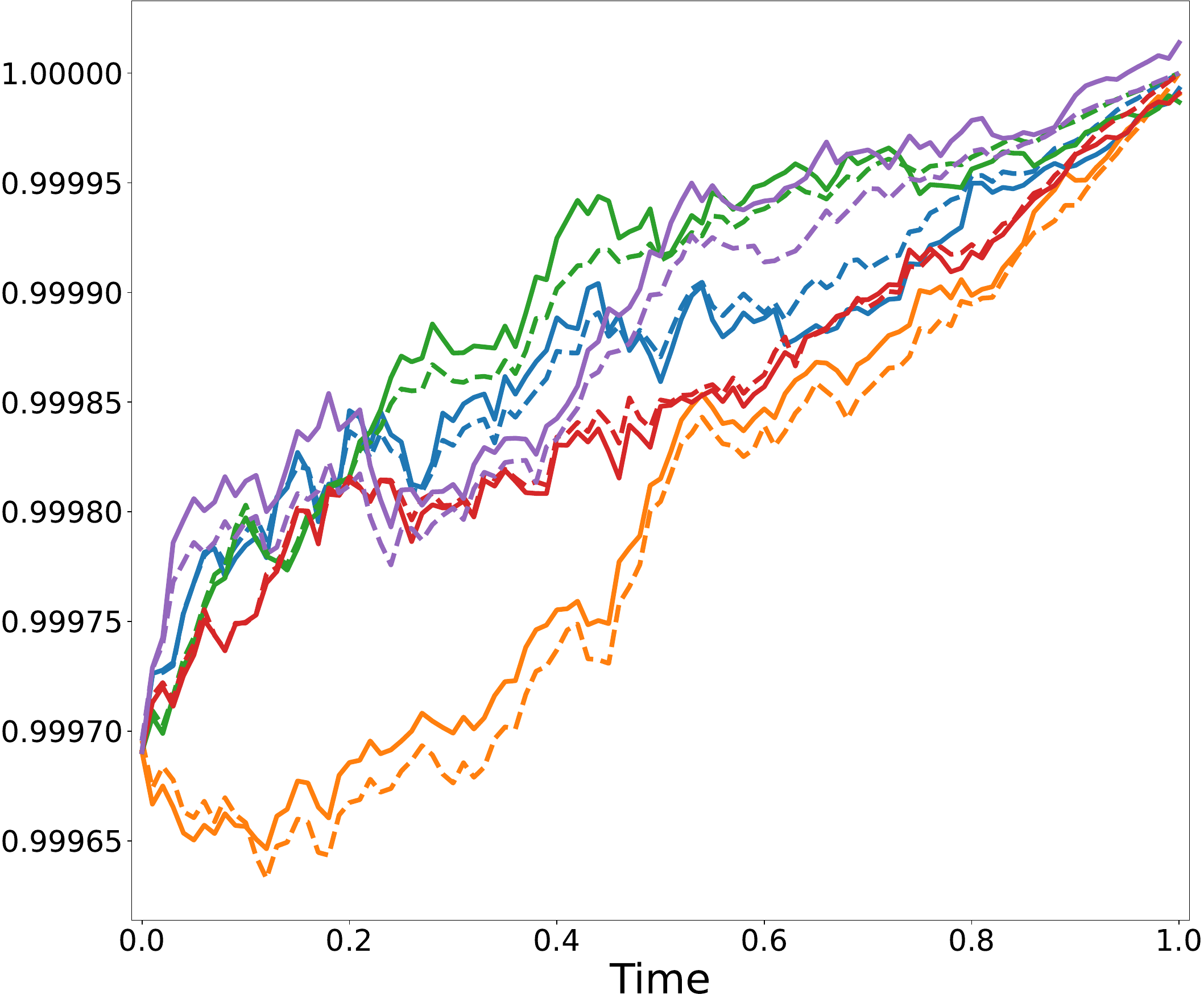} \\	
			\begin{turn}{90}\textbf{\large $\boldsymbol{m = 5}$}\end{turn}&\includegraphics[width=0.31\textwidth, height =0.25\textwidth]{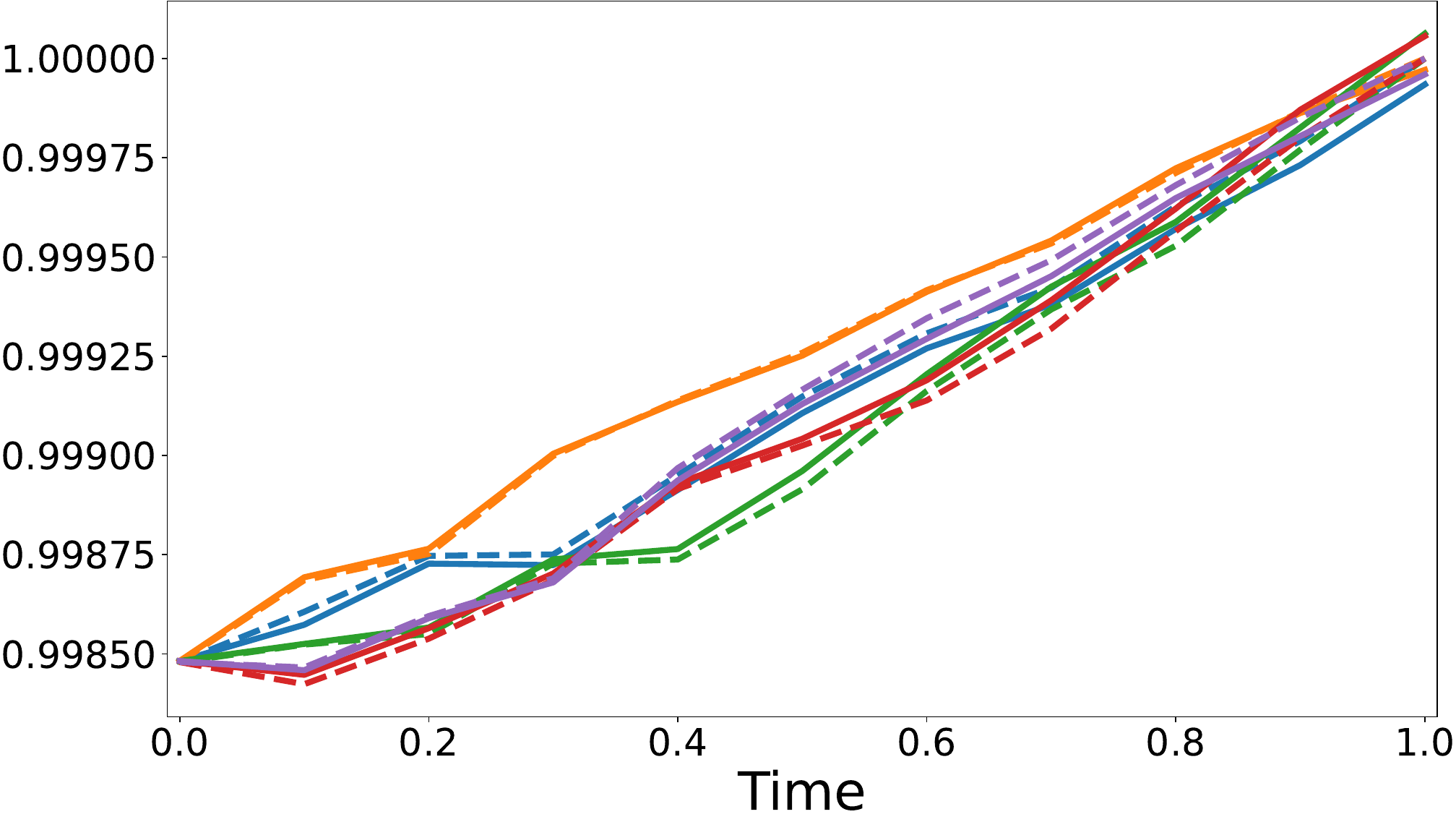} &
			\includegraphics[width=0.31\textwidth, height =0.25\textwidth]{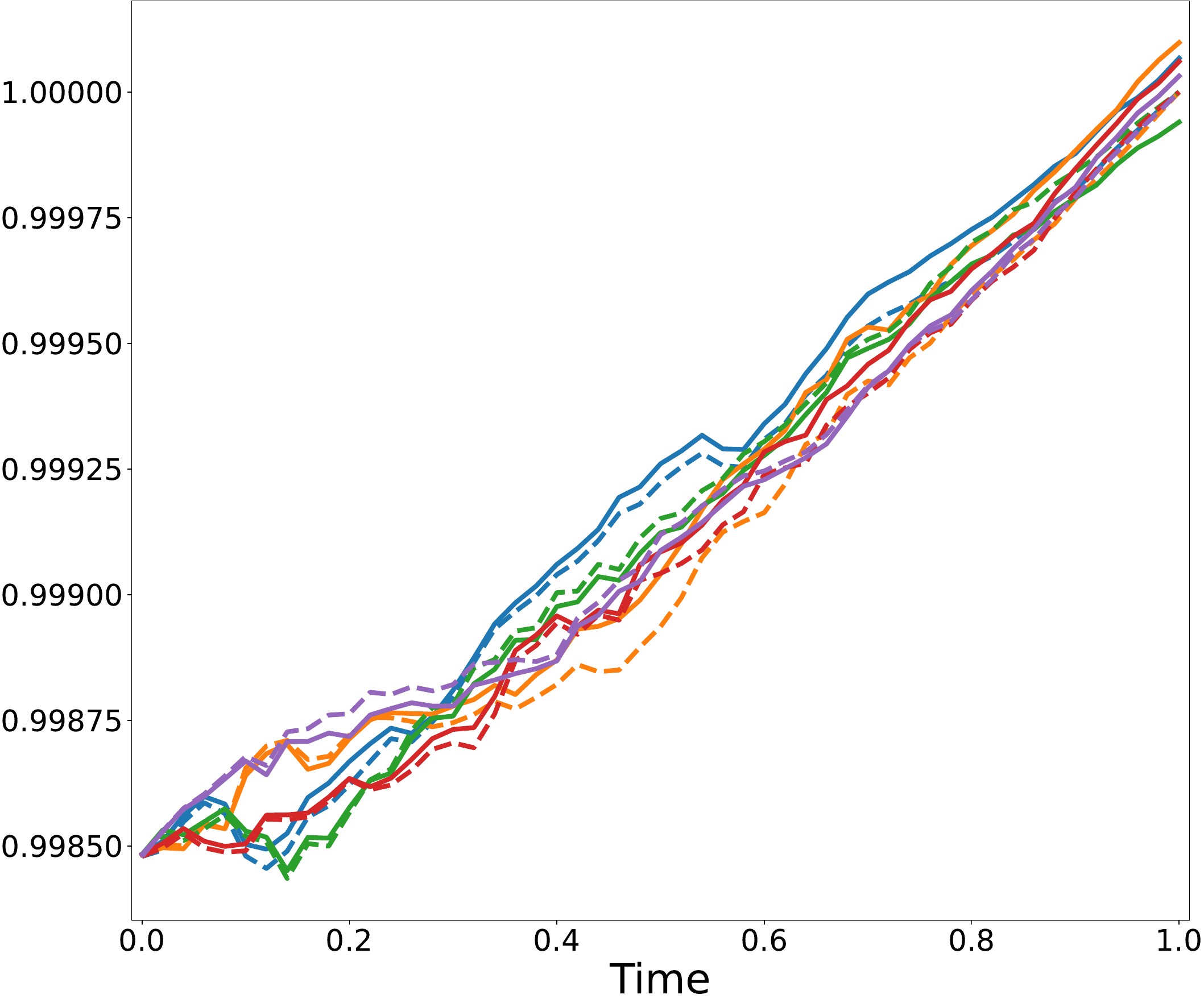}&
			\includegraphics[width=0.31\textwidth, height =0.25\textwidth]{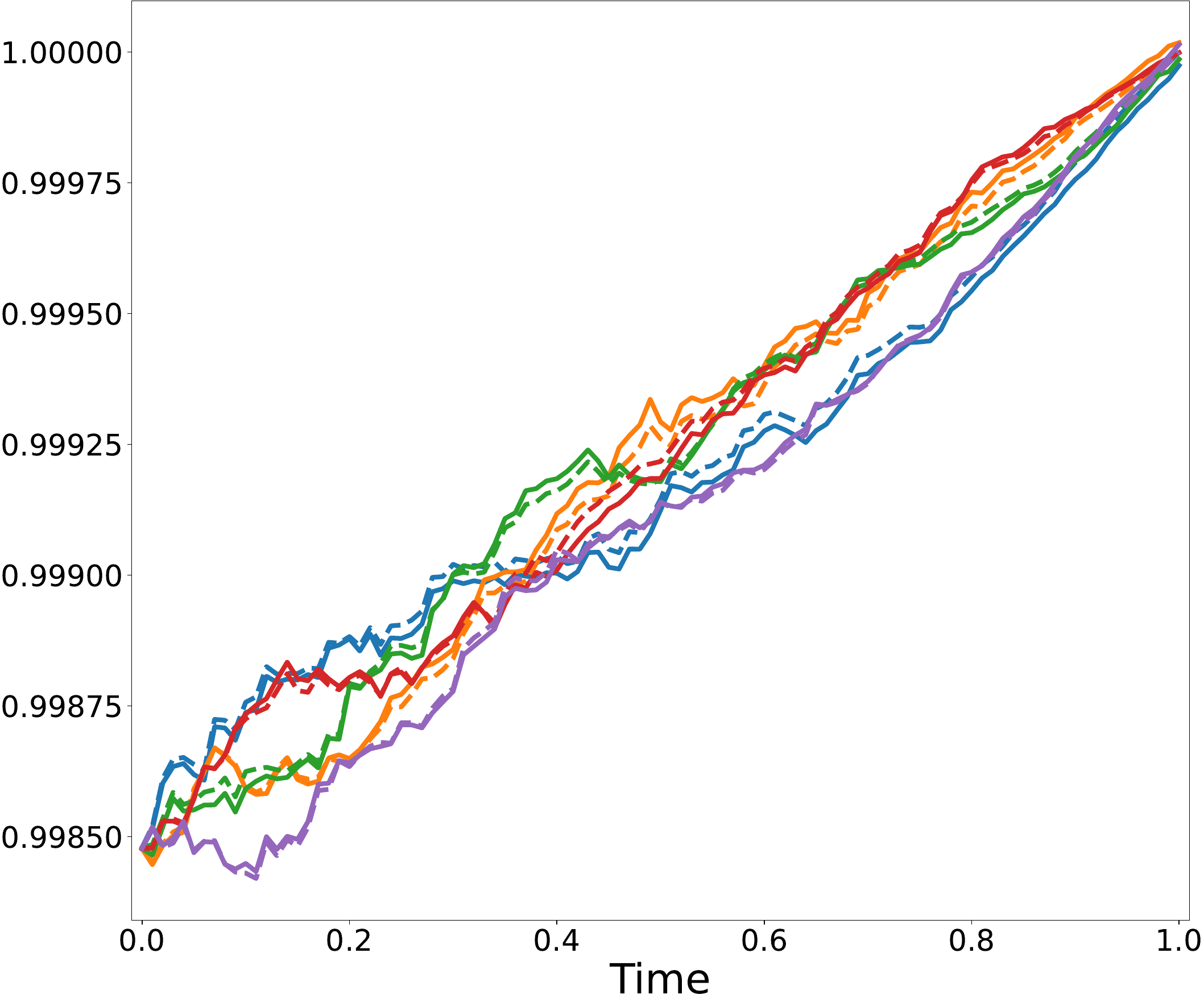} \\	
			\begin{turn}{90}\textbf{\large $\boldsymbol{m = 20}$}\end{turn}&\includegraphics[width=0.31\textwidth, height =0.25\textwidth]{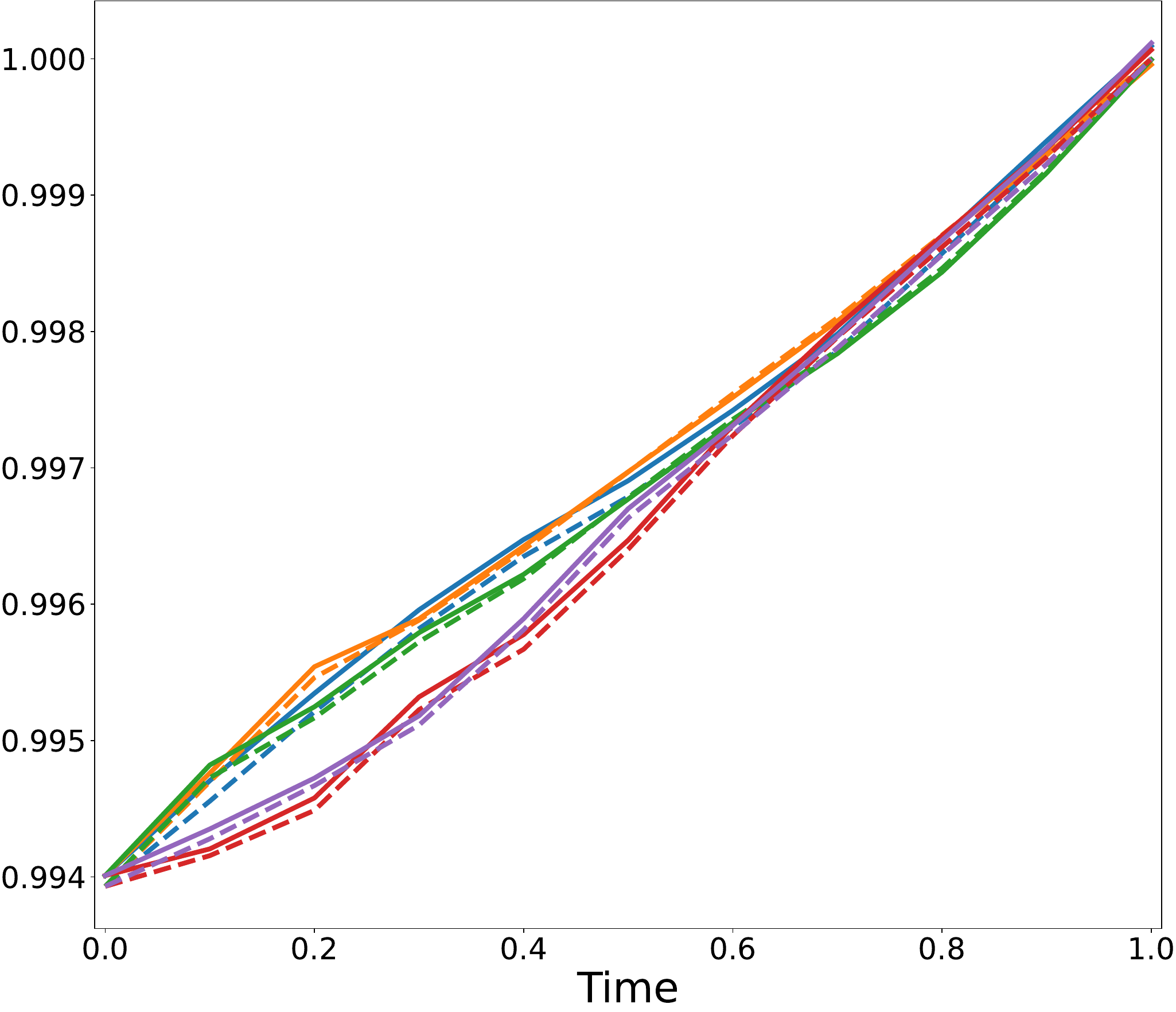}&
			\includegraphics[width=0.31\textwidth, height =0.25\textwidth]{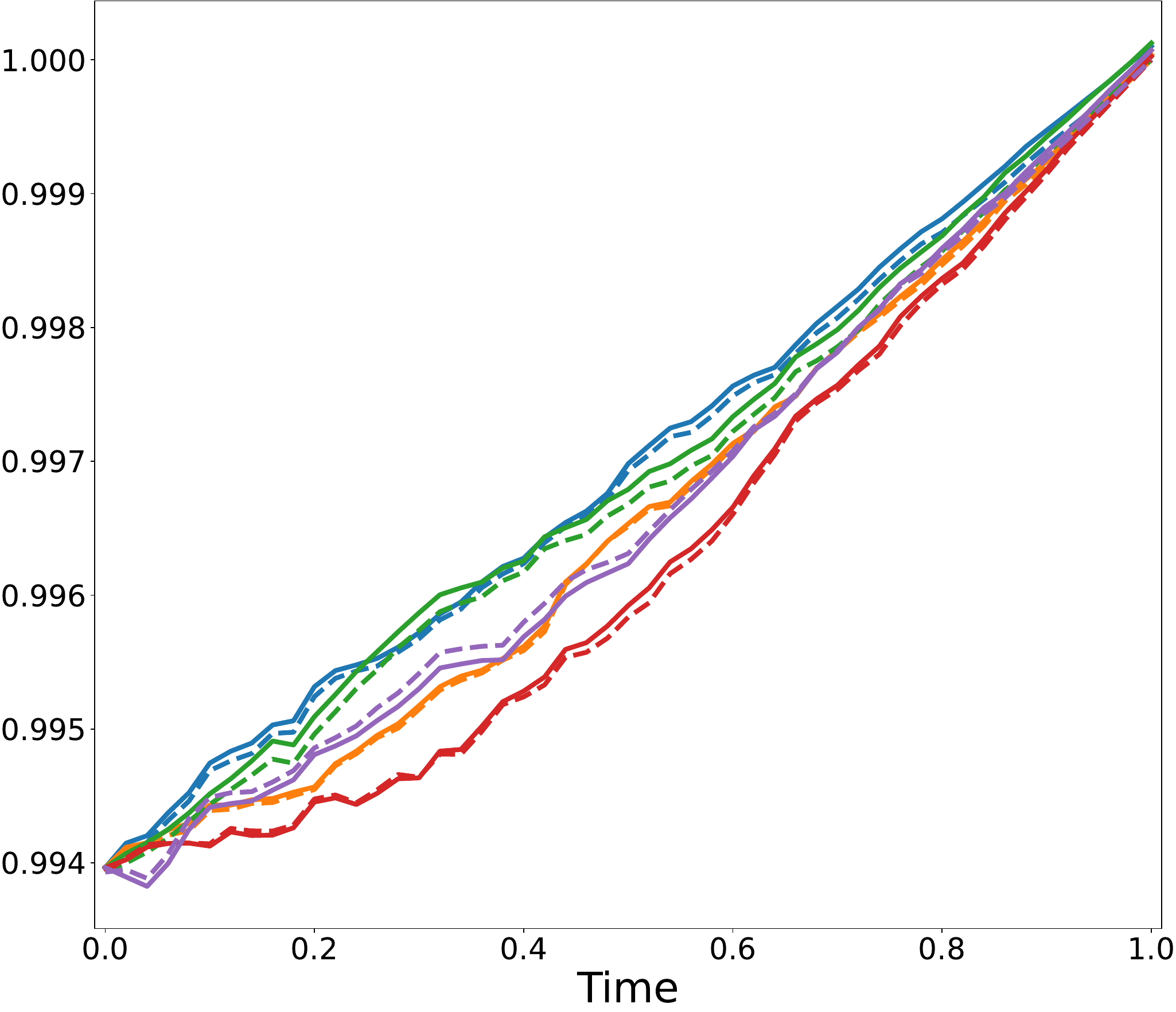} &	
			\includegraphics[width=0.31\textwidth, height =0.25\textwidth]{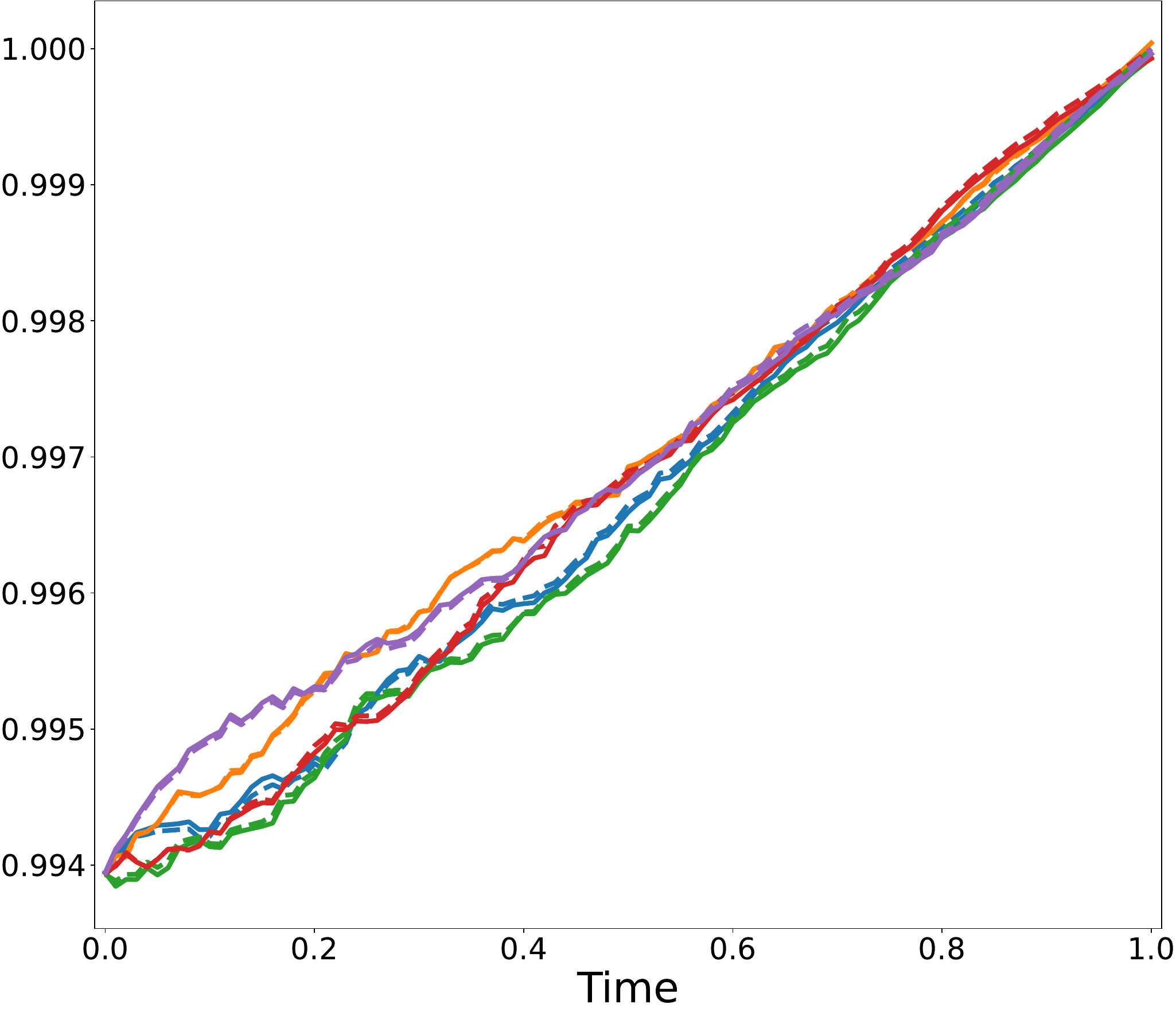}\\
			\begin{turn}{90}\textbf{\large $\boldsymbol{m = 100}$}\end{turn}&\includegraphics[width=0.31\textwidth, height =0.25\textwidth]{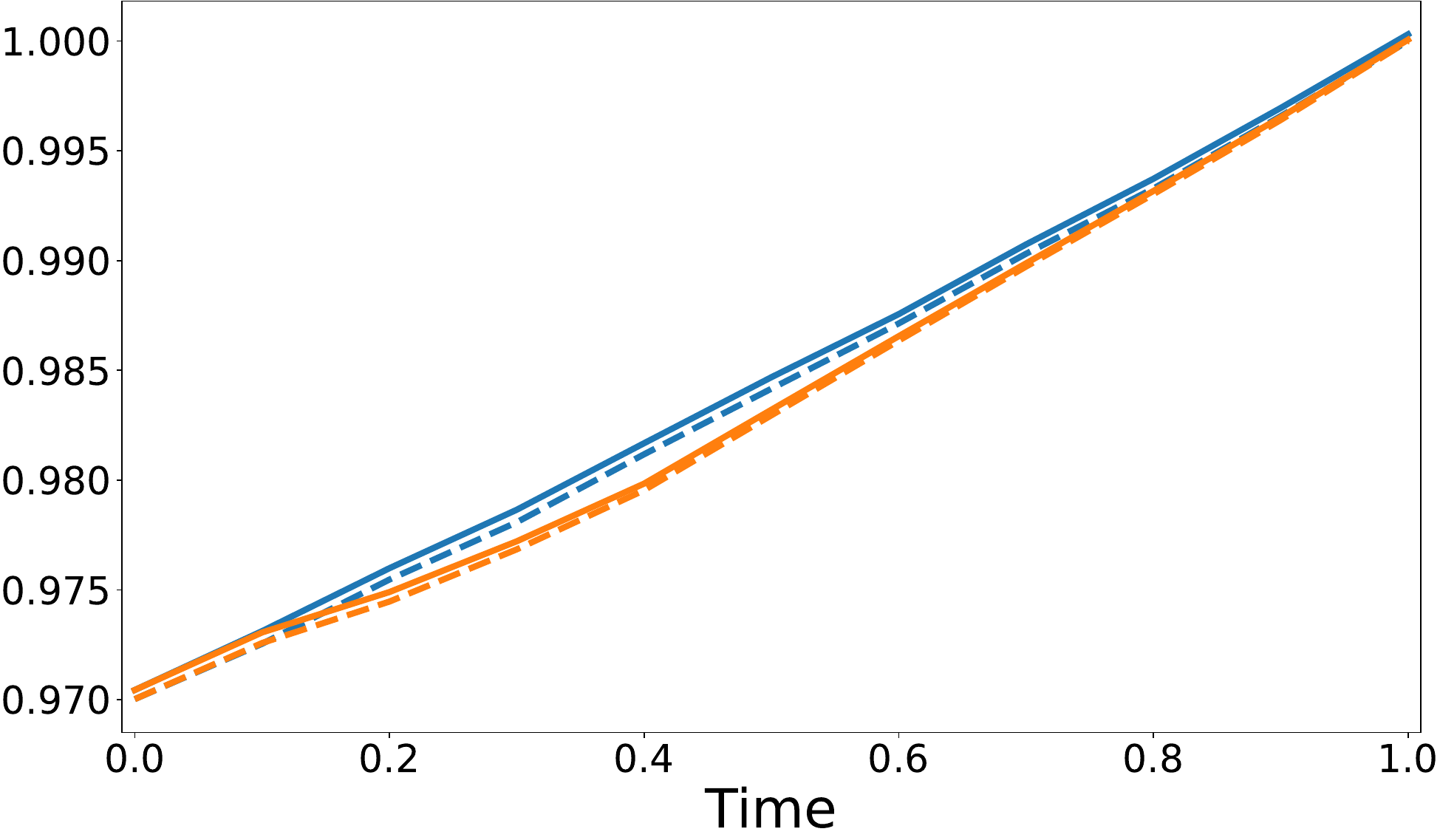}&
			\includegraphics[width=0.31\textwidth, height =0.25\textwidth]{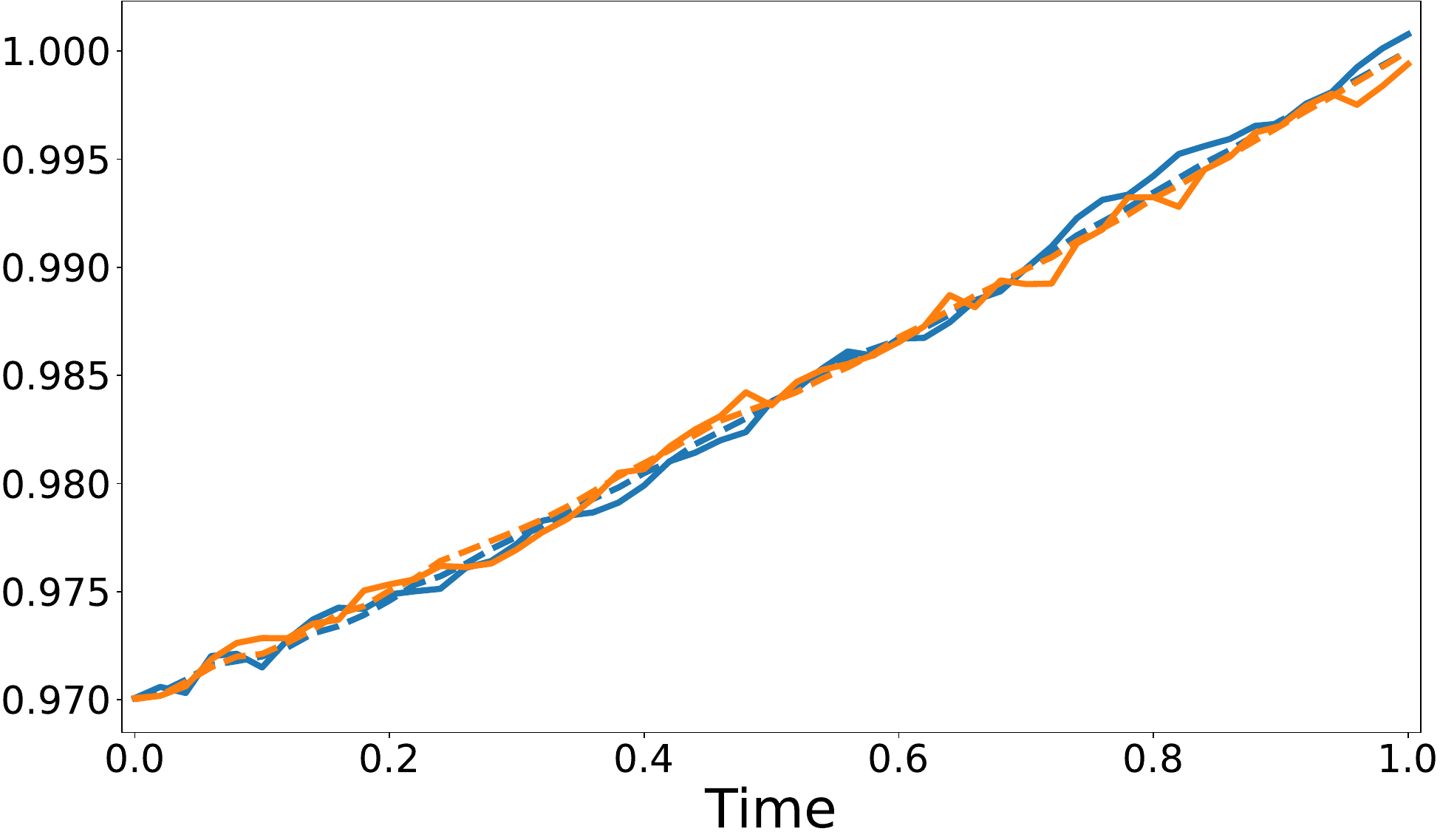} &	
			\includegraphics[width=0.31\textwidth, height =0.25\textwidth]{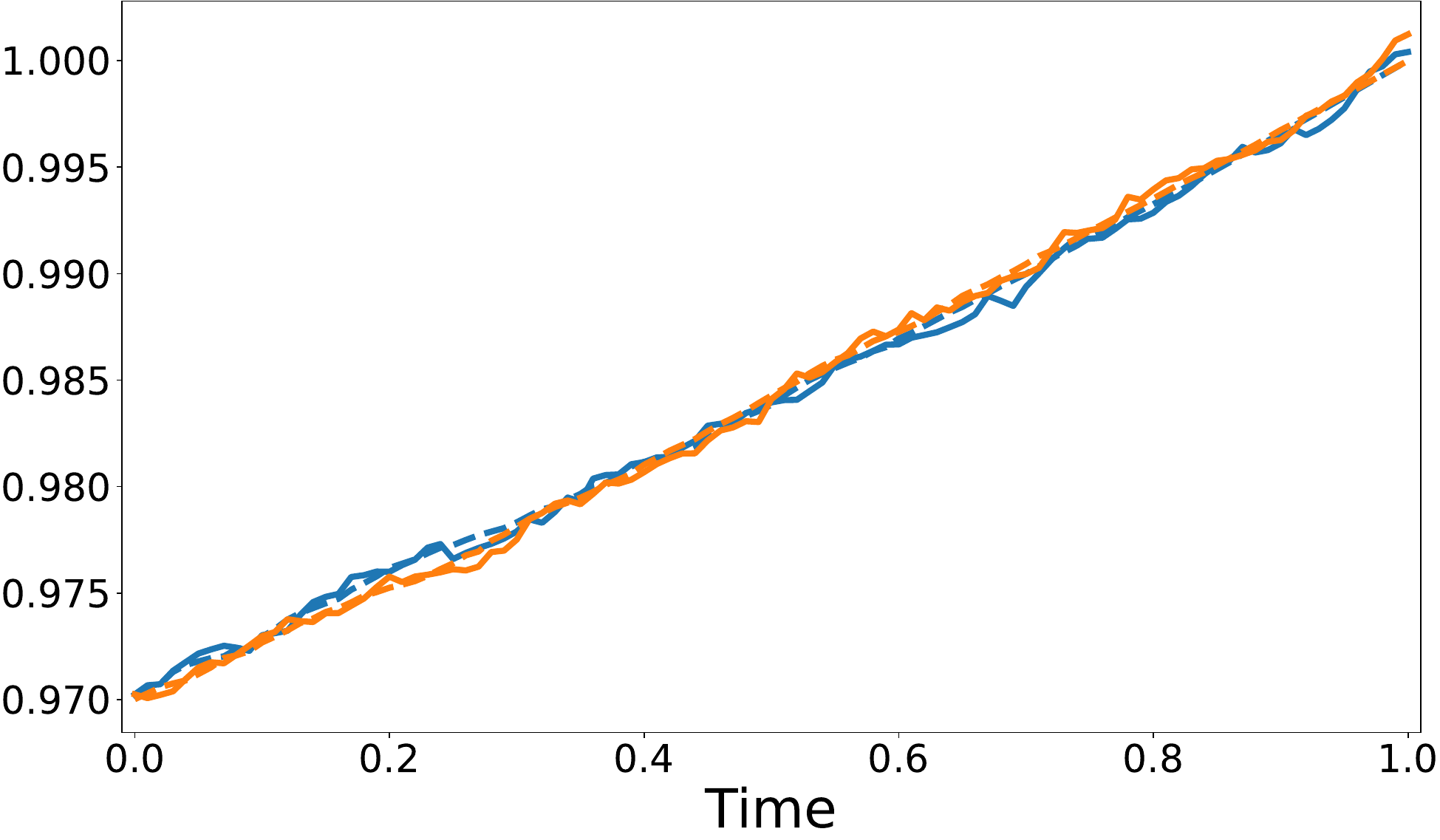}\\
			
		\end{tabular}
	}
	\caption{{\color{black}Deep solver solution (solid line) and benchmark solution (dashed line) of the stochastic Riccati equation \eqref{eq:StochasticRiccati} in the interval $[0, 1]$ for some random samples, and for the different experiment configurations presented in Table \ref{table:resultsMV}. For $m=1, 5, 20$ we consider $5$ trajectories, while for $m=100$ only $2$ samples are plotted to improve readability.} \label{plot:L}}
\end{figure}

\subsection{Deep local risk minimization results}
The numerical results for the deep local risk minimization are presented in Table \ref{table:resultsLR}, and  in Figure \ref{LR_logloss} we show the evolution of the logarithmic loss as a function of the deep solver iterations.

As before, for each portfolio dimension we compute the Monte Carlo (MC) price by simulating $10^5$ samples with $100$ points of time discretization under the minimal martingale measure as presented in \eqref{HestonQLR}. We then train and run the solver for solving the BSDE as in Section \ref{sec:deepLR}, obtaining an estimate of the call option price (BSDE solver price). We report the training time and the relative error, which is computed in the standard way starting from the previously simulated MC price. For the portfolio in dimension $1$, we also compute the option price with the Heath, Platen and Schweizer approach presented in Section \ref{localrisk1dim} and the corresponding relative error.

From Table \ref{table:resultsLR} we observe that in all the experiments the error for the option price is below $1.5\%$, and in particular below $1\%$ for $N=50$ and $N=100$. Clear convergence is also confirmed by the evolution of the logarithmic loss in Figure \ref{LR_logloss}.

For the portfolio of dimension $1$, we compute the call option price and hedging strategies evolution, namely units of cash account and shares of risky asset, both with the BSDE approach and with the Heath, Platen and Schweizer approach. We report them in Figure \ref{LR10}, \ref{LR50} and \ref{LR100}, respectively for $N=10$, $N=50$ and $N=100$. We observe that the solver is capturing well the paths of the option price, of the units and of the shares. 

A more complete picture is given by the MSE in Figure \ref{MSE}. Here we indeed observe that for the option price the MSE accumulates over time, but is clearly decreasing with the number of time steps, $N$, hence the error can be controlled by decreasing the mesh size in the time discretization. 
Less clear is the situation for the units of cash account and the shares of risky asset. However, the dashed lines representing the mean over time of the MSE show that also in these two cases the MSE improves from $N=10$ to $N=50$, but gets slightly worse for $N=100$. Again, this may be due to an insufficient training of the ANNs for high values of $N$.

\section{Conclusions}\label{sec:concl}
We implement a new methodology for quadratic hedging in high-dimensional settings. Our strategy involves {\color{black}formulating} local risk minimization and mean-variance hedging by means, respectively,  of BSDEs and systems of BSDEs in high dimension. We solve these (systems of) BSDEs by adapting the deep BSDE solver of \cite{weinar2017}, and we test the proposed methodology on Heston's  stochastic volatility model. 
For the one-dimensional setting, we validate our algorithms against known analytical approaches, showing high levels of precision in option pricing. We also obtain fairly accurate simulated paths for both the option price process and the optimal portfolio hedging strategy process. For high-dimensional settings, we benchmark the approach against Monte Carlo simulations, showing that high levels of precision are achieved also in these cases.
Our approach allows {\color{black} us} to hedge contingent claims in the presence of a high number of risk factors, a setting where traditional techniques are not feasible.

There are different directions for future work. On the one hand, it could be interesting to consider counterparty credit risk problems along the lines of \cite{gpr2022}, who consider only a complete market setting. Recently, the deep solver of \cite{weinar2017} has been extended to the case of BSDEs with jumps in \cite{gpp2022}, hence a natural idea would be to investigate whether the technique we propose could be also employed in a jump-diffusion setting. Since in \cite{lim2005} the stochastic control approach to mean-variance hedging that we are using is extended to include jumps, we conjecture that such an extension is possible. Finally, a further possibility would be to study the current approach in relation to the multivariate stochastic volatility model of \cite{dfgt2007} driven by Wishart processes.

\section*{Aknowledgements}
The authors are grateful to Martino Grasselli for useful remarks, \textcolor{black}{and to two anonymous referees for valuable suggestions}.

\begin{table}[tp]
	\centering
	\begin{tabular}{rrrr}
		\multicolumn{4}{c}{\huge \textbf{Local risk minimization}}\\
		&&& \\
		\toprule
		\multicolumn{1}{r}{\large \textbf{Portfolio dimension:} $\boldsymbol{1}$}&\multicolumn{3}{r}{\large \textbf{MC price:} $\boldsymbol{6.854}$}\\
		\midrule
		\textbf{Time steps}&$\boldsymbol{10}$ & $\boldsymbol{50}$&$\boldsymbol{100}$ \\
		\midrule
		\textbf{BSDE solver price}& $6.829$& $6.846$& $6.855$ \\
		\textbf{Relative error (\%)}& $\boldsymbol{0.360}$&  $\boldsymbol{0.120}$&$\boldsymbol{0.0162}$ \\
		\textbf{Training time (s)}&  $128$& $735$& $1546$\\
		\hdashline
		\textbf{PDE price}&  $6.850$& $6.850$& $6.850$ \\
		\textbf{Relative error (\%)}& $\boldsymbol{0.0488}$&  $\boldsymbol{0.0613}$&$\boldsymbol{0.0618}$ \\
		\midrule
		\multicolumn{1}{r}{\large \textbf{Portfolio dimension:} $\boldsymbol{5}$}&\multicolumn{3}{r}{\large \textbf{MC price:} $\boldsymbol{15.412}$}\\
		\midrule
		\textbf{Time steps}&$\boldsymbol{10}$ & $\boldsymbol{50}$&$\boldsymbol{100}$ \\
		\midrule
		\textbf{BSDE solver price}& $15.197$& $15.366$& $15.365$ \\
		\textbf{Relative error (\%)}& $\boldsymbol{1.395}$&  $\boldsymbol{0.299}$&$\boldsymbol{0.309}$ \\
		\textbf{Training time (s)}&  $255$& $1250$& $2866$\\
		\midrule
		\multicolumn{1}{r}{\large \textbf{Portfolio dimension:} $\boldsymbol{20}$}&\multicolumn{3}{r}{\large \textbf{MC price:} $\boldsymbol{30.761}$}\\
		\midrule
		\textbf{Time steps}&$\boldsymbol{10}$ & $\boldsymbol{50}$&$\boldsymbol{100}$ \\
		\midrule
		\textbf{BSDE solver price}& $30.704$& $30.783$& $30.828$ \\
		\textbf{Relative error (\%)}& $\boldsymbol{1.322}$&  $\boldsymbol{0.568}$&$\boldsymbol{0.218}$ \\
		\textbf{Training time (s)}&  $418$& $1993$& $3660$\\
		\midrule
		\multicolumn{1}{r}{\large \textbf{Portfolio dimension:} $\boldsymbol{100}$}&\multicolumn{3}{r}{\large \textbf{MC price:} $\boldsymbol{68.950}$}\\
		\midrule
		\textbf{Time steps}&$\boldsymbol{10}$ & $\boldsymbol{50}$&$\boldsymbol{100}$ \\
		\midrule
		\textbf{BSDE solver price}& $68.269$& $68.427$& $69.020$ \\
		\textbf{Relative error (\%)}& $\boldsymbol{0.988}$&  $\boldsymbol{0.758}$&$\boldsymbol{0.101}$ \\
		\textbf{Training time (s)}&  $1772$& $9096$& $16527$\\
		
		\bottomrule
	\end{tabular}
	\caption{Local risk minimization results for different portfolio dimensions and different number of total time steps in the discretization grid. For each configuration, we compute the Monte Carlo (MC) price by simulating $10^5$ samples under the minimal martingale measure, and we use it to compute the relative error in the classical way. For the portfolio with one risky asset, we report the price obtained with the benchmark approach via PDE presented in Appendix \ref{localrisk1dim}. \label{table:resultsLR}}
\end{table}

\begin{figure}[tp]
	\resizebox{1\textwidth}{!}{
		\begin{tabular}{@{}>{\centering\arraybackslash}m{0.04\textwidth}@{}>{\centering\arraybackslash}m{0.32\textwidth}@{}>{\centering\arraybackslash}m{0.32\textwidth}@{}>{\centering\arraybackslash}m{0.32\textwidth}@{}}
			\multicolumn{4}{c}{\huge \textbf{Local risk minimization}}\\
			&&& \\
			&\textbf{\large $\boldsymbol{N = 10}$} & \textbf{\large $\boldsymbol{N = 50}$} & \textbf{\large  $ \boldsymbol{N =100}$}\\
			&&& \\
			\begin{turn}{90}\textbf{\large $\boldsymbol{m = 1}$}\end{turn}&\includegraphics[width=0.32\textwidth, height =0.25\textwidth]{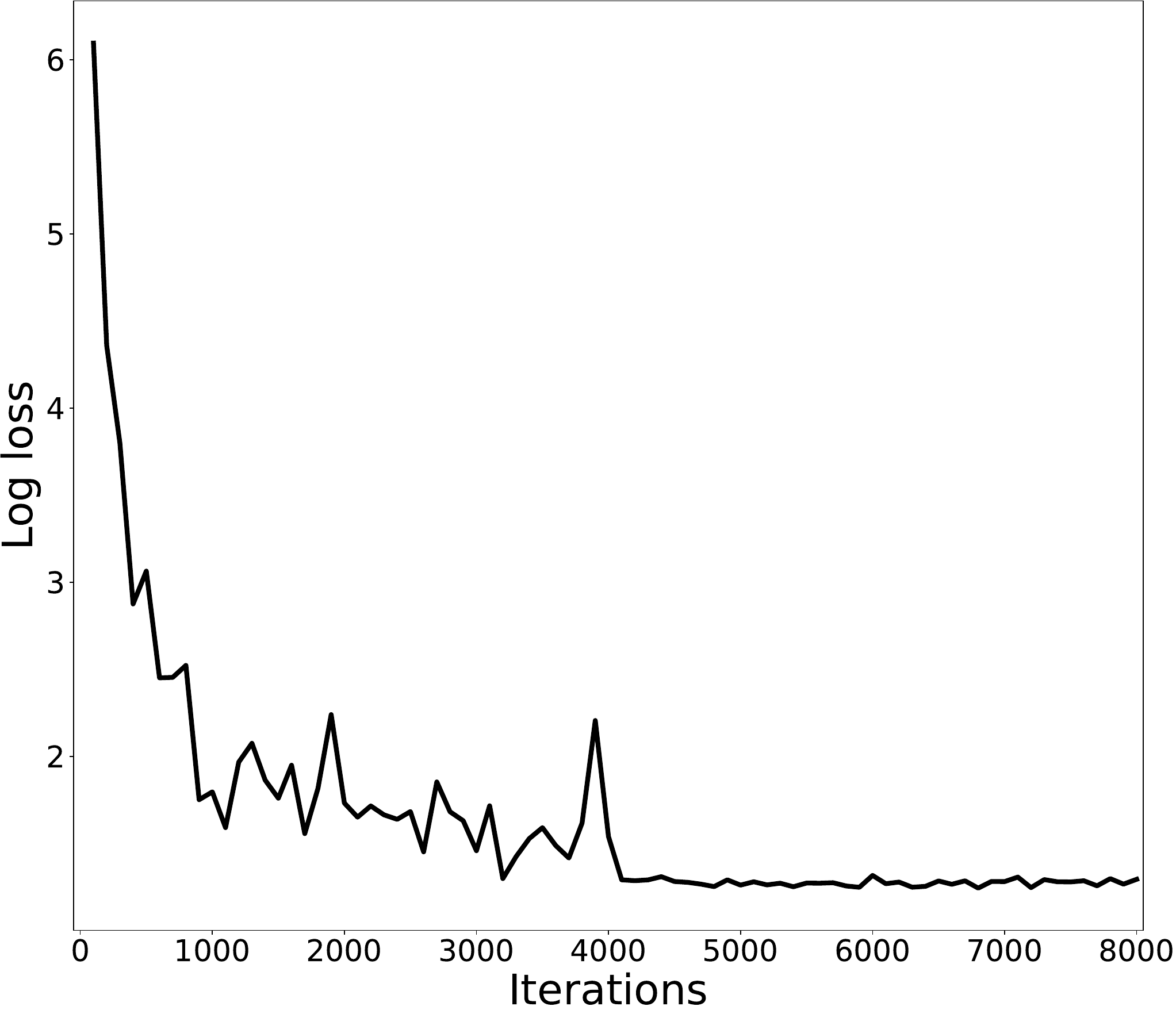} &
			\includegraphics[width=0.32\textwidth, height =0.25\textwidth]{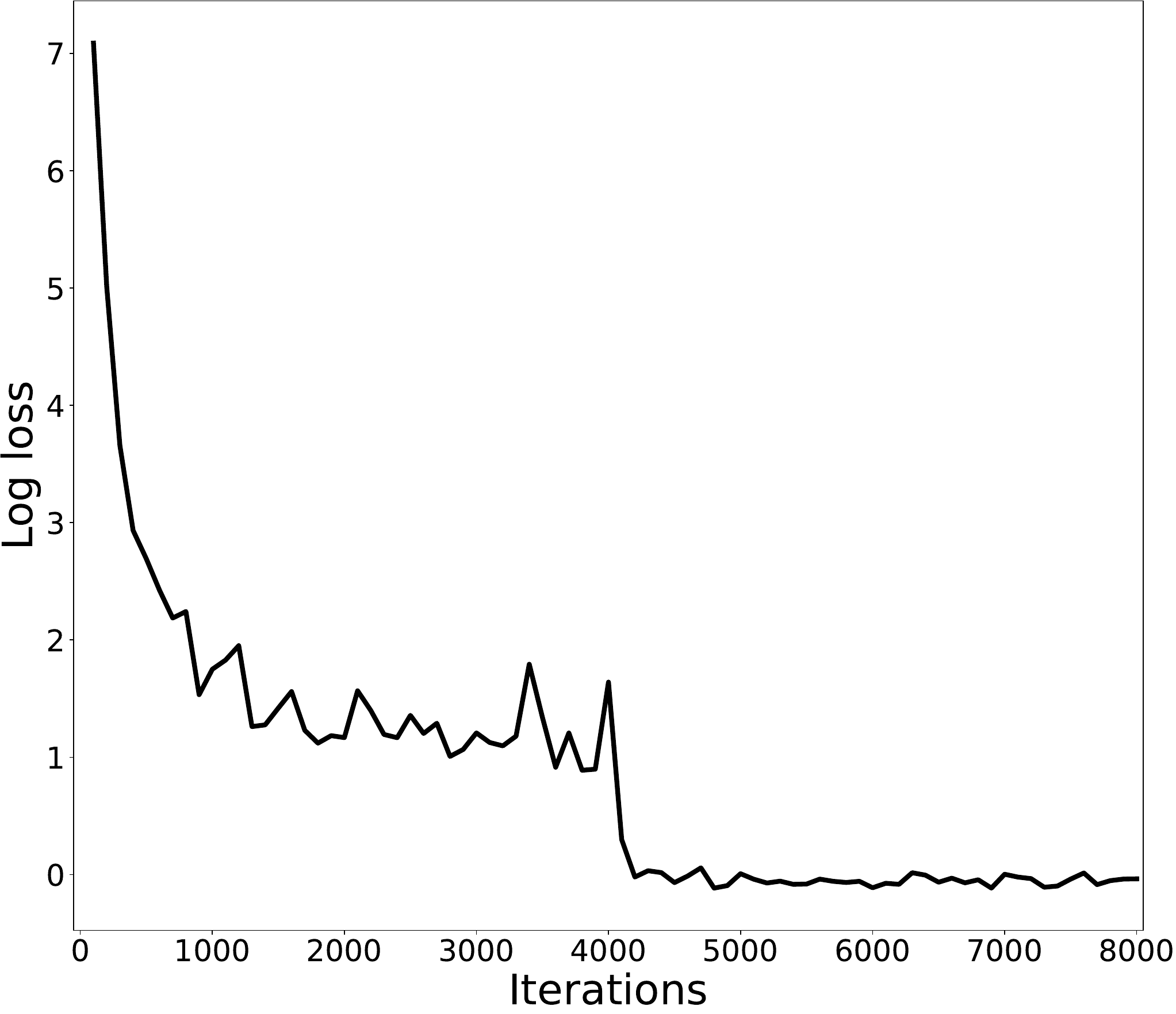}&
			\includegraphics[width=0.32\textwidth, height =0.25\textwidth]{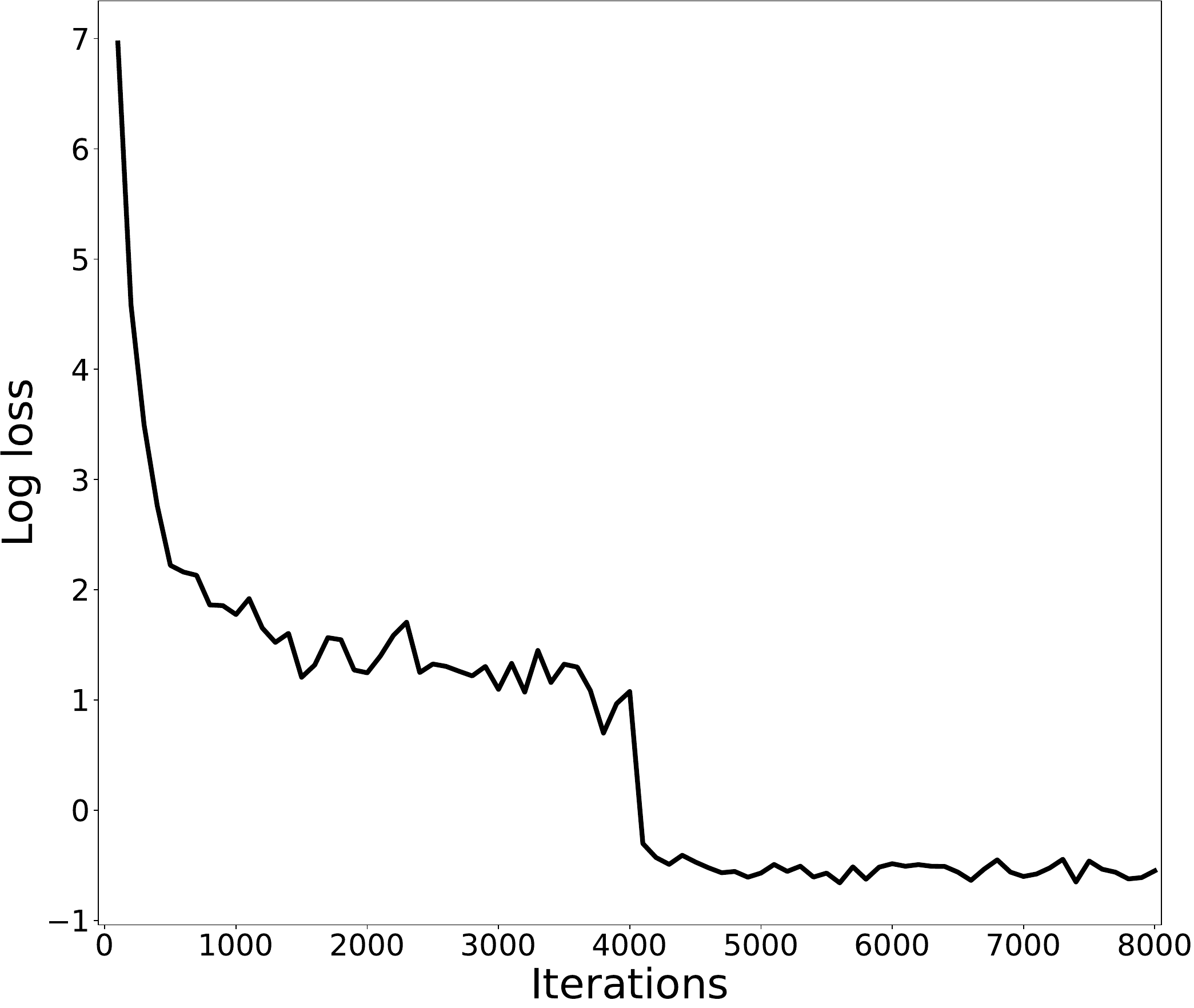} \\	
			\begin{turn}{90}\textbf{\large $\boldsymbol{m = 5}$}\end{turn}&\includegraphics[width=0.32\textwidth, height =0.25\textwidth]{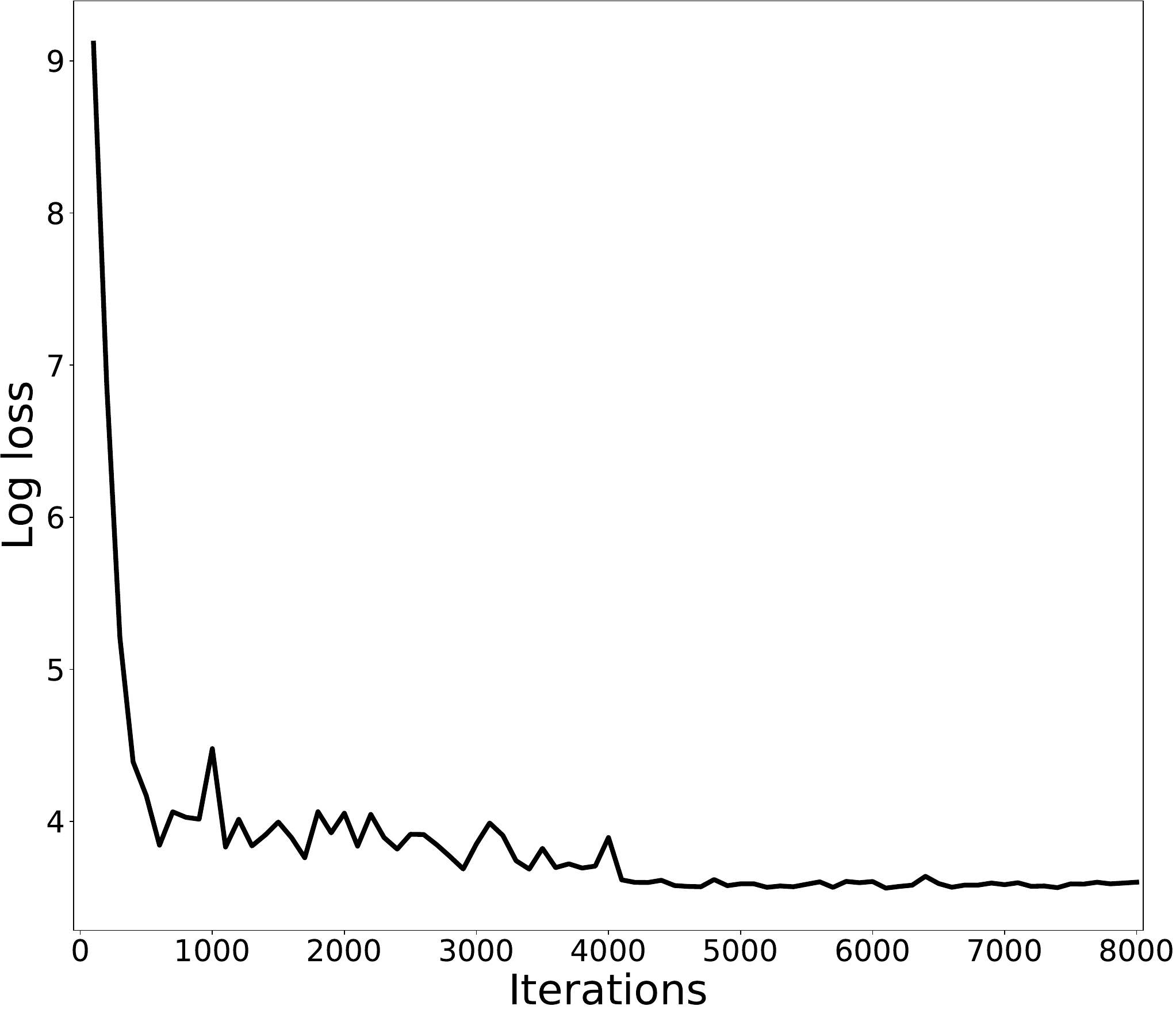} &
			\includegraphics[width=0.32\textwidth, height =0.25\textwidth]{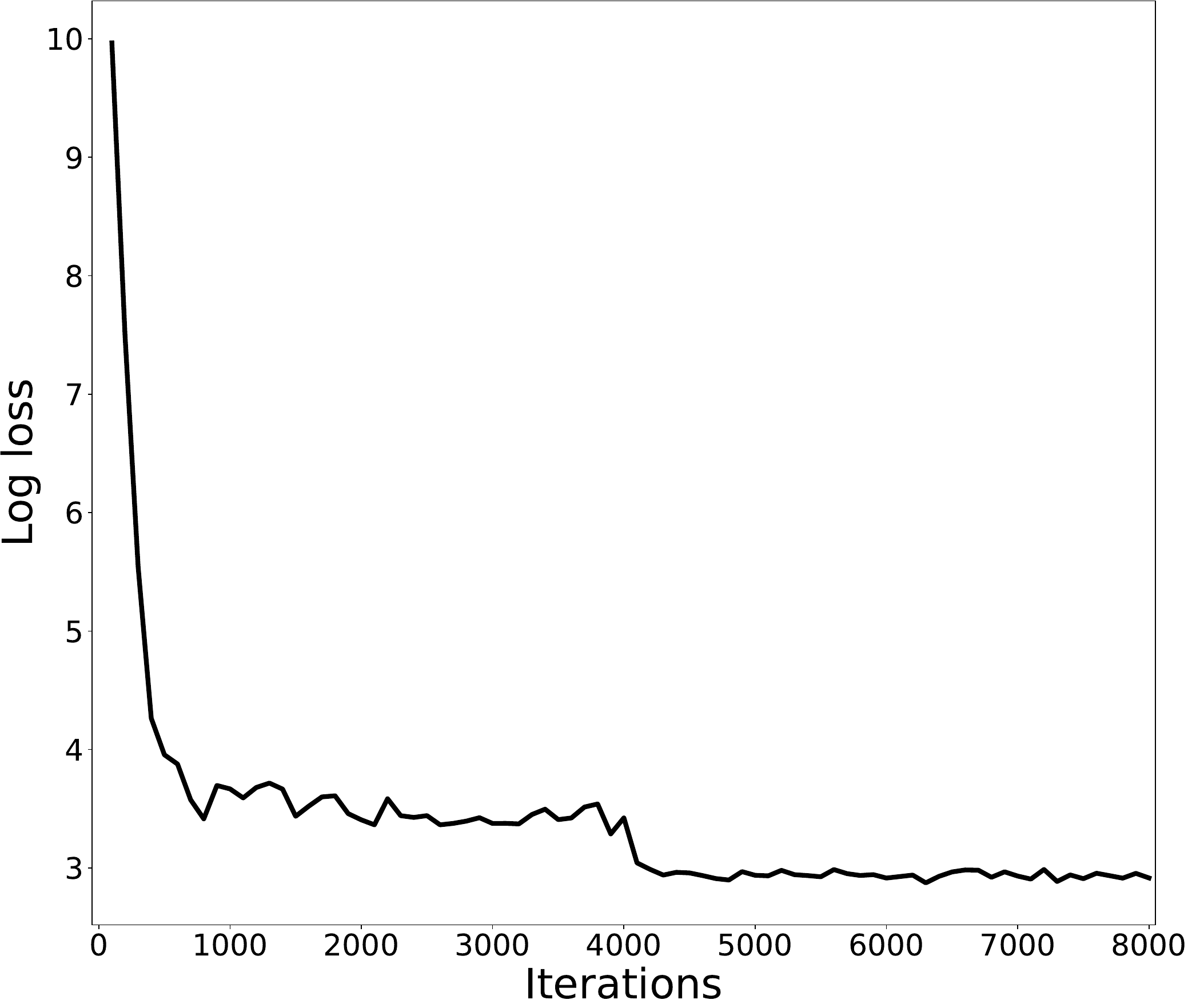}&
			\includegraphics[width=0.32\textwidth, height =0.25\textwidth]{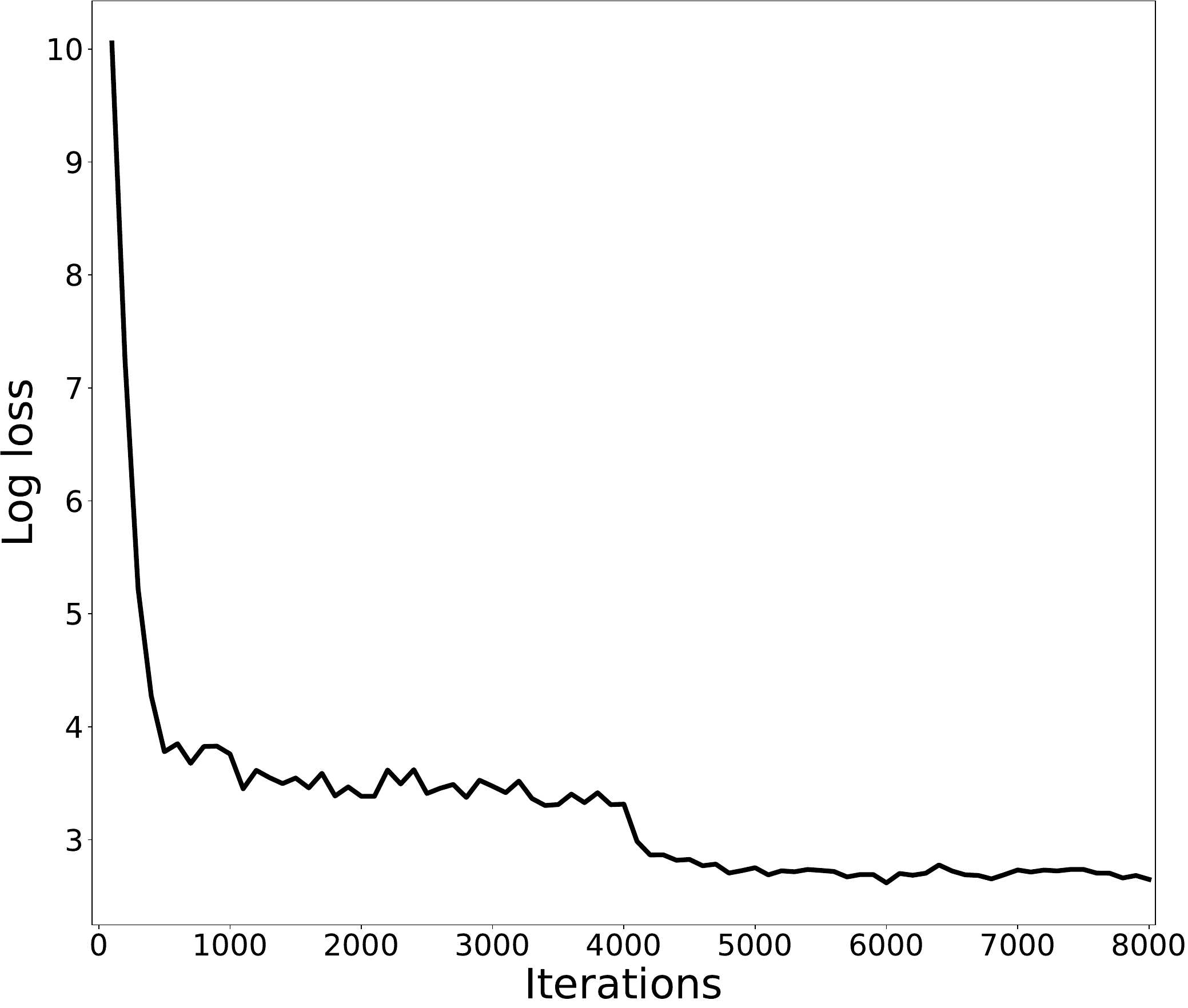} \\	
			\begin{turn}{90}\textbf{\large $\boldsymbol{m = 20}$}\end{turn}&\includegraphics[width=0.32\textwidth, height =0.25\textwidth]{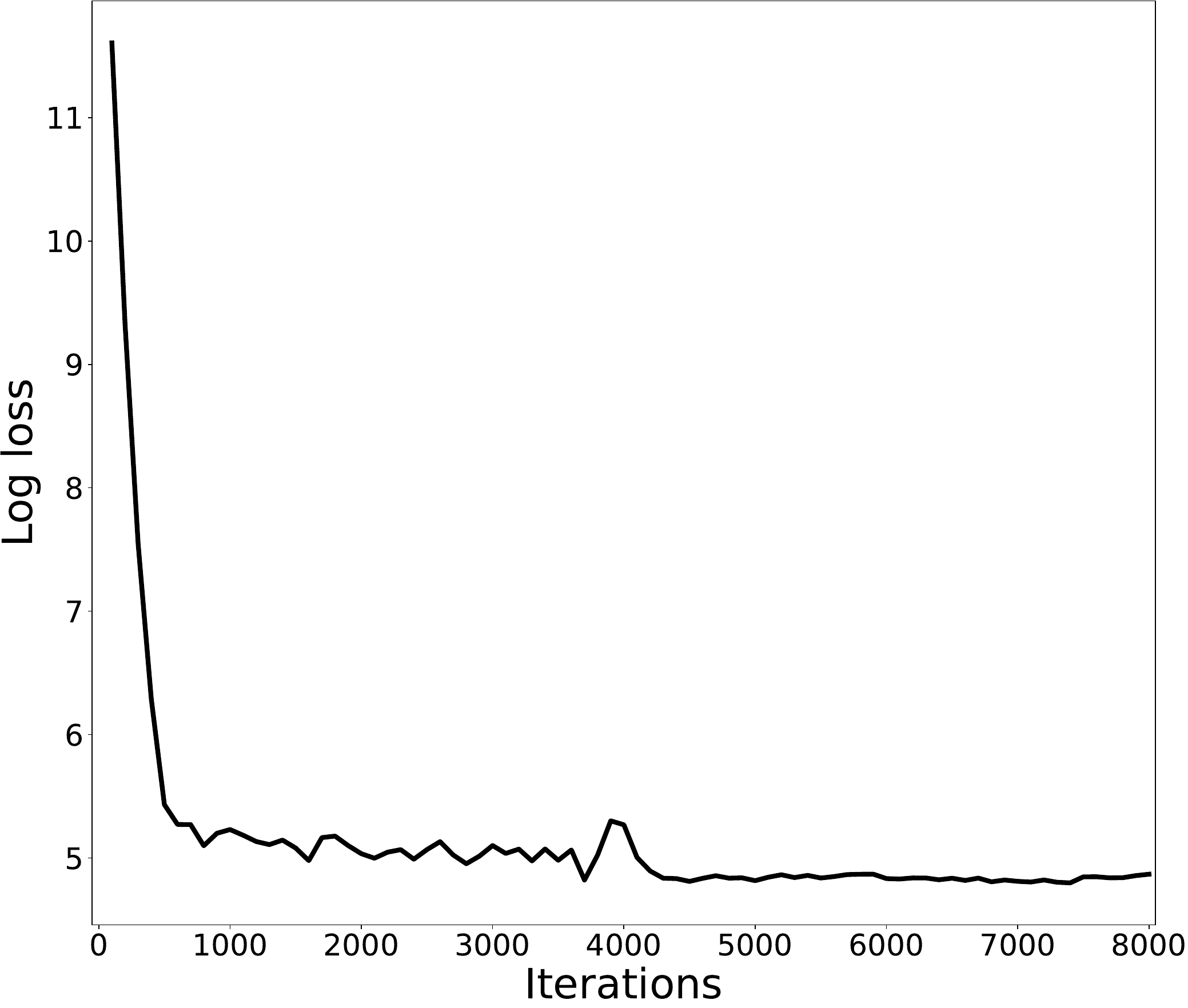}&
			\includegraphics[width=0.32\textwidth, height =0.25\textwidth]{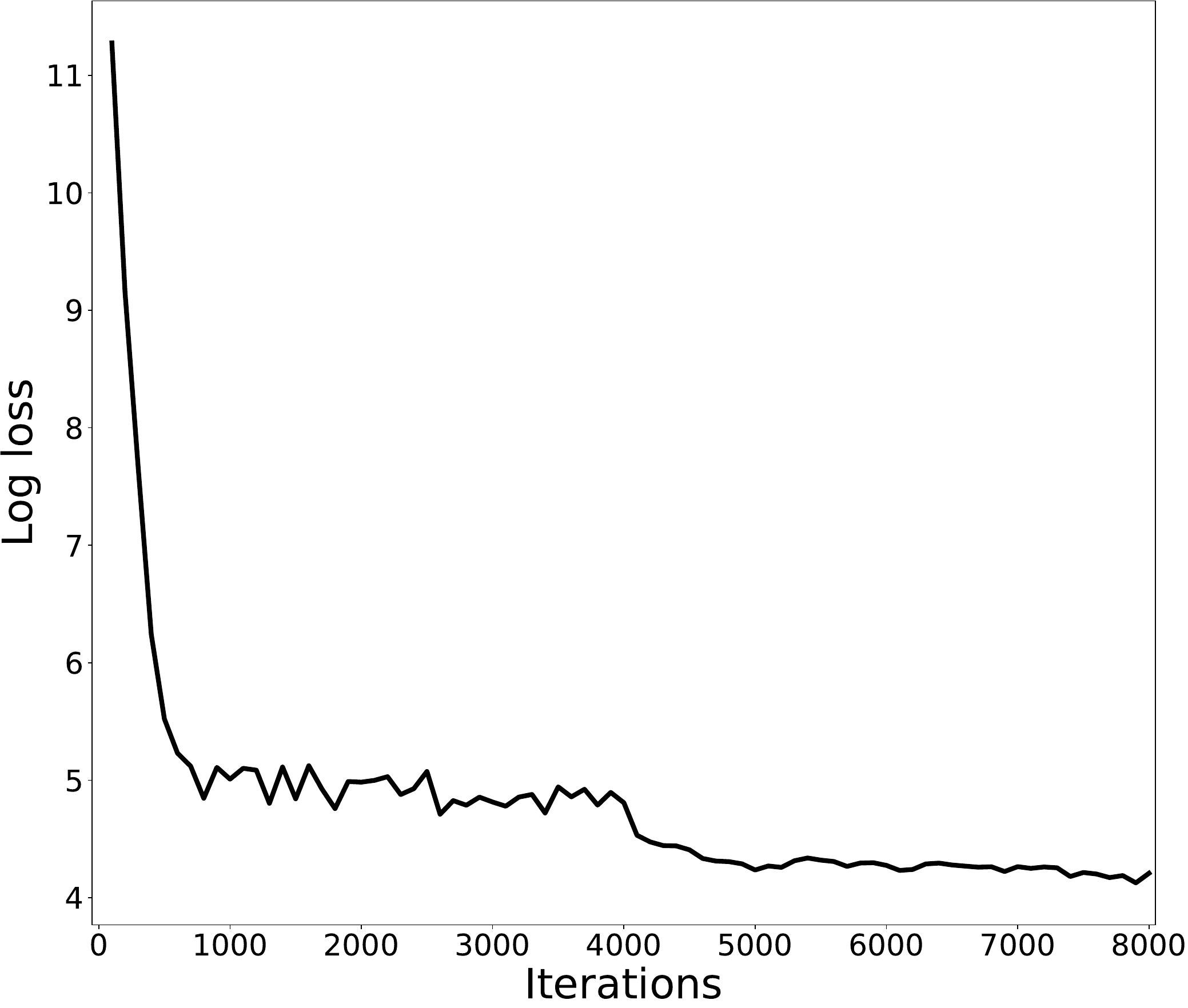} &	
			\includegraphics[width=0.32\textwidth, height =0.25\textwidth]{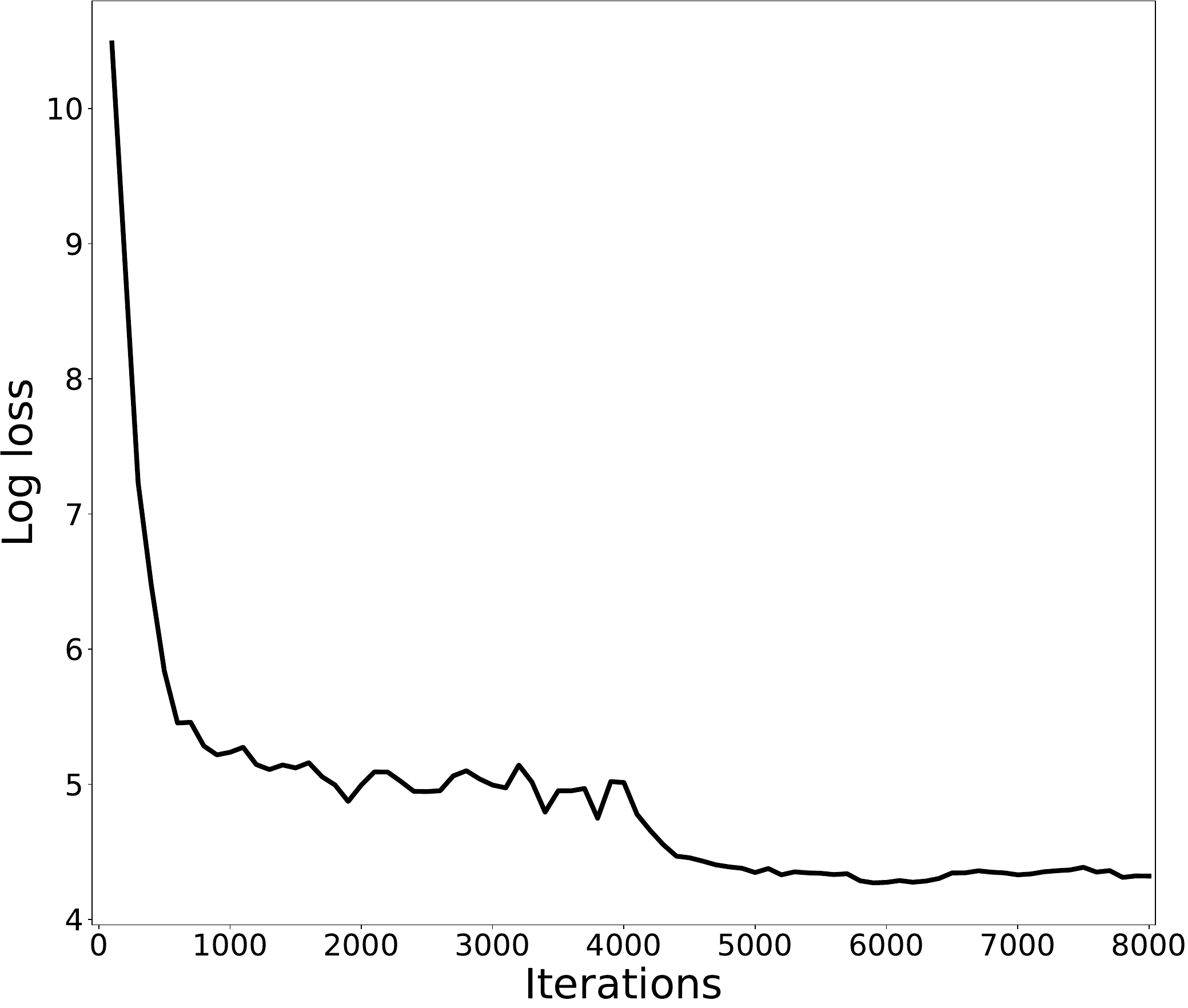}\\
			\begin{turn}{90}\textbf{\large $\boldsymbol{m = 100}$}\end{turn}&\includegraphics[width=0.32\textwidth, height =0.25\textwidth]{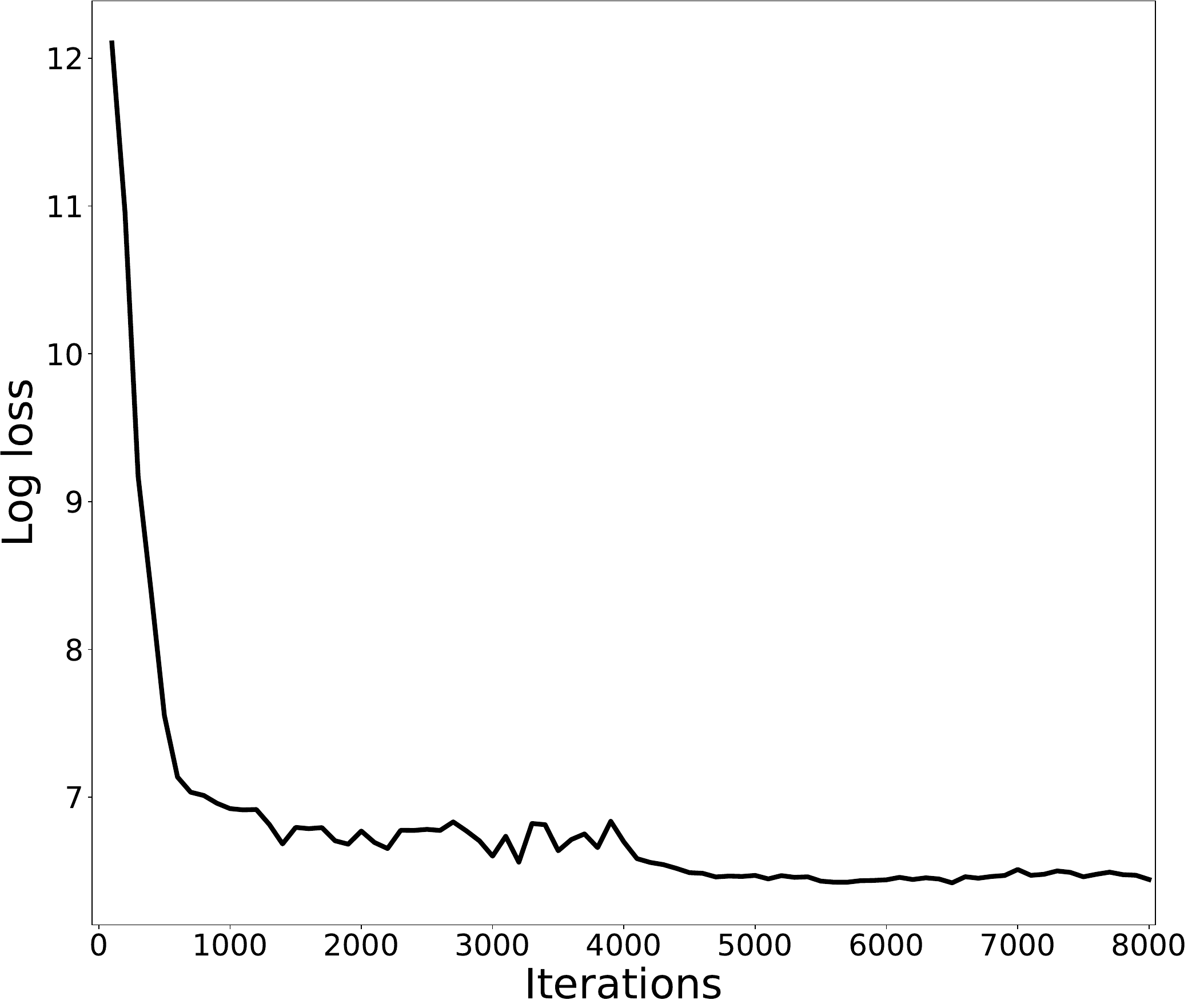}&
			\includegraphics[width=0.32\textwidth, height =0.25\textwidth]{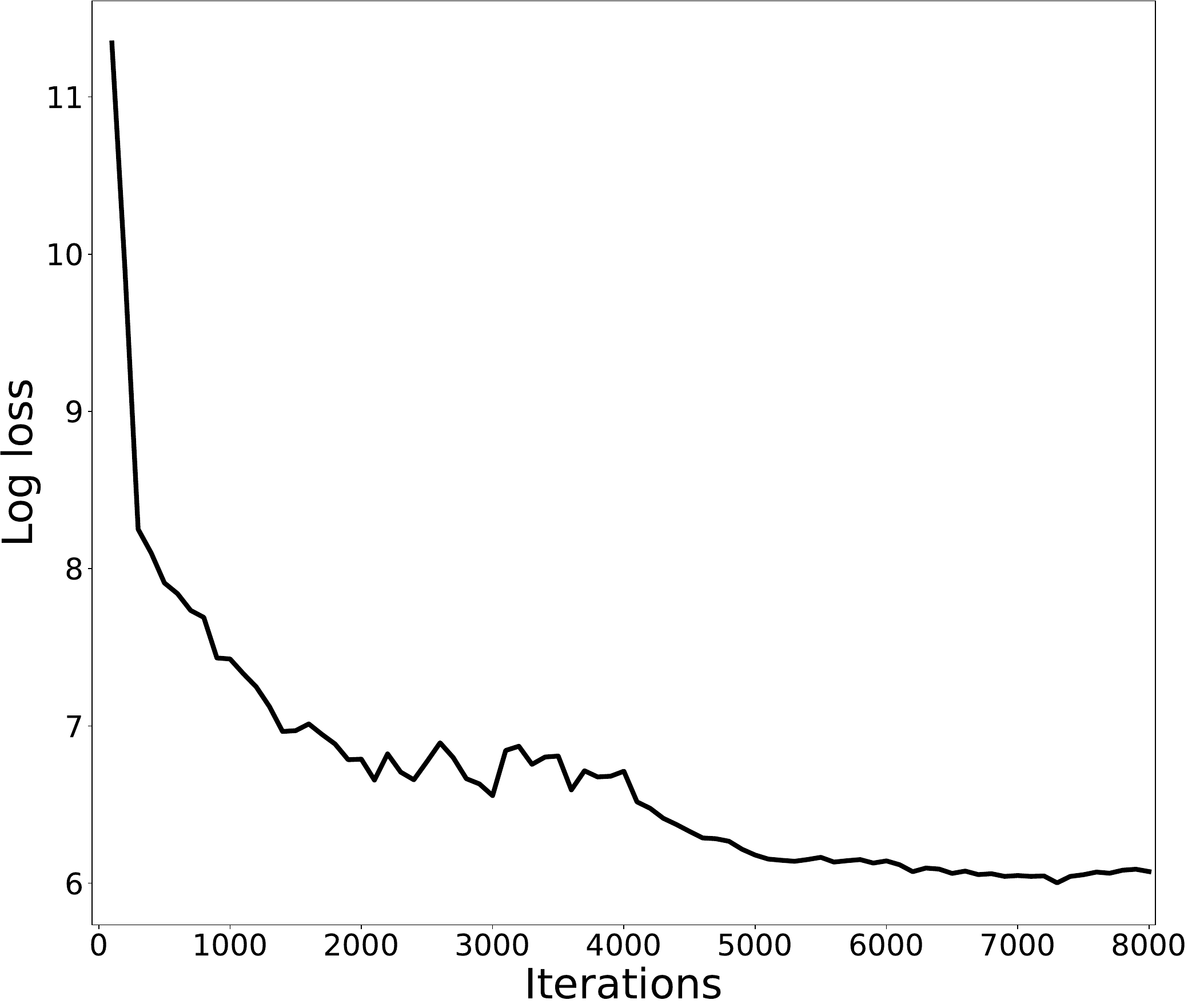} &	
			\includegraphics[width=0.32\textwidth, height =0.25\textwidth]{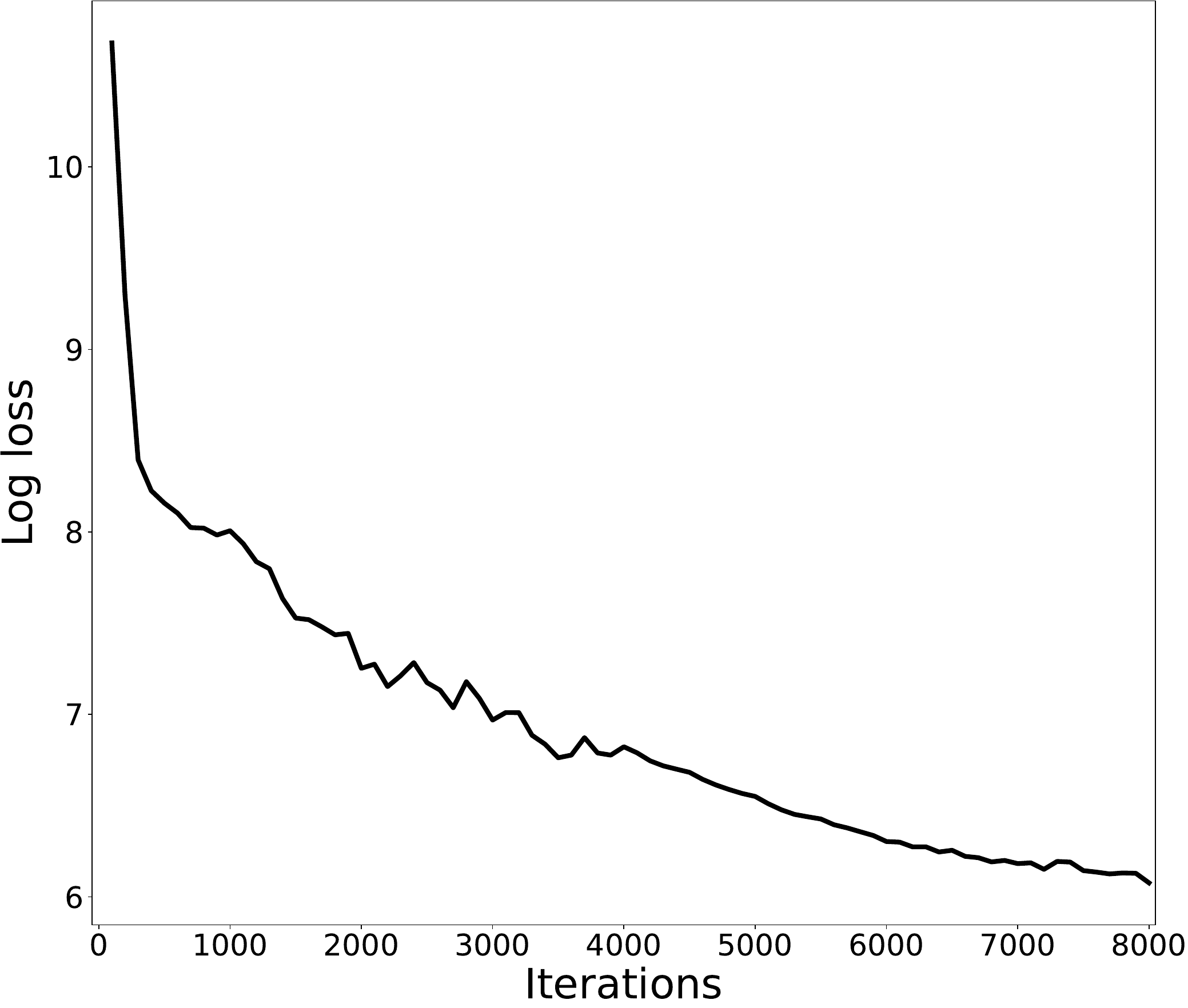}\\
			
		\end{tabular}
	}
	\caption{Logarithm of the loss functional as a function of the iteration number for the different experiment configurations presented in Table \ref{table:resultsLR}.\label{LR_logloss}}
\end{figure}

\begin{figure}[tp]
	\resizebox{1\textwidth}{!}{
		\begin{tabular}{@{}>{\centering\arraybackslash}m{0.5\textwidth}@{}>{\centering\arraybackslash}m{0.5\textwidth}@{}}
			\multicolumn{2}{c}{\huge \textbf{Local risk minimization: $\boldsymbol{N = 10}$}}\\
			& \\
			\includegraphics[width=0.48\textwidth, height =0.35\textwidth]{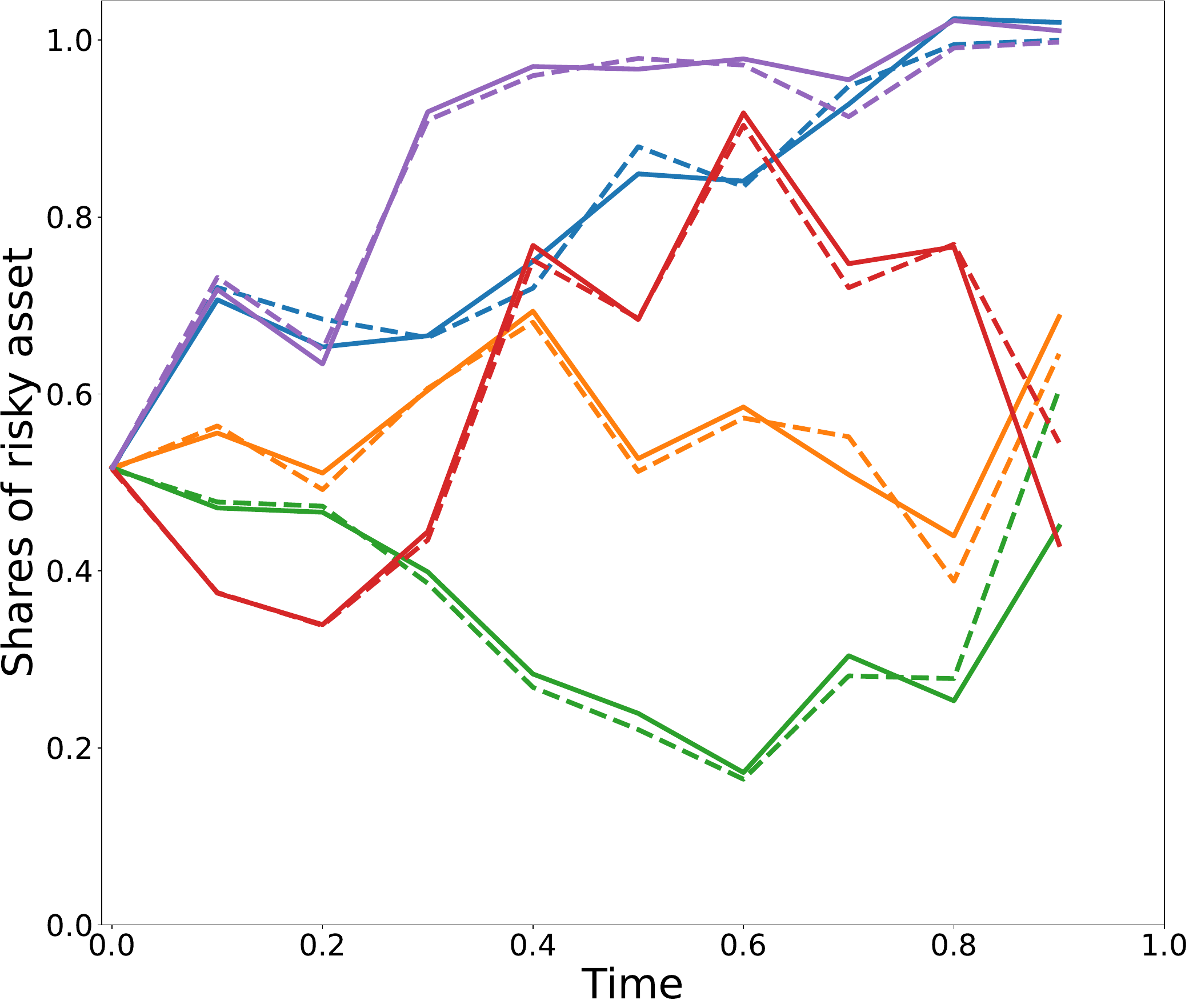} &
			\includegraphics[width=0.48\textwidth, height =0.35\textwidth]{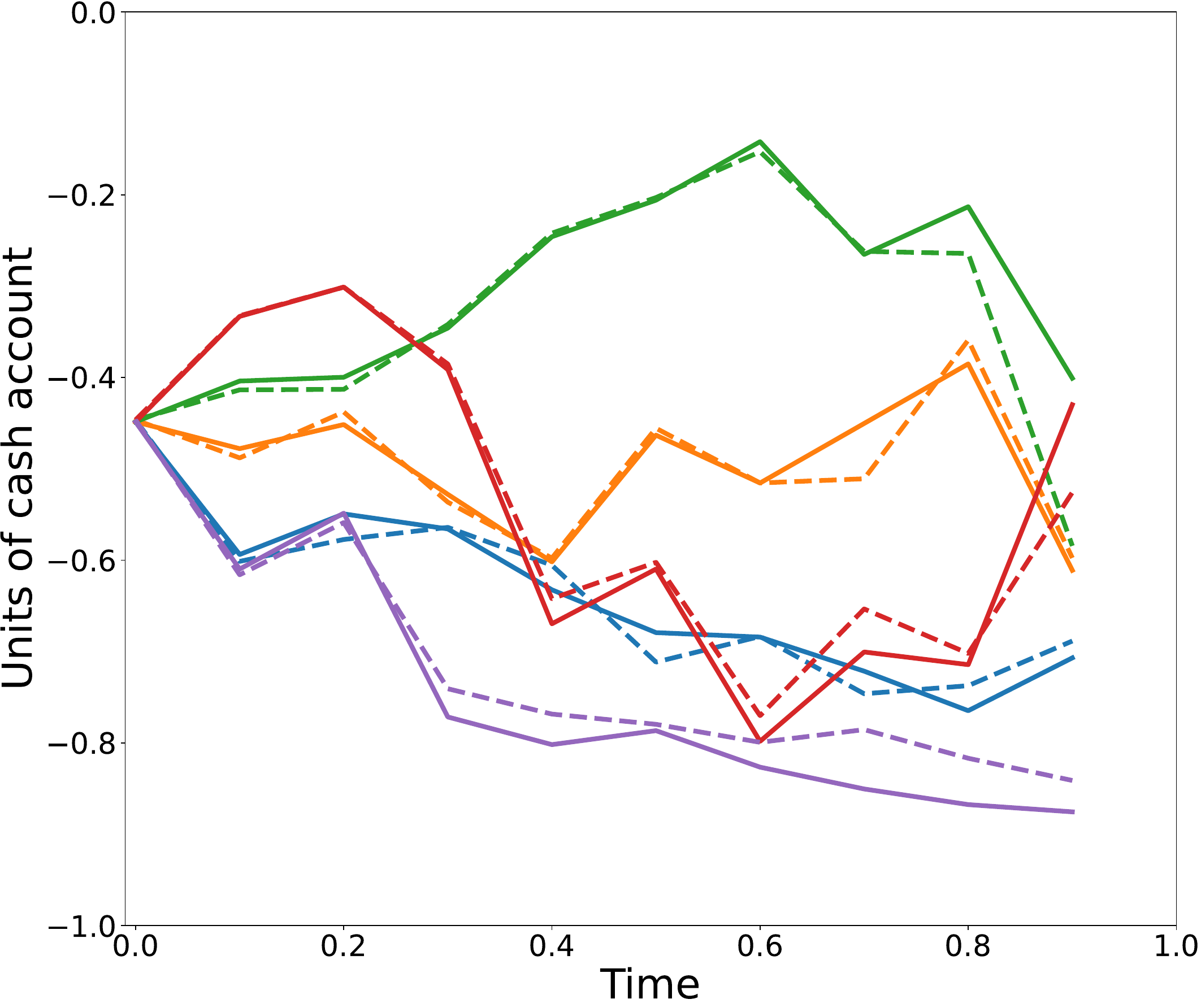}\\
			\multicolumn{2}{c}{\includegraphics[width=0.48\textwidth, height =0.35\textwidth]{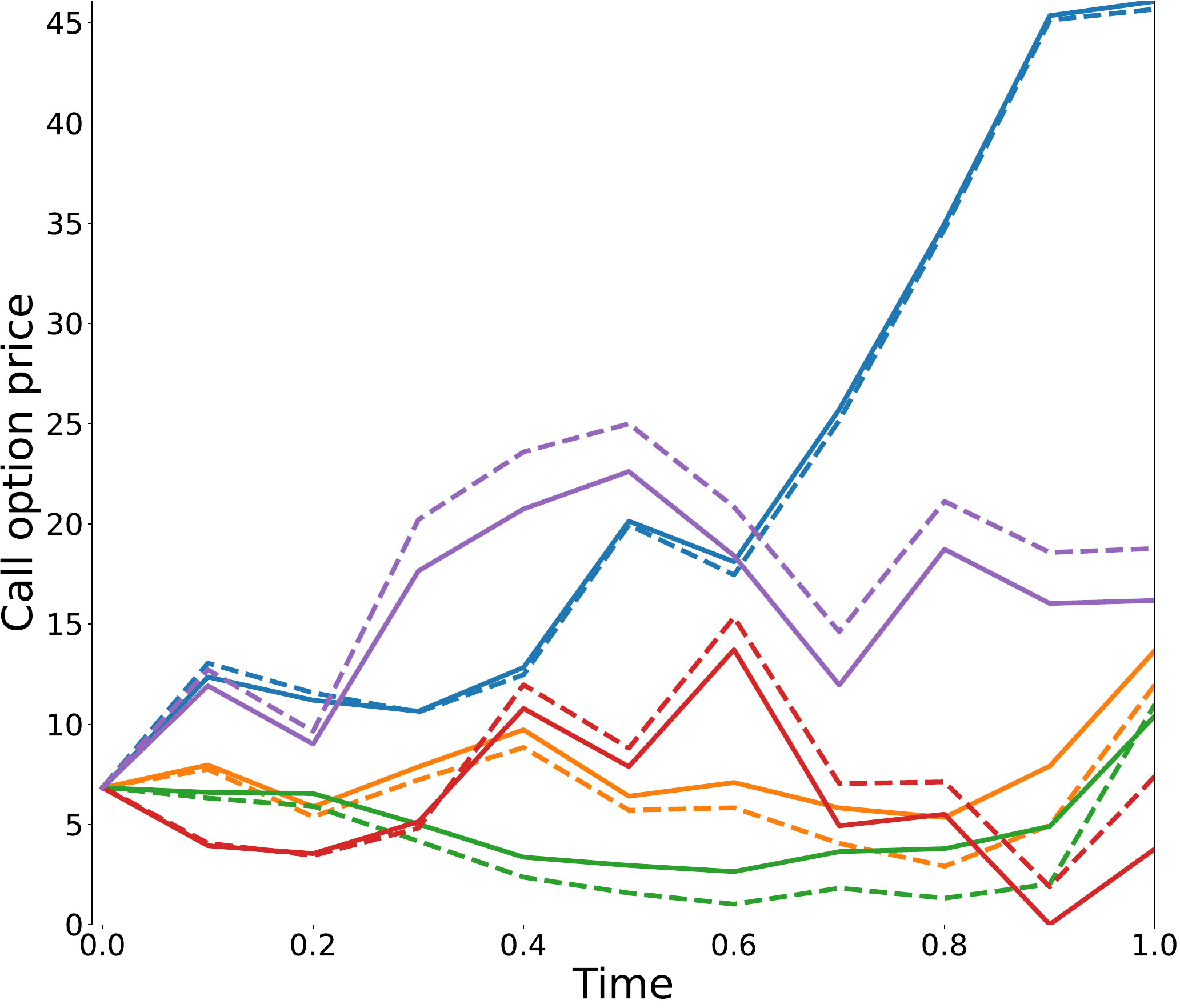}} \\
		\end{tabular}
	}
	\caption{Deep solver solution (solid line) and benchmark solution (dashed line) for the local risk minimization in a $10$ points discretization grid for $5$ random samples  in the interval $[0, 1]$. Upper panel: the shares of risky asset (left) and the units of cash account (right); lower panel: the call option price. \label{LR10}}
\end{figure}

\begin{figure}[tp]
	\resizebox{1\textwidth}{!}{
		\begin{tabular}{@{}>{\centering\arraybackslash}m{0.5\textwidth}@{}>{\centering\arraybackslash}m{0.5\textwidth}@{}}
			\multicolumn{2}{c}{\huge \textbf{Local risk minimization: $\boldsymbol{N = 50}$}}\\
			& \\
			\includegraphics[width=0.48\textwidth, height =0.35\textwidth]{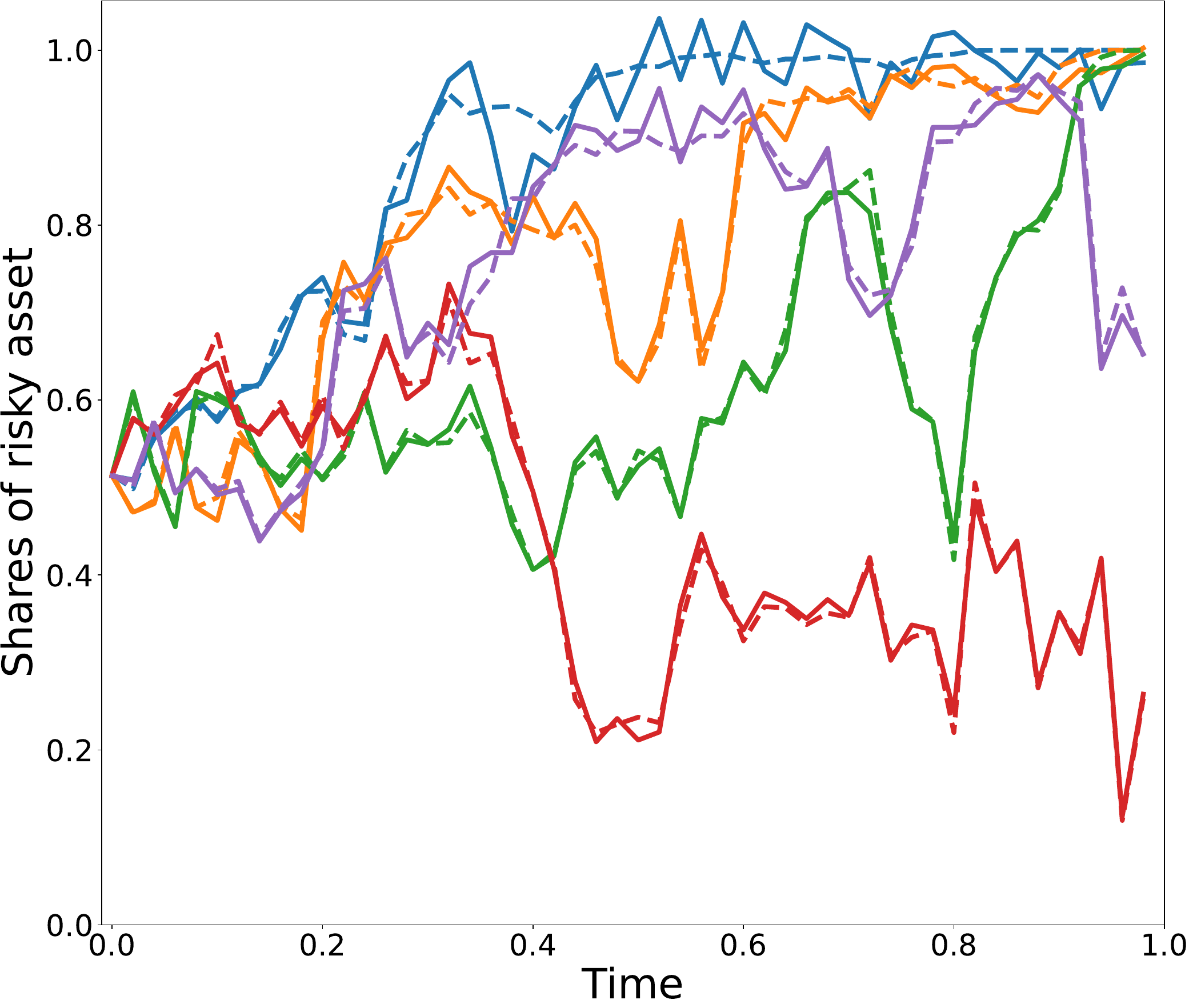} &
			\includegraphics[width=0.48\textwidth, height =0.35\textwidth]{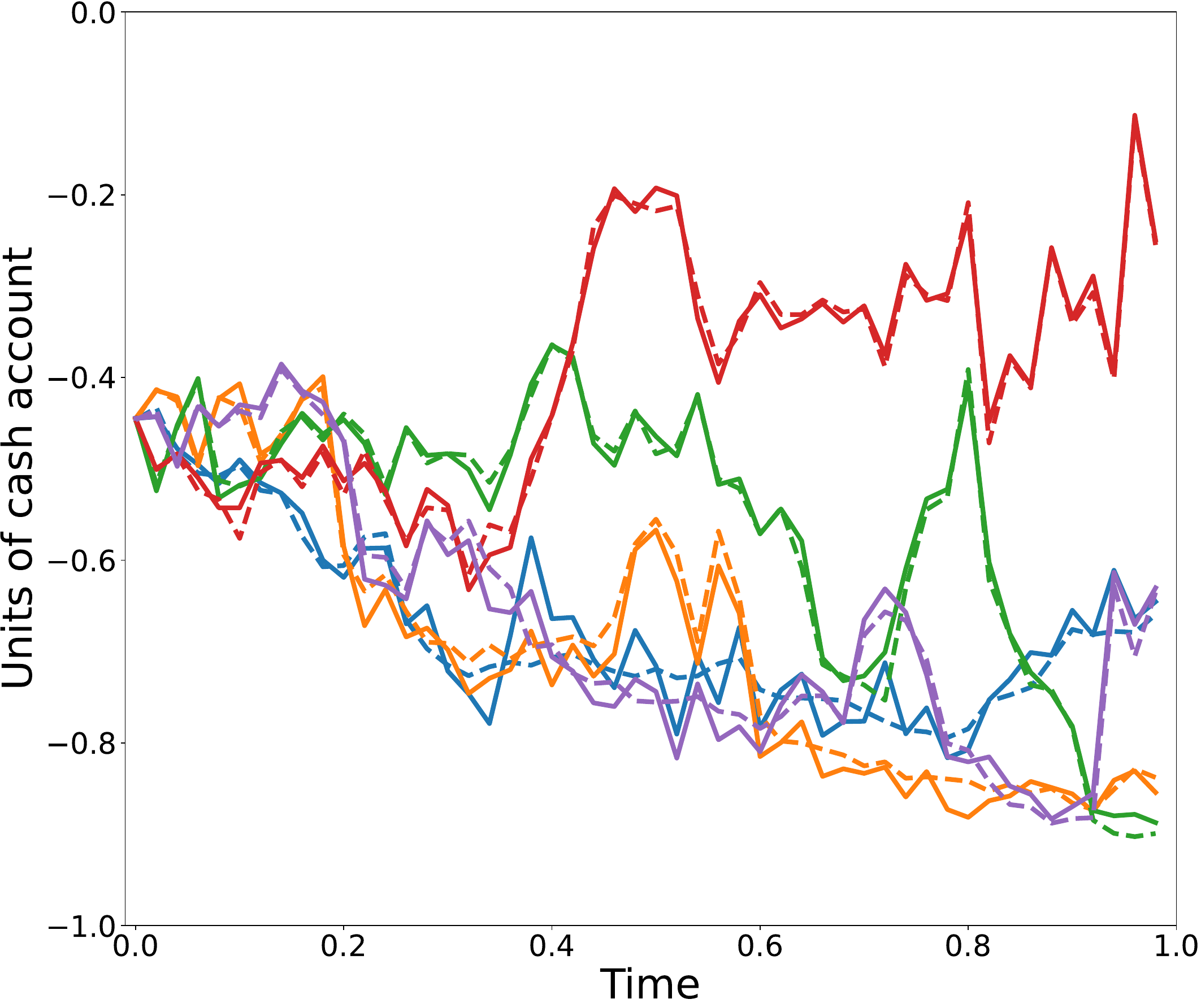}\\
			\multicolumn{2}{c}{\includegraphics[width=0.48\textwidth, height =0.35\textwidth]{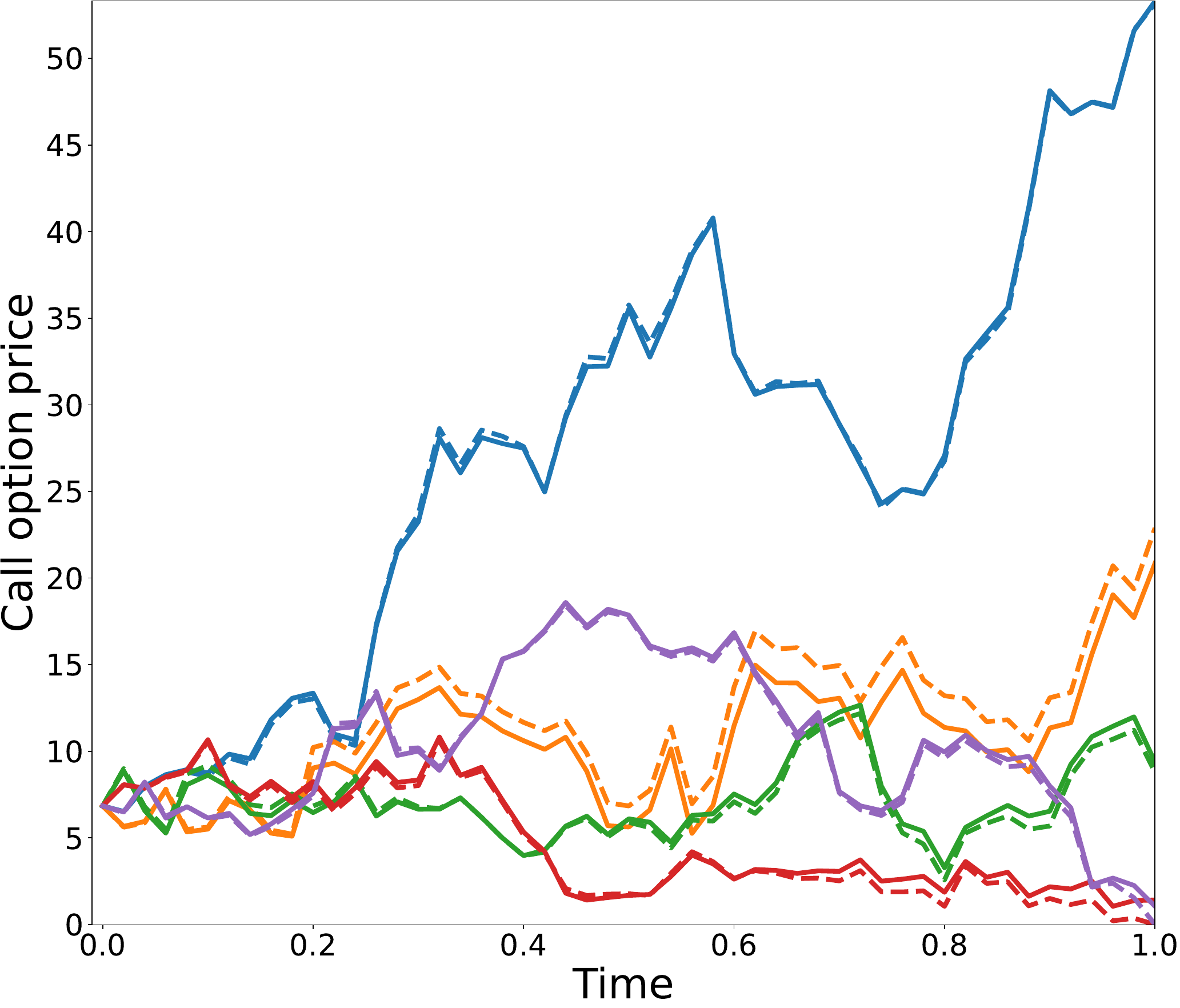}} \\
		\end{tabular}
	}
	\caption{Deep solver solution (solid line) and benchmark solution (dashed line) for the local risk minimization in a $50$ points discretization grid for $5$ random samples  in the interval $[0, 1]$. Upper panel: the shares of risky asset (left) and the units of cash account (right); lower panel: the call option price. \label{LR50}}
\end{figure}

\begin{figure}[tp]
	\resizebox{1\textwidth}{!}{
		\begin{tabular}{@{}>{\centering\arraybackslash}m{0.5\textwidth}@{}>{\centering\arraybackslash}m{0.5\textwidth}@{}}
			\multicolumn{2}{c}{\huge \textbf{Local risk minimization: $\boldsymbol{N = 100}$}}\\
			& \\
			\includegraphics[width=0.48\textwidth, height =0.35\textwidth]{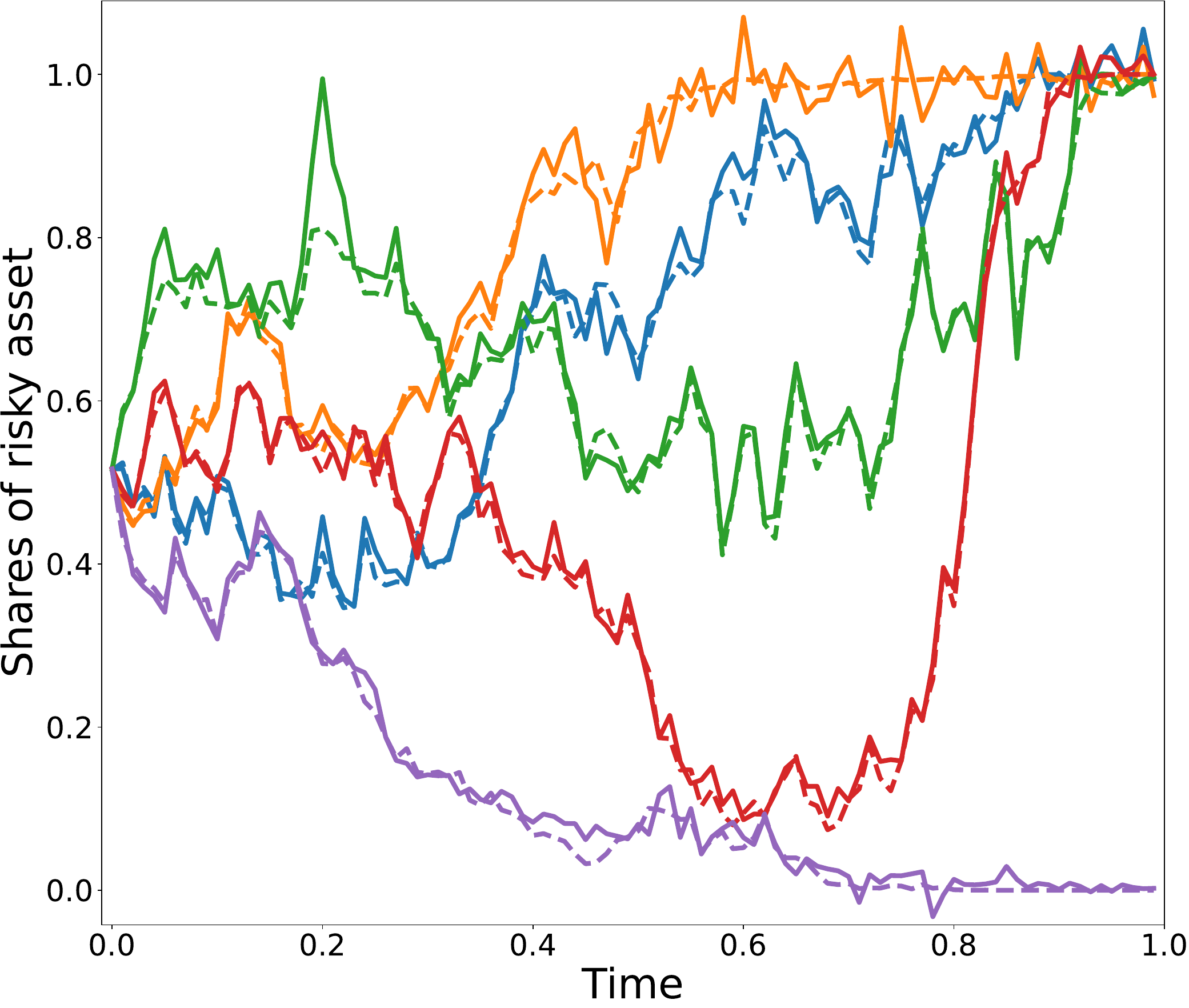} &
			\includegraphics[width=0.48\textwidth, height =0.35\textwidth]{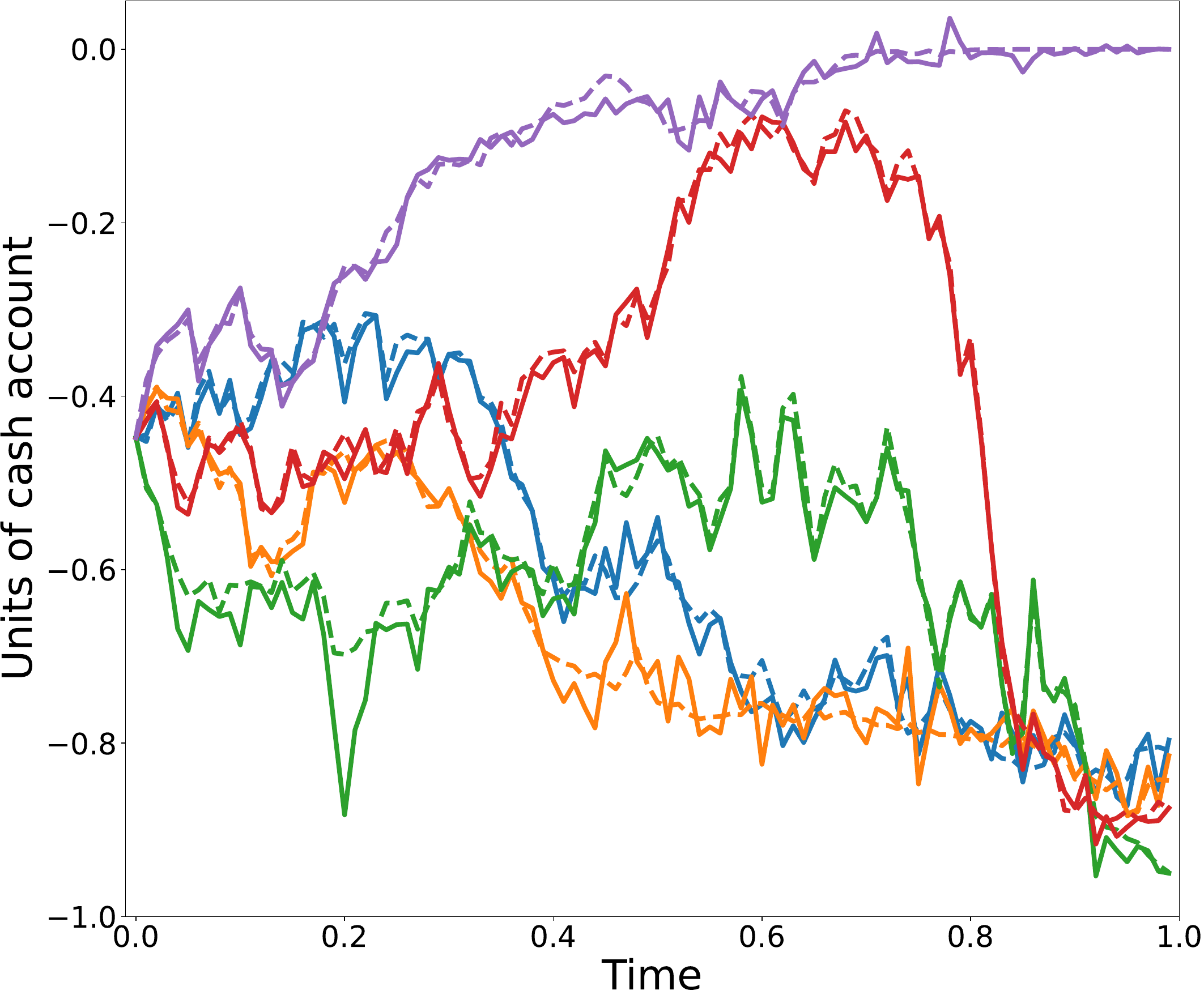}\\
			\multicolumn{2}{c}{\includegraphics[width=0.48\textwidth, height =0.35\textwidth]{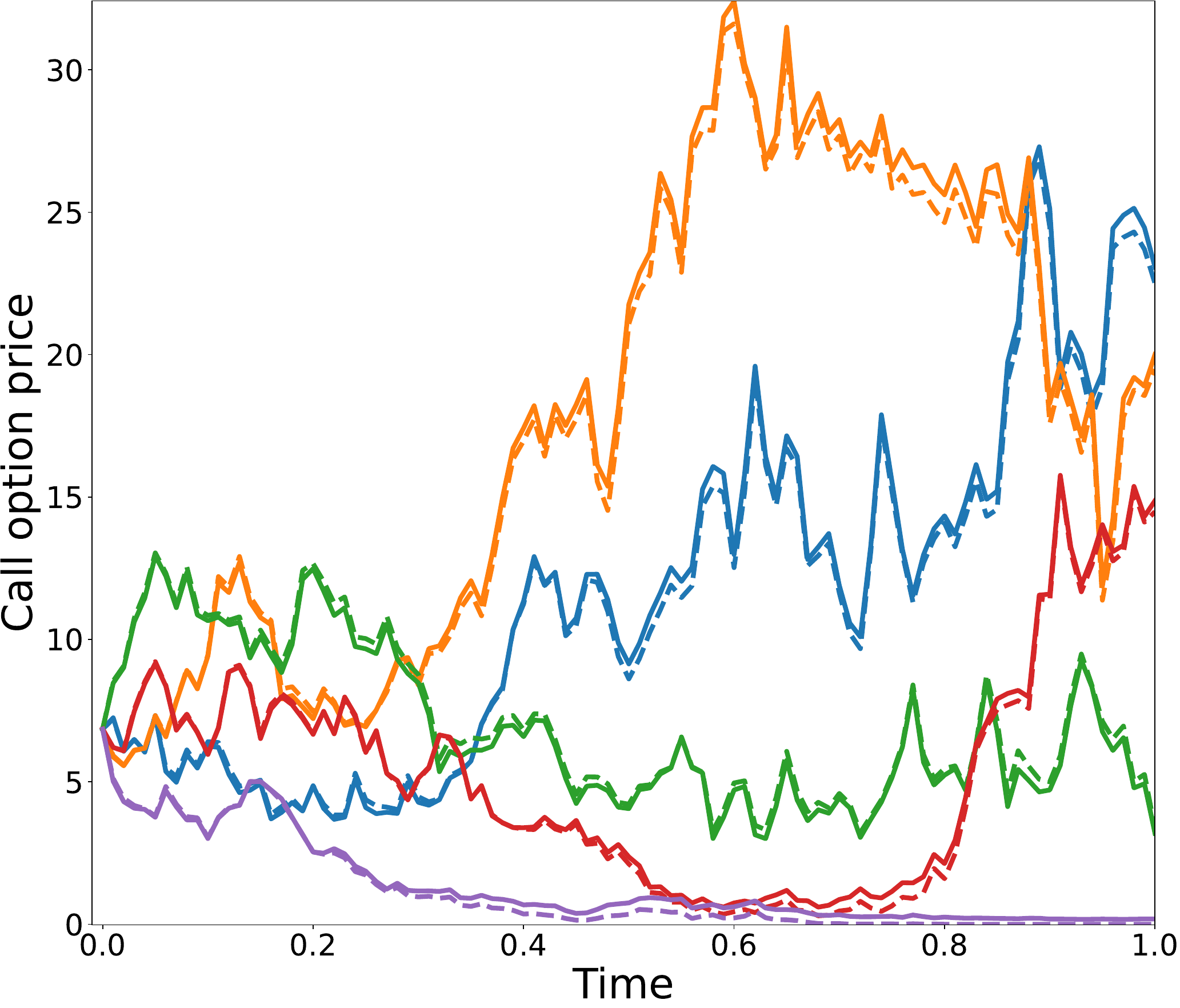}} \\
		\end{tabular}
	}
	\caption{Deep solver solution (solid line) and benchmark solution (dashed line) for the local risk minimization in a $100$ points discretization grid for $5$ random samples  in the interval $[0, 1]$. Upper panel: the shares of risky asset (left) and the units of cash account (right); lower panel: the call option price. \label{LR100}}
\end{figure}

\begin{figure}[tp]
	\resizebox{1\textwidth}{!}{
		\begin{tabular}{@{}>{\centering\arraybackslash}m{0.5\textwidth}@{}>{\centering\arraybackslash}m{0.5\textwidth}@{}}
			\multicolumn{2}{c}{\huge \textbf{Mean Squared Error}}\\
			& \\
			\textbf{\Large Mean-variance hedging} & \textbf{\Large Local risk minimization} \\
			& \\
			\includegraphics[width=0.48\textwidth, height =0.4\textwidth]{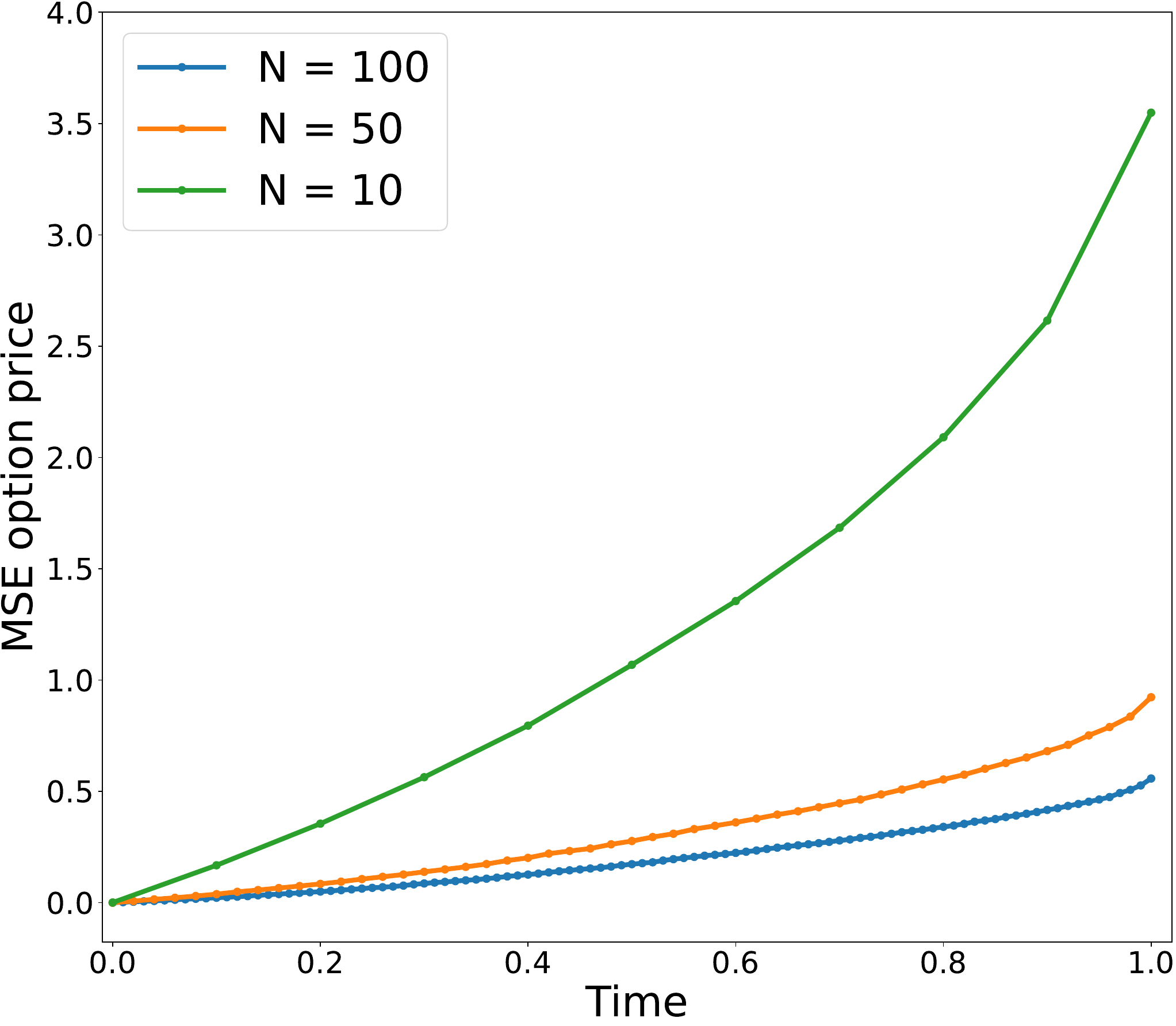}&\includegraphics[width=0.48\textwidth, height =0.4\textwidth]{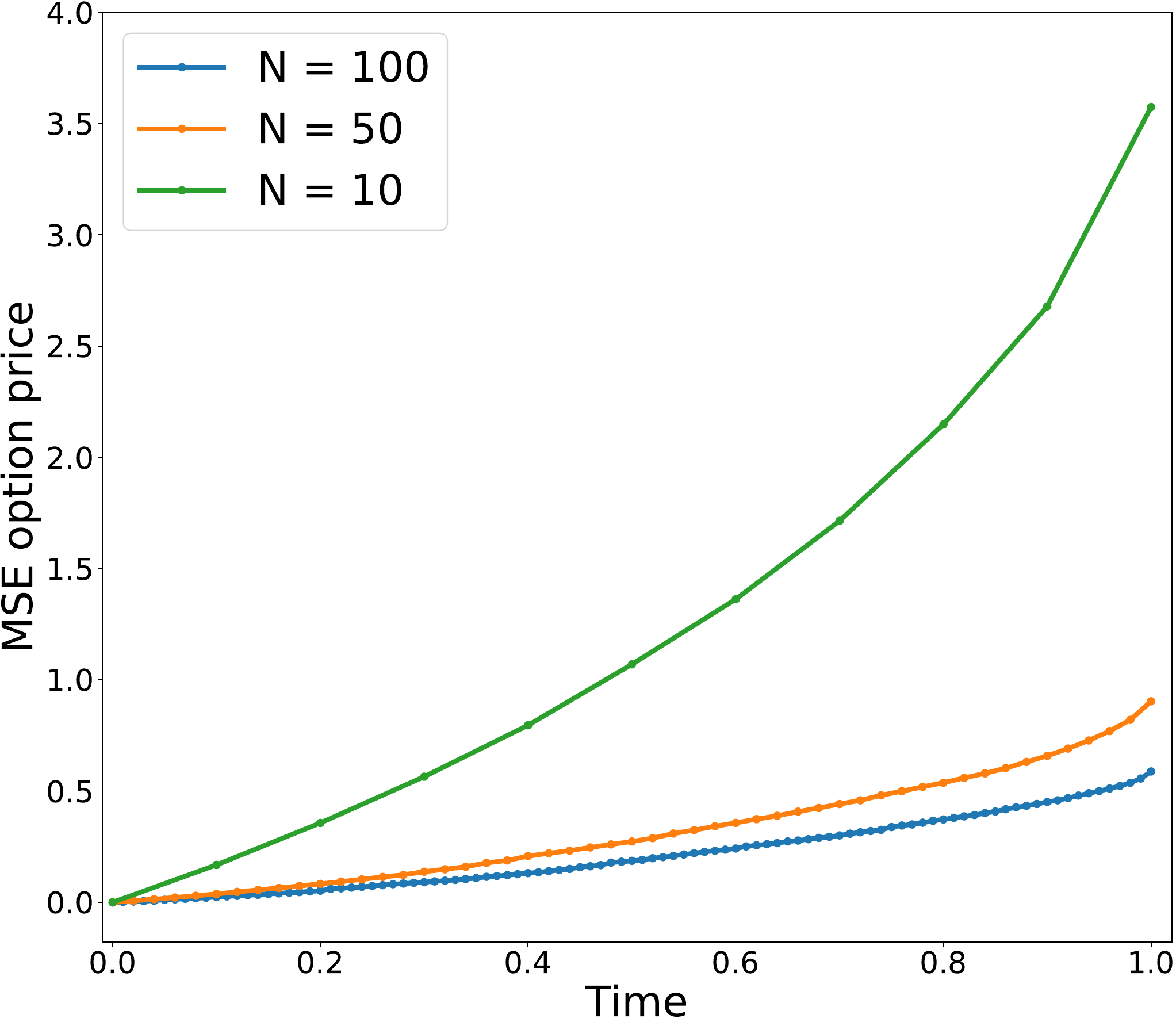}\\
			\includegraphics[width=0.48\textwidth, height =0.4\textwidth]{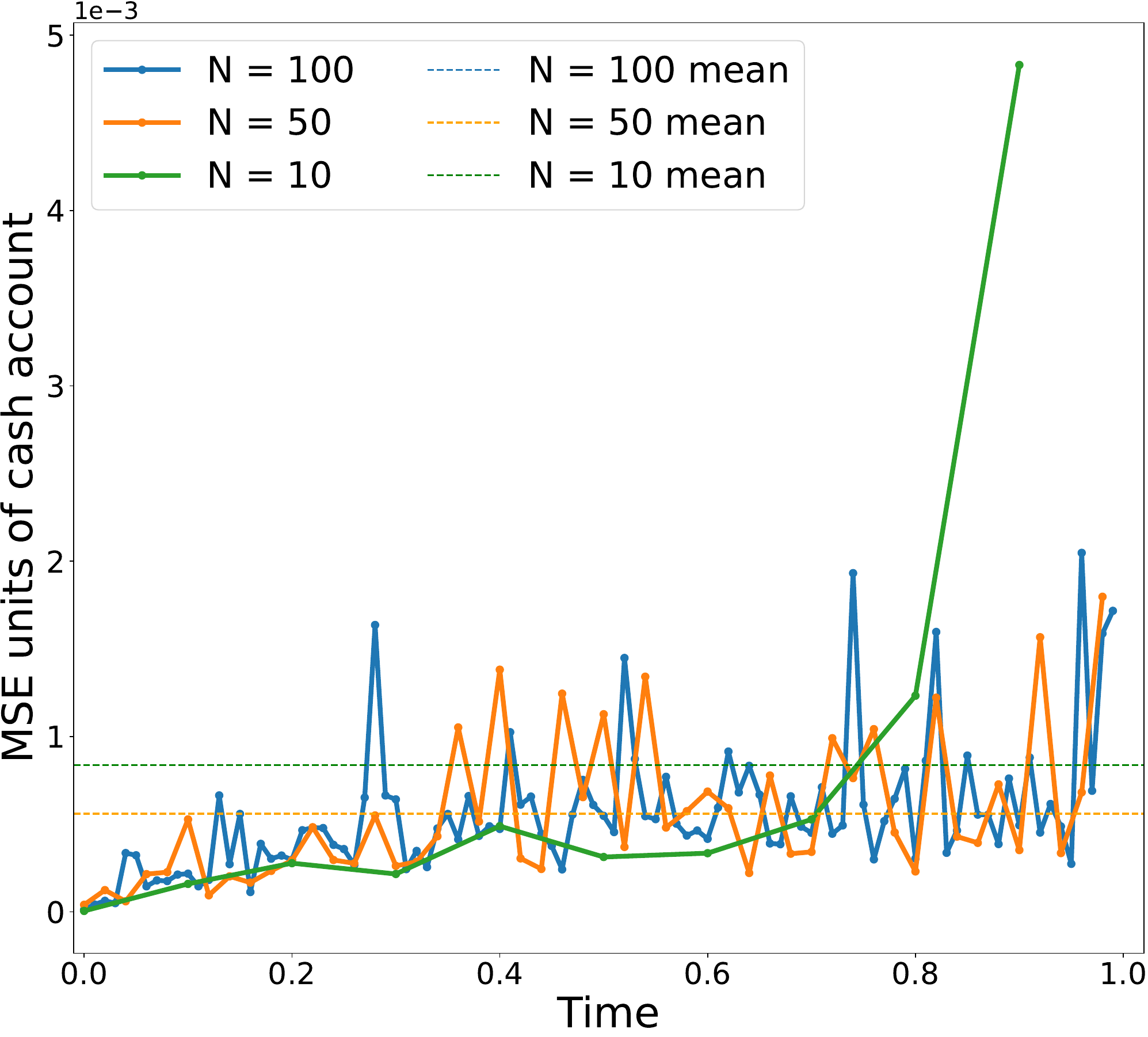}&\includegraphics[width=0.48\textwidth, height =0.4\textwidth]{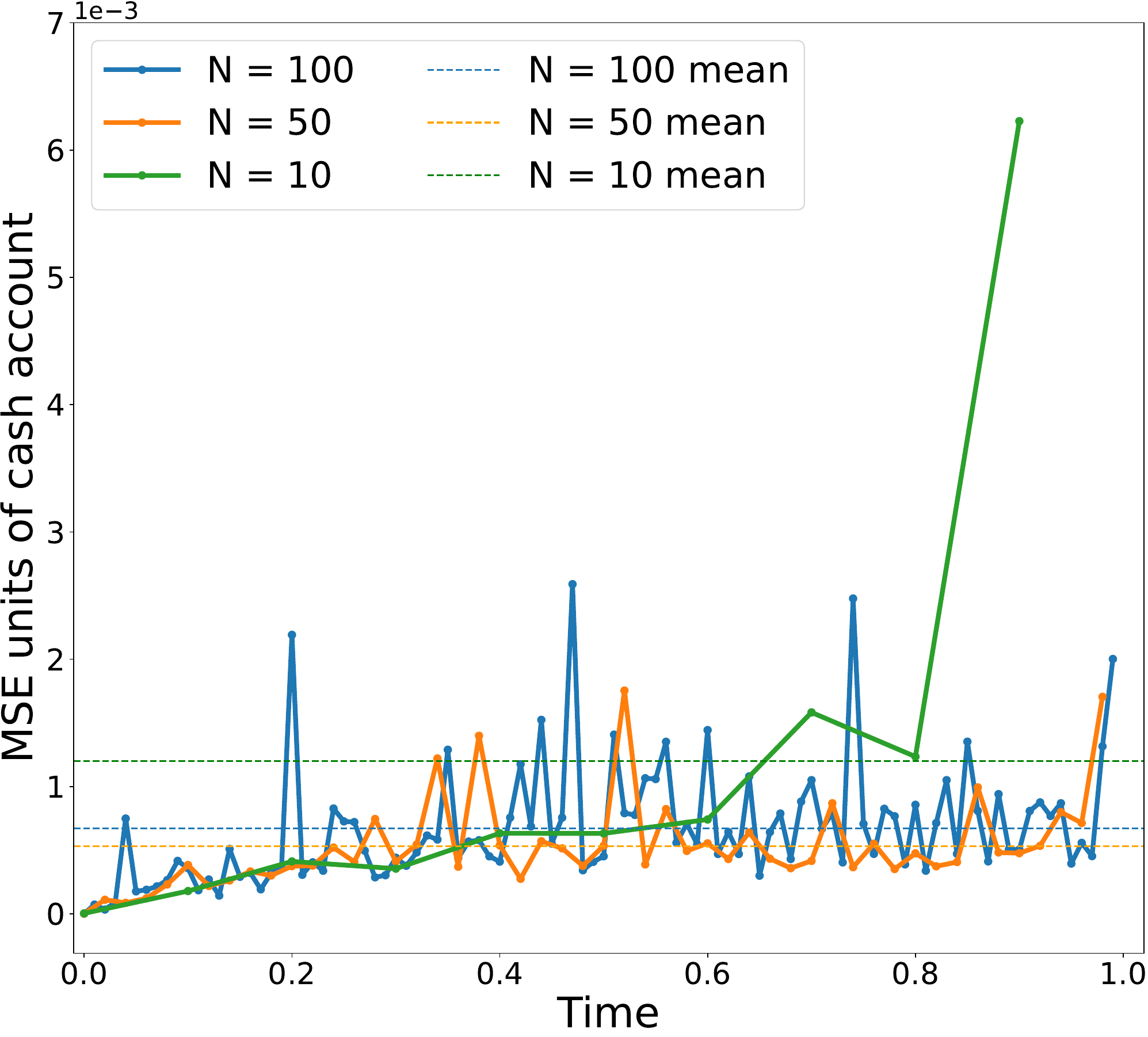}\\
			\includegraphics[width=0.48\textwidth, height =0.4\textwidth]{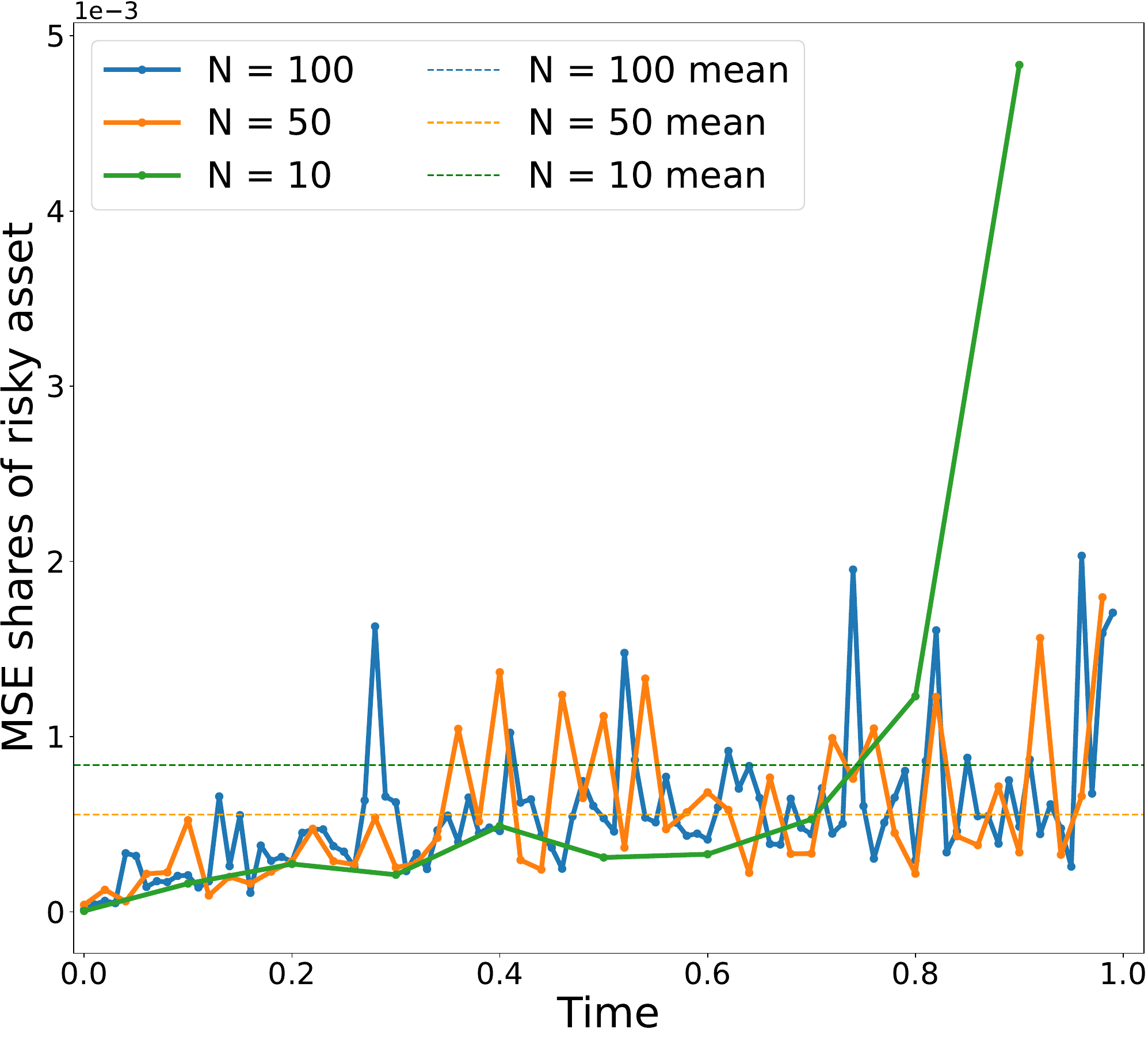}&\includegraphics[width=0.48\textwidth, height =0.4\textwidth]{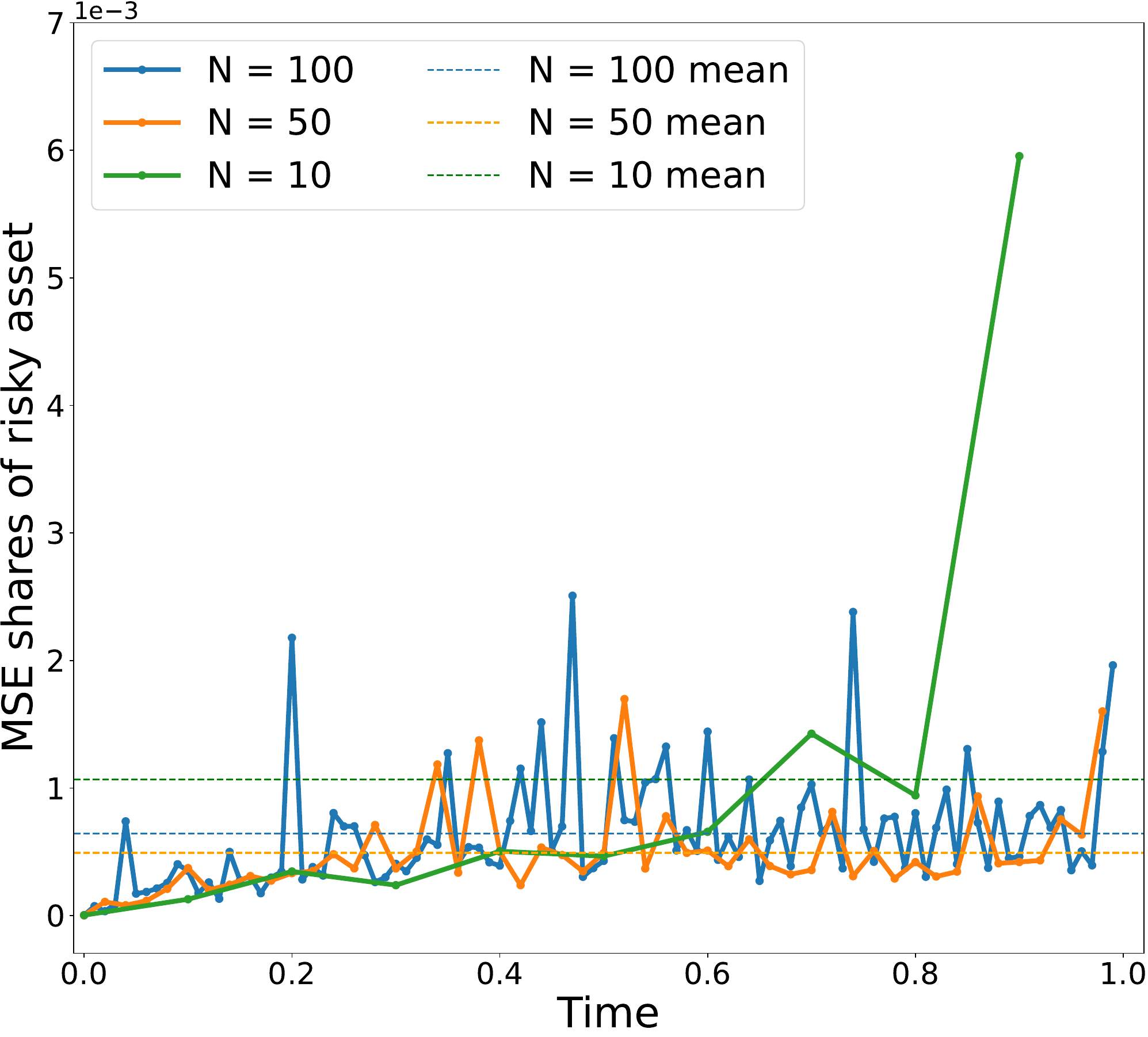}\\
		\end{tabular}
	}
	\caption{Mean Squared Error (MSE) for the mean-variance hedging (left) and local risk minimization (right). From above: MSE for the option price, MSE for the units of cash account and MSE for the shares of risky asset. \label{MSE}}
\end{figure}

\bibliographystyle{plain} 
\bibliography{biblio.bib} 

\appendix

\section{Semi-explicit solutions for mean-variance hedging and local risk minimization for the Heston model in one dimension}
\label{sec:benckmarks}
We derive semi-explicit solutions for the mean-variance hedging and for the local risk minimization by adapting, respectively, the results from \cite{cerny2008} and \cite{heath2001numerical}. These allow us to benchmark the price process and the trading strategy trajectories for the Heston model in dimension $1$, see Section \ref{sec:numerics}.

\subsection{The mean-variance hedging of {\v{C}}erný  and Kallsen \cite{cerny2008}}
\label{meanvariance1dim}
In the mean-variance hedging, we solve recursively two BSDEs, respectively equation \eqref{eq:StochasticRiccati} and equation \eqref{hedgingproblemX}, and we find (an approximation of) the two processes $L$ and $\tilde X^{\mathrm{mv}}$.

The process $L$ corresponds to what {\v{C}}erný  and Kallsen define as the \emph{opportunity process}, which characterizes the so-called \emph{opportunity-neutral measure} $\p^*$. This is a non-martingale equivalent measure such that the variance optimal martingale measure $\Q_{{\text{mv}}}$ introduced in Section \ref{sec:MeanVarTheory} can be computed as the minimal martingale measure under $\p^*$, see \cite{cerny2007} or \cite{cerny2008} for details.

Following the approach of \cite{cerny2008}, the opportunity process is of the form
\begin{equation}
	\label{definitionLCK}
	L^{\text{CK}}_t = \exp\left\{\chi_0(t)+\chi_1(t)Y^2_t\right\},
\end{equation}
$Y^2$ being the variance process for the one-dimensional Heston model \eqref{HestonModel}, while the functions $\chi_0$ and $\chi_1$ are solutions of a system of Riccati ODEs, see \cite[Proposition 3.2]{cerny2008}. 
In particular, by adapting \cite[Lemma 6.1]{cerny2008}, 
 the functions $\chi_0$ and $\chi_1$ at time $t\in[0, T]$ are
\begin{align}
&\label{chi0} \chi_{0}(t) = \mathfrak{F}\left(-\frac{\mathfrak{B}}{2\mathfrak{C}}(T-t) -\frac{1}{\mathfrak{C}}\log\left(\frac{(\mathfrak{B+D})e^{-\mathfrak{D}(T-t)/2}-(\mathfrak{B-D})e^{\mathfrak{D}(T-t)/2}}{2\mathfrak{D}} \right)\right),\\
&\label{chi1} \chi_{1}(t) = -\frac{\mathfrak{B}}{2\mathfrak{C}} + \frac{\mathfrak{D}}{2\mathfrak{C}}\frac{(\mathfrak{B+D})e^{-\mathfrak{D}(T-t)/2}+(\mathfrak{B-D})e^{\mathfrak{D}(T-t)/2}}{(\mathfrak{B+D})e^{-\mathfrak{D}(T-t)/2}-(\mathfrak{B-D})e^{\mathfrak{D}(T-t)/2}},
\end{align}
where
\begin{equation*}
	\begin{array}{l ll}
		\mathfrak{A} := -\mu^2,& \mathfrak{B} := -\kappa - 2\rho\sigma\mu,&
		\mathfrak{C} := \frac{1}{2}\sigma^2(1-2\rho^2), \\\mathfrak{D} := \sqrt{\mathfrak{B}^2-4\mathfrak{A}\mathfrak{C}},&
		\mathfrak{F} := \kappa \theta,&
	\end{array}
\end{equation*}
with $\mathfrak{C}, \mathfrak{D}\ne0$. We refer to \cite[Lemma 6.1]{cerny2008} for the case $\mathfrak{C} = 0$ and $\mathfrak{D}=0$.

As we mentioned above, the process $L^{\text{CK}}$ characterizes  the opportunity-neutral measure $\p^{*}$. In particular, the process
\begin{equation*}
	Z_t := \frac{L^{\text{CK}}_t}{L^{\text{CK}}_0}\exp\left\{-\int_0^t \left(\mu + \rho\sigma \chi_{1}(s)\right)^2Y_s^2\dd s \right\} = \mathcal{E}\left(\int_0^t\sigma \chi_{1}(s)Y_s\left(\rho \dd W_s+ \sqrt{1-\rho^2}\dd B_s\right)\right)
\end{equation*}
is a bounded positive martingale, and  by virtue of the Girsanov theorem
\begin{equation}
	\label{change1}
	\begin{aligned}
		&W^*_t := W_t - \int_0^t \rho \sigma\chi_{1}(s)Y_s \dd s,\\
		&B^*_t := B_t - \int_0^t \sqrt{1-\rho^2} \sigma\chi_{1}(s)Y_s \dd s,
	\end{aligned}
\end{equation}
are Brownian motions under $\p^*$, with $\dd \p^*/\dd \p = Z_T$. The variance optimal martingale measure coincides then with the minimal measure relative to $\p^*$:
\begin{equation*}
	\frac{\dd\mathbb{Q}_{\text{mv}}}{\dd\p^*} =  \mathcal{E}\left( -\int_0^T \left(\mu + \rho\sigma \chi_{1}(s)\right)Y_s \dd W^*_s\right).
\end{equation*}
By the Girsanov theorem,
\begin{equation}
	\label{change2}
	\begin{aligned}
		&W^{\text{mv}}_t := W^*_t + \int_0^T \left(\mu + \rho\sigma \chi_{1}(s)\right)Y_s \dd  s, \;\text{and}\\
		&B^{\text{mv}}_t := B^*_t 
	\end{aligned}
\end{equation}
are uncorrelated Brownian motions under $\mathbb{Q}_{{\text{mv}}}$. By combining equation \eqref{change1} and \eqref{change2} with the Heston dynamics \eqref{HestonModel}, the dynamics of $\tilde S$ and $Y^2$ under the variance optimal martingale measure become
\begin{equation}
	\label{HestonQ}
	\begin{cases}
		\dd \tilde S_t = \tilde S_t Y_t \dd W^{\text{mv}}_t, \\
		\dd Y_t^2 =  \left(\kappa\theta - \kappa_tY_t^2\right)\dd t + \sigma Y_t \left(\rho \dd W_t^{\text{mv}}+ \sqrt{1-\rho^2}\dd B_t^{\text{mv}}\right),\\
	\end{cases}
\end{equation}
with ${\kappa}_t := \kappa + \rho\sigma\mu -\chi_1(t)\sigma^2(1-\rho^2)$.

Let us now focus on the solution of the second BSDE, $\tilde X^{\text{mv}}$. Following \cite{cerny2008}, if the value of the contingent claim $\tilde H$ is given by $g(Y^2_T, \tilde S_T)$, with $g$ bounded and continuous function, we can consider a function $f\in \mathcal{C}^{1,2,2}$ such that 
\begin{equation}
	\label{definitionf}
	\tilde X^{\text{CK}}_t =  \mathbb{E}^{\mathbb{Q}_{{\text{mv}}}}\left[\left. \tilde H\right| \mathcal{F}_{t}\right] = f(T-t, Y^2_t, \tilde S_t),
\end{equation}
where $\tilde X^{\text{CK}}$ is the mean-variance hedging contingent claim price in the {\v{C}}erný  and Kallsen approach.
From \cite[Proposition 4.1]{cerny2008}, the function $f$ is the unique classical solution of the PDE
\begin{equation}
	\label{PDEf}
	\begin{cases}
	-f_1 + \frac{1}{2} y \left(\sigma^2 f_{22}+2\rho\sigma s f_{23} + s^2f_{33}\right)+ (\kappa \theta -(\kappa +\rho\sigma\mu -\chi_1(t)\sigma^2(1-\rho^2))y) f_2 =0,\\
	f(0, y, s) = g(y,s),
\end{cases}
\end{equation}
where $f_i := \frac{\partial f}{\partial x_i}$ and $f_{ij} := \frac{\partial^2 f}{\partial x_i \partial x_j}$, for $i, j\in \{1,2,3\}$. We point out that the PDE \eqref{PDEf} has time-dependent coefficients, hence we need to adapt the classical methodology for solving the Heston PDE to that. We consider the approach of \cite{in2010adi}, where a finite difference method is applied in the space dimension and a splitting scheme of the Alternative Direction Implicit (ADI) type is applied in the time dimension. We refer to Appendix \ref{sec:PDEsolver} for a brief presentation of the methodology, or to \cite{in2010adi} directly for a more detailed description. We denote by $\widehat{f}^{n}$ the approximated value of $f$ in $(t_n, \overline{Y}^2_n, \overline{S}_n )$ and by $\widehat{f}^n_{i}$ its derivatives, for $n=1, \dots, N$ and $i\in \{1,2,3\}$. Let $\widehat X^{\text{CK}}_n$ be the approximated $\tilde X_t^{\text{CK}}$ in $t=t_n$.

	\begin{oss}
		In the discussion above, we consider a  bounded and continuous payoff function $g$. This covers the case of European style put options, whose payoff function is of the form $g(x) = \max(K-x, 0)$ in the one-dimensional setting. However, it does not cover the case of European style call options with payoff function $g(x) = \max(x-K, 0)$, which is the case treated in the numerical experiments in Section \ref{sec:numerics}. However, due to the put-call parity, we can safely use also the call payoff function.
	\end{oss}

\begin{oss}
	\label{oss:interpolate}
By solving the PDE as described in Appendix \ref{sec:PDEsolver}, for each time step $n$  we find the values of the function $f$ on a grid for the bounded domain $[0, \mathbb{S}]\times [0, \mathbb{Y}]$, for a certain $ \mathbb{S}$ and a certain $ \mathbb{Y}$. Let $M_s\ge 1$ and $M_y\ge 1$ be the number of points, respectively, in the $s$- and $y$-direction, which we denote with $s_k$, $k=1, \dots, M_s$, and $y_{\ell}$, $\ell=1, \dots, M_y$. We then find $\widehat{f}^{(n, \ell, k)}\approx f(t_n, y_{\ell}, s_{k})$ as the approximation of $f$ in $(t_n, y_{\ell}, s_{k})$, for $n=1, \dots, N$.

If we want to compare the $\widehat X^{\text{mv}}$ found with the deep BSDE solver and the corresponding $\widehat X^{\text{CK}}$ from the {\v{C}}erný  and Kallsen approach for a specific realization of the forward process $(\tilde S, Y^2)$, we then need first to interpolate $\widehat{f}^{(n, \ell, k)}$ to find the right values for comparison. Namely, given a realization $(\overline S_n, \overline Y_n^2)$, for $n=1,\dots, N$,  for each time step $n$ we need to find where the simulated $\overline S_n$ and $\overline Y_n^2$ are located in the mesh, this means to find $k\in\{1, \dots, M_s-1\}$ and $\ell\in\{1, \dots, M_y-1\}$ such that $\overline S_n \in ({s_{k}}, {s_{k+1}}]$ and $\overline Y_n^2 \in ({y_{\ell}}, {y_{\ell+1}}]$. We then denote with $\widehat{f}^{n}$ the value that we obtain with a bi-linear interpolation from $\widehat{f}^{(n, \ell, k)}$, $\widehat{f}^{(n, \ell+1, k)}$, $\widehat{f}^{(n, \ell, k+1)}$ and $\widehat{f}^{(n, \ell+1, k+1)}$. Similarly if we want to compute the derivatives $\widehat{f}^n_{i}$. 
\end{oss}


The approximated optimal trading strategy and the corresponding value process can then be recursively computed starting from \cite[Equations (3.3) and (4.1)]{cerny2008} as follows:
\begin{equation}
\begin{cases}
\widehat{V}_0^{\text{CK}} = y,\\
\widehat{\xi}^{\text{CK}}_n= \widehat{f}^n_3 +\frac{\rho\sigma \widehat{f}^n_2}{\overline{{S}}_n}  +\frac{\mu + \rho\sigma \chi_{1}(t_n)}{\overline{{S}}_n}\left(\widehat{f}^n-  \widehat{V}_n^{\text{CK}} \right), \\
\widehat{\psi}^{\text{CK}}_n=  \widehat{V}_{n}^{\text{CK}} - \widehat{\xi}^{\text{CK}}_n\overline{{S}}_n,\\
\widehat{V}_{n+1}^{\text{CK}} = \widehat{V}_n^{\text{CK}} + \widehat{\xi}^{\text{CK}}_n\left(\overline{{S}}_{n+1} - \overline{{S}}_{n}\right),
\end{cases}\hfill{\text{for $n=0,\ldots, N-1$}}.
\end{equation}

\subsection{The local risk minimization of Heath, Platen and Schweizer \cite{heath2001numerical}}
\label{localrisk1dim}
In the local risk minimization we solve the BSDE \eqref{eq:localriskX} and we find (an approximation of) the process $\tilde X^{\text{lr}}$. If the dimension is $d=1$, we can alternatively follow the  approach of \cite{heath2001numerical}.

As in the mean-variance setting, we assume that the value of the contingent claim $\tilde H$ is given by a function $g(Y^2_T, \tilde S_T)$. Then, there exists a function $f\in \mathcal{C}^{1,2,2}$ such that 
\begin{equation}
	\label{definitionfLR}
	\tilde X^{\text{HPS}}_t =  \mathbb{E}^{\mathbb{Q}_{\text{lr}}}\left[\left. \tilde H\right| \mathcal{F}_{t}\right] = f(T-t, Y^2_t, \tilde S_t),
\end{equation}
where $\tilde X^{\text{HPS}}$ is the local risk minimizing contingent claim price in the {Heath, Platen and Schweizer approach.
By adapting the results of \cite[Section 2]{heath2001numerical}, the function $f$ is the unique classical solution of the PDE
\begin{equation}
	\label{PDEfLR}
	\begin{cases}
		 -f_1 + \frac{1}{2} y \left(\sigma^2 f_{22}+2\rho\sigma s f_{23} + s^2f_{33}\right)+ (\kappa \theta -(\kappa +\rho\sigma\mu )y) f_2 =0,\\
		 f(0, y, s) = g(y,s),
	\end{cases}
\end{equation}
where $f_i := \frac{\partial f}{\partial x_i}$ and $f_{ij} := \frac{\partial^2 f}{\partial x_i \partial x_j}$, for $i, j\in \{1,2,3\}$.

Once solved \eqref{PDEfLR} with the PDE solver described in Appendix \ref{sec:PDEsolver}, we denote by $\widehat{f}^{n}$ the approximated value of $f$ in $(t_n, \overline{Y}^2_n, \overline{S}_n )$ and by $\widehat{f}^n_{i}$ its derivatives, for $n=1, \dots, N$, and $i\in \{1,2,3\}$. Of course, the same arguments of Remark \ref{oss:interpolate} apply here. The approximated optimal trading strategy can then be  computed from \cite[Equation (2.3)]{heath2001numerical}:
\begin{equation}
	\begin{cases}
		\widehat \xi^{\text{HPS}}_n =  \widehat{f}^n_3 +\frac{\rho\sigma \widehat{f}^n_2}{\overline{{S}}_n}, \\
		\widehat \psi^{\text{HPS}}_n = \widehat{f}^n - \overline{S}_n \widehat \xi^{\text{HPS}}_n,
	\end{cases}\hfill{\text{for $n=0,\ldots, N-1$}}.
\end{equation}

\section{Proof of Proposition \ref{prop:Lim200433}}
\label{Proofprop}
To ease notation, we omit the superscript $^{\mathrm{mv}}$ in the computations below. Since we work with discounted quantities, for the sake of clarity we adapt the steps of \cite[Proposition 3.3]{lim2004}. The dynamics of the discounted wealth satisfy
\begin{align*}
\dd \tilde{V}_t=\xi^{\top}_t \operatorname{diag} (\tilde{S}_t)\left(\mu_{t}-r_t\mathbbm{1}\right) \dd t+\xi^{\top}_t \operatorname{diag}(\tilde{S}_t) \sigma_{t} \dd W_{t}.
\end{align*}
We also have
\begin{align*}
\dd\tilde{X}_t=\frac{1}{S_{t}^{0}}\left(\left(\phi_{t}^{\top} \eta_{1, t}-\frac{\Lambda_{2, t}^{\top} \eta_{2, t}}{L_{t}}\right) \dd t+\eta_{1, t}^{\top} \dd W_{t}+\eta_{2, t}^{\top} \dd B_{t}\right),
\end{align*}
hence
\begin{align*}
\begin{aligned} \dd(\tilde{X}-\tilde{V})_{t}&=\left(\frac{1}{S^0_{t}}\left(\phi_{t}^{\top} \eta_{1, t}-\frac{\Lambda_{2,t}^{\top} \eta_{2, t}}{L_{t}}\right)-\xi_t^{\top} \operatorname{diag}(\tilde{S}_t)\left(\mu_{t}-r_ t\mathbbm{1}\right)\right) \dd t\\ &+\left(\frac{1}{S_{t}^{0}} \eta_{1, t}^{\top}-\xi_t^{\top}  \operatorname {diag}(\tilde{S}_t) \sigma_{t}\right) \dd W_{t}+\frac{1}{S_{t}^{0}} \eta_{2, t}^{\top} \dd B_{t}. \end{aligned}
\end{align*}
Now,
\begin{align*}
\dd\left(\tilde{X}-\tilde{V}\right)^{2}_t=2\left(\tilde{X}_{t}-\tilde{V}_{t}\right) \dd(\tilde{X}-\tilde{V})_{t}+ \dd\langle\tilde{X}-\tilde{V}\rangle_{t}
\end{align*}
where
\begin{align*}
\begin{aligned} \dd\langle\tilde{X}-\tilde{V}\rangle_t=\left(\frac{1}{\left(S_{t}^{0}\right)^{2}} \eta_{1, t}^{\top} \eta_{1, t}\right.&+\left(\xi_t^{\top} \operatorname{diag}(\tilde{S}_t) \sigma_t\right)\left(\xi_t^{\top} \operatorname{diag}(\tilde{S}_t){\sigma_t}\right)^{\top} \\ &\left.-2 \frac{1}{S^0_{t}} \eta_{1, t}^{\top}\left(\xi_t^{\top} \operatorname{diag}(\tilde{S}_t) \sigma_{t}\right)+\frac{1}{\left(S_{t}^{0}\right)^{2}} \eta_{2, t}^{\top} \eta_{2, t}\right) \dd t. \end{aligned}
\end{align*}
Regrouping terms we get
\begin{align*}
\begin{aligned}\dd\left(\tilde{X}-\tilde{V}\right)^{2}_t&=\left[ 2\left(\tilde{X}_{t}-\tilde{V}_{t}\right)\frac{1}{S_{t}^{0}} \left[\phi_{t}^{\top} \eta_{1, t}-\frac{\Lambda_{2, t}^{\top} \eta_{2, t}}{L_{t}}\right]+\frac{1}{\left(S_{t}^{0}\right)^{2}} \eta_{1, t}^{\top} \eta_{1, t}+\frac{1}{\left(S_{t}^{0}\right)^{2}} \eta_{2, t}^{\top} \eta_{2, t}\right. \\ &+\xi_t^{\top} \operatorname{diag}(\tilde{S}_t) \sigma_{t} \sigma_{t}^{\top} \operatorname{diag}(\tilde{S}_t) \xi_t \\ &\left.-2 \xi_{t}^{\top} \operatorname{diag}(\tilde{S}_t)\left[\left(\mu_{t}-r_t \mathbbm{1}\right)\left(\tilde{X}_{t}-\tilde{V}_{t}\right)+\sigma_{t} \frac{\eta_{1, t}}{S_{t}^{0}}\right]\right] \dd t \\ &+2\left(\tilde{X}_{t}-\tilde{V}_{t}\right)\left(\frac{1}{S_{t}^{0}} \eta_{1, t}^{\top}-\xi_t^{\top} \operatorname{diag}(\tilde{S}_t) \sigma_{t}\right) \dd W_{t}+2\left(\tilde{X}_{t}-\tilde{V}_{t}\right) \frac{1}{S_{t}^{0}} \eta_{2, t}^{\top} \dd B_{t}. \end{aligned}
\end{align*}
We now apply again the It\^o formula:
\begin{align*}
\dd L_{t}\left(\tilde{X}_{t}-\tilde{V}_{t}\right)^{2}=\left(\tilde{X}_t-\tilde{V}_{t}\right)^{2} \dd L_{t}+L_{t} \dd(\tilde{X}-\tilde{V})_{t}^{2}+\dd\left\langle L,\left(\tilde{X}-\tilde{V}\right)\right\rangle_{t},
\end{align*}
where the covariation term is given by
\begin{align*}
\dd\left\langle L,\left(\tilde{X}-\tilde{V}\right)\right\rangle_{t}=\left(2\left(\tilde{X}_{t}-\tilde{V}_{t}\right)\left(\frac{1}{S_{t}^{0}} \eta_{1, t}^{\top}-\xi_t^{\top} \operatorname{diag}(\tilde{S}_t) \sigma_{t}\right) \Lambda_{1, t}+2\left(\tilde{X}_{t}-\tilde{V}_{t}\right) \frac{1}{S_{t}^{0}} \eta_{2, t}^{\top} \Lambda_{2, t} \right)\dd t,
\end{align*}
hence
\begin{align*}
\begin{aligned}
\dd L_{t}\left(\tilde{X}_{t}-\tilde{V}_{t}\right)^{2}&=\left(\tilde{X}_{t}-\tilde{V}_{t}\right)^{2}\left[\left(\left|\phi_{t}\right|^{2} L_{t}+2 \phi_{t}^{\top} \Lambda_{1, t}+\frac{\Lambda_{1, t}^{\top} \Lambda_{1, t}}{L_{t}}\right) \dd t+\Lambda_{1, t}^{\top} \dd W_{t}+\Lambda_{2, t}^{\top} \dd B_{t}\right]\\
&+L_t\left\{\left[ 2\left(\tilde{X}_{t}-\tilde{V}_{t}\right)\frac{1}{S_{t}^{0}}\left[ \phi_{t}^{\top} \eta_{1, t}-\frac{\Lambda_{2, t}^{\top} \eta_{2, t}}{L_{t}}\right]+\frac{1}{\left(S_{t}^{0}\right)^{2}} \eta_{1, t}^{\top} \eta_{1, t}+\frac{1}{\left(S_{t}^{0}\right)^{2}} \eta_{2, t}^{\top} \eta_{2, t}\right.\right. \\ &+\xi_t^{\top} \operatorname{diag}(\tilde{S}_t) \sigma_{t} \sigma_{t}^{\top} \operatorname{diag}(\tilde{S}_t) \xi_t \\ &\left.-2 \xi_{t}^{\top} \operatorname{diag}(\tilde{S}_t)\left[\left(\mu_{t}-r_t \mathbbm{1}\right)\left(\tilde{X}_{t}-\tilde{V}_{t}\right)+\sigma_{t} \frac{\eta_{1, t}}{S_{t}^{0}}\right] \right]\dd t \\ &+\left.2\left(\tilde{X}_{t}-\tilde{V}_{t}\right)\left(\frac{1}{S_{t}^{0}} \eta_{1, t}^{\top}-\xi_t^{\top} \operatorname{diag}(\tilde{S}_t) \sigma_{t}\right) \dd W_{t}+2\left(\tilde{X}_{t}-\tilde{V}_{t}\right) \frac{1}{S_{t}^{0}} \eta_{2, t}^{\top} \dd B_{t} \right\}\\
&+2\left(\tilde{X}_{t}-\tilde{V}_{t}\right)\left(\frac{1}{S_{t}^{0}} \eta_{1, t}^{\top}-\xi_t^{\top} \operatorname{diag}(\tilde{S}_t) \sigma_{t}\right) \Lambda_{1, t}+2\left(\tilde{X}_{t}-\tilde{V}_{t}\right) \frac{1}{S_{t}^{0}} \eta_{2, t}^{\top} \Lambda_{2, t} \dd t.
\end{aligned}
\end{align*}
Some lengthy computations show that the drift term can be expressed in the following form:
\begin{align*}
\dd L_{t}\left(\tilde{X}_{t}-\tilde{V}_{t}\right)^{2}&=\left\{L_{t}\left[\operatorname{diag}(\tilde{S}_{t}) \xi_{t}-\left(\sigma_{t} \sigma_{t}^{\top}\right)^{-1}\left(\left[\left(\mu_{t}-r_t\mathbbm{1}\right)^{\top}+\frac{\sigma_{t} \Lambda_{1, t}}{L_{t}}\right]\left(\tilde{X}_{t}-\tilde{V}_{t}\right)+\sigma_{t} \frac{\eta_{1, t}}{S_{t}^{0}}\right)\right]^{\top}\right.\\
&\times\sigma_t\sigma_t^\top\left[\operatorname{diag}(\tilde{S}_{t}) \xi_{t}-\left(\sigma_{t} \sigma_{t}^{\top}\right)^{-1}\left(\left[\left(\mu_{t}-r_t\mathbbm{1}\right)^{\top}+\frac{\sigma_{t} \Lambda_{1, t}}{L_{t}}\right]\left(\tilde{X}_{t}-\tilde{V}_{t}\right)+\sigma_{t} \frac{\eta_{1, t}}{S_{t}^{0}}\right)\right]\\
&\left.+\frac{L_t}{\left(S_{t}^{0}\right)^{2}} \eta_{2, t}^{\top} \eta_{2, t}\right\}\dd t\\
&+\left(2L_{t}\left(\tilde{X}_{t}-\tilde{V}_{t}\right)\left(\frac{1}{S_{t}^{0}} \eta_{1, t}^{\top}-\xi_t^{\top} \operatorname{diag}(\tilde{S}_t) \sigma_{t}\right) +\left(\tilde{X}_{t}-\tilde{V}_{t}\right)^2\Lambda_{1, t}^{\top} \right)\dd W_{t}\\
&+\left(2L_{t}\left(\tilde{X}_{t}-\tilde{V}_{t}\right) \frac{1}{S_{t}^{0}} \eta_{2, t}^{\top} +\left(\tilde{X}_{t}-\tilde{V}_{t}\right)^2\Lambda_{2, t}^{\top} \right)\dd B_{t}.
\end{align*}
Let us now introduce a localizing sequence $(\tau_k)_{k\in\mathbb{N}}$ with $\tau_k\nearrow \infty$ as $k\to\infty$. We fix $t\in[0,T]$, integrate on both sides and take expectations:
\begin{align}
\label{eq:stoppedIto}
\begin{aligned}
\mathbb{E}&\left[L_{t\wedge \tau_k}\left(\tilde{X}_{t\wedge \tau_k}-\tilde{V}_{t\wedge \tau_k}\right)^{2}\right]=L_{0}\left(\tilde{X}_{0}-\tilde{V}_{0}\right)^{2}\\
&+\mathbb{E}\left[\int_0^{t\wedge \tau_k}\left\{L_{s}\left[\operatorname{diag}(\tilde{S}_{s}) \xi_{s}-\left(\sigma_{s} \sigma_{s}^{\top}\right)^{-1}\left(\left[\left(\mu_{s}-r_{s}\mathbbm{1}\right)^{\top}+\frac{\sigma_{s} \Lambda_{1, s}}{L_{s}}\right]\left(\tilde{X}_{s}-\tilde{V}_{s}\right)+\sigma_{s} \frac{\eta_{1, s}}{S_{s}^{0}}\right)\right]^{\top}\right.\right.\\
&\quad\quad\quad\quad\times\sigma_s\sigma_s^\top\left[\operatorname{diag}(\tilde{S}_{s}) \xi_{s}-\left(\sigma_{s} \sigma_{s}^{\top}\right)^{-1}\left(\left[\left(\mu_{s}-r_{s}\mathbbm{1}\right)^{\top}+\frac{\sigma_{s} \Lambda_{1, s}}{L_{s}}\right]\left(\tilde{X}_{s}-\tilde{V}_{s}\right)+\sigma_{s} \frac{\eta_{1, s}}{S_{s}^{0}}\right)\right]\\
&\quad\quad\quad\quad\left.\left.+\frac{L_s}{\left(S_{s}^{0}\right)^{2}} \eta_{2, s}^{\top} \eta_{2, s}\right\}\dd s\right].
\end{aligned}
\end{align}
We let $k\to\infty$ and set $t=T$. The right-hand side converges by monotone convergence. From to the definition of admissibility that we currently assume, the left-hand side converges to $\mathbb{E}\left[\left(\tilde{X}_{T}-\tilde{V}_{T}\right)^{2}\right]$ since we also have $L_{T}=1$.
The expression is minimized by choosing
\begin{align*}
\xi^\star_{t}=\operatorname{diag}(\tilde{S}_{t})^{-1}\left(\sigma_{t} \sigma_{t}^{\top}\right)^{-1}\left(\left[\left(\mu_{t}-r_t\mathbbm{1}\right)^{\top}+\frac{\sigma_{t} \Lambda_{1, t}}{L_{t}}\right]\left(\tilde{X}_{t}-\tilde{V}_{t}\right)+\sigma_{t} \frac{\eta_{1, t}}{S_{t}^{0}}\right).
\end{align*}
It remains to show that $\xi^\star$ is admissible. For this aim, we substitute $\xi^\star$ into \eqref{eq:stoppedIto} and we obtain
\begin{align*}
\mathbb{E}&\left[L_{t\wedge \tau_k}\left(\tilde{X}_{t\wedge \tau_k}-\tilde{V}_{t\wedge \tau_k}\right)^{2}\right]=L_{0}\left(\tilde{X}_{0}-\tilde{V}_{0}\right)^{2}+\mathbb{E}\left[\int_0^{t\wedge \tau_k}\frac{L_s}{\left(S_{s}^{0}\right)^{2}} \eta_{2, s}^{\top} \eta_{2, s}\dd s\right].
\end{align*}
We know that $L$ is bounded in $(0,1]$. Also, since the interest rate $r$ is bounded, so is the denominator $\left(S_{\cdot}^{0}\right)^{2}$. We also know that $\eta_2\in L^2_{\FF}([0,T];\R^m)$, hence, there exists some $\mathcal{C}\in\mathbb{R}_+$ such that for any  $t$ and $k$ one has
\begin{align*}
\mathbb{E}\left[\int_0^{t\wedge \tau_k}\frac{L_s}{\left(S_{s}^{0}\right)^{2}} \eta_{2, s}^{\top} \eta_{2, s}\dd s\right]\leq \mathcal{C}\, \mathbb{E}\left[\int_0^{T} \eta_{2, s}^{\top} \eta_{2, s}\dd s\right]<\infty,
\end{align*}
hence the quantity $L_{t\wedge \tau_k}\left(\tilde{X}_{t\wedge \tau_k}-\tilde{V}_{t\wedge \tau_k}\right)^{2}$ is uniformly integrable for any localizing sequence $(\tau_k)_{k\in\mathbb{N}}$, and $\xi^\star$ is admissible.

\section{The Heston-PDE solver}
\label{sec:PDEsolver}
We briefly present in this section the numerical scheme that we use for solving the two Heston PDEs arising, respectively, in the mean-variance hedging, Section \ref{meanvariance1dim}, and in the local risk hedging, Section \ref{localrisk1dim}, for the one-dimensional case. The idea is to first discretize in the spatial variables, and then to apply a splitting scheme of the Alternating Direction Implicit (ADI) type.  In particular, we follow the approach proposed by \cite{in2010adi}, with the only difference that we allow the coefficient of the derivative w.r.t. the price to depend on time. Since the rest of the algorithm is unchanged, we shall omit details such as the meshes and the finite difference schemes applied. These can be found in \cite{in2010adi}.

Let $f(T-t, y, s)$ denote the price of a European-type option with maturity $T$ and payoff function $g:\R_+\times \R_+\to \R_+$, whose underlying has price at time $t$ equal to $s$ and volatility equal to $y$. Under the Heston model \eqref{HestonModel}, the price $f$ satisfies the parabolic PDE
\begin{equation}
	\label{PDE}
	\begin{cases}
		-f_1 + \frac{1}{2}y\left( \sigma^2 f_{22} +  2\rho\sigma s f_{23} + s^2f_{33}\right) + {\kappa}_t\left({\theta}_t-y\right) f_2 =0,\\
		f(0, y, s) = g(y,s),\\
		f(t, y, 0) = 0, \qquad \mbox{for } 0\le t \le T,
	\end{cases}
\end{equation}
where $f_i := \frac{\partial f}{\partial x_i}$ and $f_{ij} := \frac{\partial ^2 f}{\partial x_i\partial x_j}$, with $i,j \in \{1,2,3\}$. Here ${\kappa}:\R_+ \to \R_+$ and ${\theta}:\R_+ \to \R_+$. We point out that equation \eqref{PDE} reduces to the classical Heston PDE when ${\kappa}_t \equiv \kappa$ and ${\theta}_t \equiv \theta$ for every $0\le t \le T$, with $\kappa$ and $\theta$ as in equation \eqref{HestonModel}. Moreover, for the mean-variance hedging in Section \ref{meanvariance1dim} one has that
\begin{equation}
	\begin{cases}
		{\kappa}_t = \kappa + \rho\sigma\mu -\chi_1(t)\sigma^2(1-\rho^2), \\
		{\theta}_t = \frac{\kappa \theta}{{\kappa}_t},
	\end{cases}
\end{equation}
and for the local risk hedging in Section \ref{localrisk1dim} one has that
\begin{equation}
	\begin{cases}
		{\kappa}_t  = \kappa + \rho\sigma\mu, \\
		{\theta}_t  = \frac{\kappa \theta}{{\kappa}_t}.
	\end{cases}
\end{equation}

For numerical purposes, we restrict the spatial domain to a bounded set $[0, \mathbb{S}]\times [0, \mathbb{Y}]$. In particular, we shall consider $\mathbb{S} = 8K$, with $K$ the strike price, and $\mathbb{Y} = 5$. This is followed by two additional conditions at $s=\mathbb{S}$ and at $y = \mathbb{Y}$:
\begin{equation}
	\label{cond}
	f_3(t, y, \mathbb{S}) = 1 \quad \mbox{ and } \quad 	f(t, \mathbb{Y}, s) = s, \qquad \mbox{for } 0\le t \le T.
\end{equation}
We then apply non-uniform meshes in both the $s$- and $y$-direction, so that relatively many mesh points lie in the neighborhood of $s=K$ and $y=0$, respectively (see \cite[Section 2.2]{in2010adi} for details). We denote with $M_s\ge 1$ and $M_y\ge 1$ the number of points in the $s$- and $y$-direction. 

The finite difference discretization yields an initial value problem for a large system of ordinary differential equation of the form
\begin{equation}
	\label{ODE}
	\begin{cases}
		F'_t = A_tF_t + b_t, \qquad \mbox{for } 0\le t \le T,\\
		F_0 = f_0.
	\end{cases}
\end{equation}
Here $A_t \in \R^{\mathfrak{m}\times \mathfrak{m}}$, $b_t\in\R^\mathfrak{m}$ and $f_0\in\R^\mathfrak{m}$, for $\mathfrak{m} := M_sM_y$. In particular, $f_0$ is obtained from the initial condition in \eqref{PDE} and the vector function $b_t$ depends on the boundary conditions in \eqref{PDE} and \eqref{cond}. We solve \eqref{ODE} by applying the modified Craig-Sneyd scheme as described in \cite[Section 2.3]{in2010adi}, which we adapt so to allow the time-dependent coefficients ${\kappa}$ and ${\theta}$.

\end{document}